\def\llncs{0}
\def\fullpage{1}
\def\anonymous{0}
\def\draft{1}
\def\submission{0}
\def\shorter{0}

\ifnum\submission=1
\def\llncs{1}
\def\draft{0}
\def\anonymous{1}
\def\fullpage{0}
\fi

\ifnum\llncs=1
    \documentclass{llncs}
    \ifnum\fullpage=1
    \usepackage{fullpage}
    \fi
\else
    \documentclass[letterpaper,hmargin=1.05in,vmargin=1.05in]{article}
    \ifnum\fullpage=1
    \usepackage{fullpage}
    \fi
    \usepackage{microtype}
\fi

\usepackage{booktabs}
\usepackage{esvect}
\usepackage{verbatim}
\usepackage{authblk}
\usepackage{graphicx} 
\usepackage{physics} 
\usepackage{xcolor} 
\usepackage{amsmath}
\allowdisplaybreaks[4] 
\usepackage{amssymb}
\usepackage{mathtools}

 \usepackage{amsthm}
 \usepackage{thmtools}
\usepackage{enumitem} 
\usepackage{dsfont}
\usepackage[colorlinks=true,linkcolor=magenta,citecolor=blue,pagebackref=true,hypertexnames=false,pdftex,pdfpagelabels,bookmarks,hyperindex,hyperfigures]{hyperref}
\usepackage{dashbox}
\usepackage{fancybox,framed}
\usepackage[skins]{tcolorbox}
\usepackage[compatibility=false]{subcaption}
\usepackage{breakcites}
\usepackage{comment} 
\let\subparagraph\paragraph
\usepackage{titlesec}
\titleformat*{\paragraph}{\normalsize\bfseries}

\usepackage[capitalise,nameinlink]{cleveref}

\ifnum\draft=1
	\newcommand{\minki}[1]{\textcolor{blue}{$\langle\langle$Minki: #1$\rangle\rangle$}}
\else
	\newcommand{\minki}[1]{}
\fi

\ifnum\draft=1
	\newcommand{\qipeng}[1]{\textcolor{red}{Qipeng: #1}}
\else
	\newcommand{\qipeng}[1]{}
\fi

\ifnum\draft=1
	\newcommand{\sung}[1]{\textcolor{brown}{Sunghyuk: #1}}
\else
	\newcommand{\sung}[1]{}
\fi

\ifnum\draft=1
	\newcommand{\revised}[1]{\textcolor{purple}{#1}}
\else
	\newcommand{\revised}[1]{#1}
\fi

\ifnum\llncs=0
	\newtheorem{theorem}{Theorem}[section]
	\newtheorem{lemma}[theorem]{Lemma}
	\newtheorem{corollary}[theorem]{Corollary}
	
	\newtheorem{definition}[theorem]{Definition}
	
	\newtheorem{fact}[theorem]{Fact}
	\newtheorem{claim}[theorem]{Claim}

	\theoremstyle{remark}
	\newtheorem{remark}[theorem]{Remark}

\else
	\spnewtheorem{fact}{Fact}{\bfseries}{\itshape}
	\spnewtheorem{algorithm}{Algorithm}{\bfseries}{\rmfamily}

	\spnewtheorem{claim}{Claim}{\bfseries}{\itshape}
\fi
\crefname{appendix}{Appendix}{Appendices}
\Crefname{appendix}{Appendix}{Appendices}
\crefname{fact}{Fact}{Facts}
\Crefname{fact}{Fact}{Facts}

\renewcommand{\epsilon}{\varepsilon}

\ifnum\draft=1
	
\else
	
\fi

\newcommand{\perm}{\mathsf{perm}}

\newcommand{\E}{\mathbb{E}}
\newcommand{\Z}{\mathbb{Z}}
\newcommand{\one}{\mathbf{1}}
\newcommand{\Samp}{\mathsf{Sample}}
\newcommand{\Ver}{\mathsf{Verify}}
\newcommand{\Query}{\mathsf{Query}}
\newcommand{\sample}{\overset{\$}{\leftarrow}}

\newcommand{\unif}{\mathsf{unif}}
\newcommand{\img}{\mathsf{img}}
\renewcommand{\norm}[1]{\left\lVert #1\right\rVert}
\newcommand{\PermInv}{\mathsf{PermInv}}

\renewcommand{\eqref}[1]{\cref{#1}}

\newcommand{\papersubtitle}{%
    \ifnum\draft=1
        Tight(er) Bounds for Decision Games and Salting%
    \else
        Tight(er) Bounds for Decision Games and Salting%
    \fi
}
\ifnum\llncs=1
    \title{An Operator-Norm Approach\texorpdfstring{\\}{ }to Security with Quantum Advice}
    \subtitle{\papersubtitle}
\else
    \title{An Operator-Norm Approach\texorpdfstring{\\}{ }to Security with Quantum Advice\\[0.4em]
        {\large\papersubtitle}}
\fi
\date{}

\ifnum\llncs=1
    \ifnum\anonymous=1
        \author{}
        \institute{}
    \else
        \author{
        Minki Hhan\inst{1} 
        \and
        Sunghyuk Jo\inst{1}
        \and
        Qipeng Liu\inst{2}
        }
        \institute{
        KAIST
        \\\email{minkihhan@kaist.ac.kr}
        \and
        \\UCSD
        \\\emaail{}
        }
    \fi
\else
    \author[1]{Minki Hhan}
    \author[1]{Sunghyuk Jo}
    \author[2]{Qipeng Liu}
    \affil[1]{{\small KAIST, Daejeon, Korea}
    \authorcr{\small minkihhan@kaist.ac.kr, \quad josung9921@kaist.ac.kr}}
    \affil[2]{{\small UC San Diego, USA}
    \authorcr{\small qipengliu0@gmail.com}}

\fi

\begin{document}

\maketitle

\begin{abstract}
    Non-uniform security allows an adversary to receive bounded advice about an oracle before attempting a fresh challenge. This captures the most realistic attacks and has already been studied extensively in prior work. In this work, we introduce an operator-norm approach for non-uniform security in the quantum random oracle and random permutation models. This new approach yields a unified reduction for both search success probability and distinguishing advantage. Previously, the reduction only worked with success probability even in the decision games, yielding a worse bound.

    Our framework enables tight bounds for Yao's box, both with and without salting, and improved bounds for pseudorandom generators. We also prove an optimal generic salting theorem for decision games. By defining a property of a game, which separates the contributions of the existing queries in the offline stage and subsequent online queries, we obtain stronger bounds for specific salted constructions. These include salted permutation inversion, tight up to logarithmic factors, and salted random-function inversion, tight up to logarithmic factors and the gap already present in classical function inversion.
\end{abstract}

\ifnum\fullpage=1
\vfill
\paragraph{AI use disclosure.} ChatGPT-5.6 Sol Ultra discovered a slightly worse PRG bound in a one-shot query. 
Subsequent conversations improved and generalized the results to the current bounds. 
The original proof was written in the language of operator inequalities.
The authors interpreted the intermediate terms as the operational characterization of the game advantages and observed their relations. The current manuscript is entirely written by the authors based on this understanding, with assistance from ChatGPT-6 Astra in improving its presentation.
\fi
\ifnum\llncs=1
\fi


\ifnum\submission=1
\else
	\clearpage
	\newpage
	\setcounter{tocdepth}{2}
	\tableofcontents
	\newpage
\fi

\section{Introduction}

In modern cryptographic designs, hash functions and block ciphers form a handy toolkit for building practical and secure constructions. They are in turn analyzed in idealized models such as the random oracle model (ROM) \cite{CCS:BelRog93}, where we treat a hash function as a uniformly random oracle. 
Public permutations and block ciphers are similarly idealized in the random permutation model (RPM) and ideal cipher model (ICM).
While these idealized models provide rigorous security analysis, the usefulness of these guarantees depends on how faithfully the model captures the resources available to an adversary.

Preprocessing is a particularly important resource. In practice, a relatively small collection of hash functions and block ciphers is standardized, while they are reused across many applications and over long periods of time. An adversary may therefore invest substantial computation in analyzing these standardized primitives in advance, and later make use of this preprocessing to accelerate the attack. This motivates non-uniform security models \cite{STOC:Yao90,C:Unruh07,C:CorDodGuo18} in which an online adversary receives \emph{advice} of bounded size, occasionally called \emph{auxiliary input}, that depends on the underlying ideal primitives.
Further work established time-space tradeoffs for function inversion and pseudorandom generators~\cite{C:DeTreTul10,EC:DodGuoKat17}, followed by tight bounds for permutation-based PRGs and matching attacks for Yao's box~\cite{SODA:GGKL21}. Other work characterized how salting limits the advantage from preprocessing~\cite{C:DonLiuWu24}.

Quantum computing adds a further resource dimension. On one hand, the quantum adversary making queries to the oracles in superposition becomes a new threat \cite{AC:BDFLSZ11}, introducing the quantum random oracle model (QROM), which allows quantum queries to the random oracles---QRPM and QICM are defined analogously.
On the other hand, massive classical preprocessing can begin before a large-scale quantum computer becomes available, introducing a possibility of quantum attacks with classical preprocessing \cite{NABT14,AC:HhaXagYam19}. 
Subsequent work established stronger quantum time-space tradeoffs for function inversion and developed a framework with applications to Yao's box and salted cryptography~\cite{CGLQ20}.

The bit-fixing (BF) model is one of the methods to deal with the preprocessing lower bound.
In this model, the adversary presamples or fixes parts of the oracles instead of receiving the advice \cite{C:Unruh07,EC:CDGS18}. The preprocessing security is then bounded by the bit-fixing security. 
Guo, Li, Liu, and Zhang developed a quantum version of bit-fixing \cite{TCC:GLLZ21} based on the implicit ideas in \cite{CGLQ20}. These ideas lead to the tight non-uniform security of function inversion and more in the QROM given \emph{classical} advice.

More generally, a quantum adversary may store a quantum state as the outcome of preprocessing, which we refer to as quantum auxiliary input (QAI). 
Early lower bounds for inversion with quantum advice were established for random permutations~\cite{AC:HhaXagYam19} and random functions~\cite{ITC:ChuLiaQia20}.
The most conservative model is therefore the quantum idealized model with quantum advice, for example QAI-QROM and QAI-QICM. The security in the presence of QAI is much more difficult to show. 
One difficulty is that quantum advice cannot generally be copied, and using it may disturb the state, complicating arguments based on reuse or rewinding. Beyond these proof-technical obstacles, quantum advice can be intrinsically more powerful than classical advice: exponential separations are known for certain games in the QROM~\cite{Liu23}, and a separation between $\mathsf{BQP/qpoly}$ and $\mathsf{BQP/poly}$ is known relative to classically accessible classical oracles~\cite{ITCS:LLPY24}. Related work on quantum versus classical proofs established a classical-oracle separation of $\mathsf{QMA}$ and $\mathsf{QCMA}$ under a quantum pseudorandomness conjecture~\cite{STOC:LiuMutYue25}, followed by an unconditional classical-oracle separation~\cite{BHNZ26}.

The reduction between QAI security and the quantum version of BF security has been established in~\cite{CGLQ20,Liu23}. There are two versions of these reductions with multiplicative loss and additive loss, which are used for the search problems and decision problems, respectively.
The multiplicative version is relatively tight and yields the tight OWF security of random oracles and permutations in the corresponding QAI models \cite{Liu23,ABC+26}.
More precisely, they obtain a security bound $O((ST + T^2)/N)$, where $N$ is the domain size of the function or permutation (assuming a codomain of size at least $N$ in the function case), $S$ is the number of advice qubits, and $T$ is the number of online quantum queries.

The tightness of the additive version (as well as the non-uniform security for many decision games) is unclear. 
For example, the proven PRG advantage is about $O((T^2/N)^{1/2} + (ST/N)^{1/3})$ while there are known attacks with the advantage $\Omega(T^2/N + (S/N)^{1/2})$ where the first term is due to quantum search and the second term is from \cite{C:DeTreTul10}.\footnote{The preprocessing inversion algorithm \cite{Hellman80} can be used to attack PRGs with the advantage $\min(ST/N , (S^2 T/ N^2)^{1/3})$. This does not change the discussion about the parameters here. This observation was suggested by ChatGPT 6 Astra.} Ignoring constant factors, obtaining an advantage bound on the order of $2^{-128}$ from the previous result requires $n\approx640$ when $S=T=2^{128}$, compared with $n\approx384$ for the corresponding OWF bound.



Another open question concerns preprocessing security with salting. The generic phenomenon called the salting-defeat-preprocessing \cite{EC:CDGS18} is shown in the QAI setting as well. However, for the size of the salting space $K$, these generic theorems incur a loss of $S/K$ for many search games \cite{C:DonLiuWu24} and $(ST/K)^{1/3}$ for all decision games \cite{Liu23}. In many security games, the desired security would be tighter; for example, the OWF security with an $O(ST/KN)$ term is known in AI-ROM \cite{EC:DodGuoKat17,EC:CDGS18}. A similar quantum result is only in the AI-QROM setting with loose bounds of $O(ST^2/KN)$ \cite{AC:HhaXagYam19}.
With the previous results, salting is known to defeat preprocessing only with a very lengthy salt, e.g., $K=2^{640}$ against $S=T=2^{128}$ to have $2^{-128}$ PRG advantage.

\subsection{Our Results}

In this paper, we systematically study the idealized model in the QAI setting. 
We present a new route towards understanding non-uniform security with quantum advice in the idealized models.
We present a new method based on the operator norm of a quantum algorithm, and connect it to the bit-fixing model~\cite{CGLQ20,TCC:GLLZ21}.

To state our results, let $S$ be the number of advice qubits, $T$ denote the number of online queries, and $P$ the number of offline queries in the bit-fixing model. Let $N$ and $M$ denote the oracle domain and range sizes, respectively.

Our first result is the unified reduction between QAI security and the quantum bit-fixing security. Previously, the bound connected the success probabilities in both the QAI setting and the BF setting. Our new result directly applies to the advantages in both cases---for decision games, the advantage is defined as the absolute difference between the success probability and $1/2$.
This improves the additive bound of \cite[Theorem 6.1]{Liu23} and \cite[Theorem 1]{TCC:GLLZ21}.

\begin{theorem}[QAI security to BF security, informal]
    Fix any game $G$.
    Let $\delta(S,T)$ be the maximum advantage that a quantum algorithm with $S$-qubit auxiliary input and $T$ quantum queries can achieve.  Let $\nu(P, T)$ be the maximum advantage in the $P$-BF model with at most $T$ quantum queries. Then we have
    \begin{align*}
        \delta(S,T) \le \sqrt 2 \cdot \nu(Sr,T), 
    \end{align*}
    where $r:=2(T_{\Samp}+T+T_{\Ver})$. Here $T_{\Samp}$ and $T_{\Ver}$ are the numbers of queries made by the sampler and verifier in $G$, respectively.
\end{theorem}

This theorem alone improves the QAI-QROM security of Yao's box problem from
$\widetilde O(\sqrt[19]{S^5T/N})$ in \cite{CGLQ20}\footnote{This can be further improved to $O(\sqrt[3]{ST/N})$ by \cite[Theorem~6.1]{Liu23}, but was never explicitly written up.} to $O(\sqrt{ST/N})$. 
For PRG security, we prove the bound $O(\sqrt{ST/N} + T^2/N)$.
The previous best bound was $O(\sqrt[3]{ST/N} + \sqrt{T^2/N})$ \cite{Liu23}.
The improvement of the first term is due to the improved reduction. For the second term, 
we observe that the PRG security in the quantum $P$-bit-fixing model can be improved from $O(\sqrt{(P+T^2)/N})$ to $O(\sqrt{P/N}+T^2/N)$ using the polynomial method \cite{FOCS:BBCMW98}. All comparisons can be found in \Cref{tab:unsalted-bounds}.




\setlength{\tabcolsep}{12pt}
\renewcommand{\arraystretch}{2}
\begin{table}[ht]
    \centering
    \begin{tabular}{lcc}
        \toprule
        \noalign{\vskip -6pt}
        \rule{0pt}{3ex} Decision Game & Known Bound & Ours\\[-3pt]
        \midrule
        PRG & $\displaystyle O\left(\sqrt[3]{\frac{ST}{N}}+\sqrt{\frac{T^{2}}{N}}\right)$ & $\displaystyle O\left(\sqrt{\frac{ST}{N}}+\frac{T^{2}}{N}\right)$ \\
        Yao's box & $\displaystyle \widetilde O\left(\sqrt[19]{\frac{S^{5}T}{N}}\right)$ or $\displaystyle O\left(\sqrt[3]{\frac{ST}{N}}\right)$ & $\displaystyle O\left(\sqrt{\frac{ST}{N}}\right)$ \\
        \bottomrule
    \end{tabular}
    \vspace{5pt}
    \caption{We assume $T\ge 1$ in this table. PRG bound is from \cite[Theorem~1.2]{Liu23}. Yao's box bound is from \cite[Theorem~1.4]{CGLQ20} and can be improved by \cite[Theorem~6.1]{Liu23}. 
    }
    \label{tab:unsalted-bounds}
\end{table}

\paragraph{Salting.}
Throughout the presentation below, we use the same conventions for $S$, $T$, $N$, $M$, and $P$. We let $K$ denote the size of the salt space.

Let $\nu_{\mathcal G}(T)$ and $\epsilon_{\mathcal G}(T)$ denote the base game's, i.e., without salting, uniform search success and decision advantage, respectively. Existing works~\cite{CGLQ20,TCC:GLLZ21} have established bit-fixing bounds for salted games. In the QROM, they show $O(\nu_{\mathcal G}(T)+P/K)$ for search games and $O(\epsilon_{\mathcal G}(T)+\sqrt{P/K})$ for decision games.
Using these bit-fixing bounds, \cite[Theorem~7.5]{Liu23} obtained QAI bounds of $O(\nu_{\mathcal G}(T)+Sr/K)$ for search games and $O(\epsilon_{\mathcal G}(T)+(Sr/K)^{1/3})$ for decision games, where $r:=2(T_{\Samp}+T+T_{\Ver})$.

Our generic salting theorem (\cref{thm:generic-public-salting}) recovers the search bound up to constant factors and improves the salt-dependent term in the decision bound from $(Sr/K)^{1/3}$ to $\sqrt{Sr/K}$. We obtain this refinement by combining the salting analysis of~\cite{CGLQ20} with our novel operator-norm-based approach.
\begin{theorem}[Salting for decision games, informal]
    \label{thm:salting-decision-informal}
    Fix a decision game $\mathcal G$ in the QROM or quantum random permutation model (QRPM) with uniform $T$-query advantage at most $\epsilon_{\mathcal G}(T)$.
    Let $\mathcal G_{\mathsf S}$ be its salted version with $K$ independent oracle instances and a uniformly random challenge salt.
    Then the maximum advantage of a quantum adversary with $S$-qubit auxiliary input and at most $T$ online queries satisfies
    \[
        \epsilon_{\mathcal G_{\mathsf S}}^{\mathsf{QAI}}(S,T)
        = O\!\left(\epsilon_{\mathcal G}(T)+\sqrt{\frac{Sr}{K}}\right),
    \]
    where $r:=2(T_{\Samp}+T+T_{\Ver})$, and $T_{\Samp}$ and $T_{\Ver}$ are the numbers of queries made by the sampler and verifier, respectively.
\end{theorem}

The square-root term is optimal for general decision games. Indeed, the one-bit Yao game ($N=1$) has zero uniform advantage, and its salted version is Yao's box on $K$ bits. The attack of~\cite[Theorem~1.3]{SODA:GGKL21} achieves $\Omega(\sqrt{S(T+1)/K})$, matching this term since $r=2(T+1)$.

As already noticed by \cite{EC:CDGS18}, the salting-defeat-preprocessing gives a worse bound than the direct application of bit-fixing to the salted games. Thus, we individually work on different games and obtain refined bounds for various games, in both the QROM and QRPM. All comparisons and results can be found in \Cref{tab:salted-bounds}.

Our unsalted and salted Yao's-box bounds are tight up to constant factors, matching the classical attack of~\cite[Theorem~1.3]{SODA:GGKL21} with domain sizes $N$ and $KN$, respectively. Our salted permutation-inversion bound is tight up to logarithmic factors, by choosing between Grover search and Hellman's attack with preprocessing distributed across salts~\cite{Hellman80,ABC+26}.

For salted random-function OWFs and $T\ge1$, our bound matches the attacks (See e.g., \cite[Table~1]{AC:HhaXagYam19}) up to logarithmic factors whenever $ST^2\le K\min\{N,M\}$. It also matches Grover search whenever $S\le KT$, since the uniform term then dominates.

{
\setlength{\tabcolsep}{12pt}
\renewcommand{\arraystretch}{2}
\begin{table}[ht]
    \centering
    \begin{tabular}{lcc}
        \toprule
        \noalign{\vskip -6pt}
        \rule{0pt}{3ex} Salted Game & Known Bound & Ours\\[-3pt]
        \midrule
        PRG & $\displaystyle O\left(\sqrt[3]{\frac{ST}{K}}+\sqrt{\frac{T^{2}}{N}}\right)$ & $\displaystyle O\left(\sqrt{\frac{ST}{KN}}+\frac{T^{2}}{N}\right)$ \\
        Yao's box & $\displaystyle \widetilde O\left(\sqrt[19]{\frac{S^{5}T}{KN}}\right)$ & $\displaystyle O\left(\sqrt{\frac{ST}{KN}}\right)$ \\
        OWF & $\displaystyle O\left(\frac{ST}{K}+\frac{T^2}{\min\{N,M\}}\right)$ & $\displaystyle O\left(\frac{ST}{K\min\{N,M\}}+\frac{T^2}{\min\{N,M\}}\right)$ \\
        PermInv & $\displaystyle \widetilde O\left(\left(\frac{ST^2}{KN}+\frac{T^2}{N}\right)^{1/3}\right)$ & $\displaystyle O\left(\frac{ST}{KN}+\frac{T^2}{N}\right)$ \\
        \bottomrule
    \end{tabular}
    \vspace{5pt}
    \caption{We assume $T\ge 1$ in this table. PRG and OWF bounds are from \cite[Theorem~1.3]{Liu23}. Yao's box bound follows from \cite[Theorem~1.4]{CGLQ20} after substituting $N\mapsto KN$. PermInv bound is from \cite[Theorem~6]{AC:HhaXagYam19}.}
    \label{tab:salted-bounds}
\end{table}
}

\subsection{Technical Overview}


In this section, we will go over the main ideas used to obtain these results.

\paragraph{Purified game operators and database hierarchies.}
We start by expressing the adversary's advantage through an operator. Let $\mathcal G$ be a game which can be either a search or a decision game with oracle space $\mathcal F$ which can be either a set of functions or a set of permutations. Fix a quantum adversary $\mathcal A$ who receives $S$-qubit advice and makes $T$ quantum queries to the oracle $H\in\mathcal F$. Let $\{\rho_H\}_{H\in \mathcal F}$ be the family of advice states, each of which depends on $H$ and is stored in advice register $A$. Without loss of generality, we can purify the algorithm and assume that $\mathcal A$ applies a unitary $V^H$ on the joint registers $AW$. Let $\Pi$ be the final projection to the acceptance space. Define the Hermitian operator
    \[
        E^H:=(I_A\otimes\langle0|_W)(V^H)^\dagger\Pi V^H(I_A\otimes|0\rangle_W).
    \]
    Then, for the search game $\mathcal G$, the winning probability is $\E_{H}[\Tr(\rho_HE^H)]$. 
    
    For a decision game, we define two projections $\Pi_{\mathsf T }$ to the accepted space and $\Pi_{\mathsf F} = I - \Pi_{\mathsf T}$ to the rejected space. As in the search game case, for each projection, we define the Hermitian operator
    \[
        \begin{aligned}
        E^H_{\mathsf T}:=(I_A\otimes\langle0|_W)(V^H)^\dagger\Pi_{\mathsf T} V^H(I_A\otimes|0\rangle_W) \\
        E^H_{\mathsf F}:=(I_A\otimes\langle0|_W)(V^H)^\dagger\Pi_{\mathsf F} V^H(I_A\otimes|0\rangle_W).
        \end{aligned}
    \]
    Finally, we define $E^H:=(E^H_{\mathsf T}-E^H_{\mathsf F})/2$. Then, its decision advantage is $|\E_{H}[\Tr(\rho_HE^H)]|$. 

    Instead of bounding $|\E_{H}[\Tr(\rho_HE^H)]|$ directly, we bound it through the high-order moment of its operator norm as follows: for any integer $k \geq 1$,
    \[ 
        \begin{aligned}
        |\E_{H}[\Tr(\rho_HE^H)]|
        &\le \E_{H}[\lVert E^H\rVert_{\mathrm{op}}] \\
        &\le \left(\E_H[\lVert E^H\rVert_{\mathrm{op}}^{2k}]\right)^{1/(2k)}
        \le \left(\E_H[\Tr|E^H|^{2k}]\right)^{1/(2k)}.
        \end{aligned}
    \]
    The second inequality follows from the monotonicity of $L^p$ norms, and the last from $\lVert E^H\rVert_{\mathrm{op}}^{2k}\le\Tr|E^H|^{2k}$.
    
    The next step is to bound this spectral moment by expressing it as the norm of a vector.
    We start our argument from the purified oracle state $|\Omega_{\mathcal F}\rangle\sim\sum_{H\in\mathcal F}|H\rangle$ in oracle register $O$. We purify the family of operators $\{E^H\}_{H\in\mathcal F}$ into  $M_{\mathcal E}$: 
    \[
    M_{\mathcal E}:=\sum_{H\in\mathcal F}|H\rangle \langle H|\otimes E^H.
    \]
    Now, assume the algorithm starts with a maximally mixed advice state, but purified as a maximally entangled state using an additional reference register $R$, 
    \[
    |\Phi\rangle := \frac{1}{\sqrt{d}}\sum_{a\in [d]} |a\rangle_A |a\rangle_R.
    \]

    Since $E^H$ only applies to the advice register, we have $d=\dim \mathcal H_{A} \le 2^S$. Let $|v_{0}\rangle:=|\Omega_{\mathcal F}\rangle |\Phi\rangle$ be the initial state, with the purified oracle and the maximally entangled state. 
    
    We can define $|v_{j}\rangle := M_{\mathcal E} |v_{j-1}\rangle = M_{\mathcal E}^{j}|v_{0}\rangle$ for $j\ge 1$. Therefore, we have for all integers $k \geq 1$,
    \begin{equation}
        \lVert |v_k\rangle\rVert_2^2
        =\E_H\left[\frac1d\sum_{a=1}^d\langle a|(E^H)^{2k}|a\rangle\right]
        =\E_H\left[\frac1d\Tr|E^H|^{2k}\right].
        \label{eq:overview_A}
    \end{equation}
    Thus, from the inequality above, the advantage is at most $d^{1/(2k)}\lVert |v_k\rangle\rVert_2^{1/k}$.

    \begin{remark}[Comparison to the existing high-moment method.]
    The high-moment viewpoint already appears in~\cite{CGLQ20,Liu23}. In the alternating-measurement framework of~\cite{Liu23}, they use two types of rounds: odd rounds test success, while even rounds project the randomness register back onto its initial uniform superposition. This argument relates the success probability of a non-uniform algorithm to the success probability in the so-called alternating-measurement game.  
    Our method works directly on the operator that represents the advantage, rather than the success probability, while the alternating measurement framework does not extend to the case when the operator $E^H$ is no longer positive semi-definite (in the decision games).
    \end{remark}

\paragraph{Bit-fixing security as a database-size-restricted operator norm.}

    Next, we will analyze the state $|v_k\rangle$.
    Since it corresponds to a repeated application of $M_{\mathcal E}$ on the initial state $|v_0\rangle$, we will use the compressed-oracle technique~\cite{Zha19} for analyzing the joint state over the algorithm and oracle. It has already been applied to analyze non-uniform security in the QROM~\cite{CGLQ20}, as well as in the QRPM using its permutation counterpart~\cite{ABC+26}. We will mostly focus on the function case in the overview.

    Intuitively, the compressed oracle uses the database size to capture the amount of information an algorithm has learned about the random function so far; the database size can only be increased by one per query. Thus, we can formally define the space of joint states having database size at most $P$ by
    \[
        \mathcal L_{P}:=\operatorname{span}\{|x,y\rangle_{A,R}|\phi_D\rangle_O  : |D|\le P\}.
    \]
    Let $\Pi_{\le P}$ be the orthogonal projection onto $\mathcal L_{P}$.
    
    We observe that $M_{\mathcal E}$ increases the database size by at most $r=2(T+T_{\Samp}+T_{\Ver})$ where the $2$ factor comes from both the applications of $(V^H)^\dagger$ and $V^H$. 
    Since \(|v_j\rangle\in\mathcal L_{jr}\subseteq\mathcal L_{(j+1)r}\) and \(|v_{j+1}\rangle\in\mathcal L_{(j+1)r}\), for \(0\le j<k\),
    \[
        |v_{j+1}\rangle
        =\Pi_{\le (j+1)r}M_{\mathcal E}\Pi_{\le jr}|v_j\rangle
        =\Pi_{\le (j+1)r}M_{\mathcal E}\Pi_{\le (j+1)r}|v_j\rangle.
    \]
    Hence, 
    \[
        \lVert |v_{j+1}\rangle\rVert_2
        \le||\Pi_{\le (j+1)r} M_{\mathcal E}\Pi_{\le (j+1)r}||_{op}\lVert |v_j\rangle\rVert_2.
    \]
    Here, we call terms like $||\Pi_{\le (j+1)r} M_{\mathcal E}\Pi_{\le (j+1)r}||_{op}$ the database-size-restricted operator norm.
    
    Iterating gives, for any $k\ge 1$,
    \begin{equation}
        \lVert |v_k\rangle\rVert_2
        \le\prod_{j=0}^{k-1}||\Pi_{\le (j+1)r} M_{\mathcal E}\Pi_{\le (j+1)r}||_{op}
        \le||\Pi_{\le kr} M_{\mathcal E}\Pi_{\le kr}||_{op}^k.
        \label{eq:overview_B}
    \end{equation}
    Combining \eqref{eq:overview_A} and \eqref{eq:overview_B} with $k=S$ gives
    \[  
        |\E_{H}[\Tr(\rho_HE^H)]| \le \sqrt{2}||\Pi_{\le Sr} M_{\mathcal E}\Pi_{\le Sr}||_{op}.
    \]

    The value $||\Pi_{\le P} M_{\mathcal E}\Pi_{\le P}||_{op}$ is in fact related to the $(P,T)$-bit-fixing model, as implicitly proposed by~\cite{CGLQ20} and defined in~\cite{TCC:GLLZ21}. The bit-fixing model has two phases: the offline phase and the online phase. In the offline phase, the adversary can make $P$ quantum queries to the oracle $H\in \mathcal F$ and do any post-selection to produce a state $|\psi\rangle$. In the online phase, a challenge is received, and the adversary makes at most $T$ quantum queries to solve the challenge starting from $|\psi\rangle$. Let $A$ be the shared register between the offline phase and the online phase. Then $|\langle \psi|M_{\mathcal E}|\psi\rangle|$ is nothing but the online phase's advantage on the game $\mathcal G$ starting from $|\psi\rangle$. 
    
    Surprisingly, we show that the offline phase can produce any unit vector $|\psi\rangle \in \mathcal L_{P}$, and therefore the $(P,T)$-bit-fixing model has exactly the same security as the database-size-restricted operator norm\footnote{The other half---any state produced in the $(P, T)$-bit-fixing model is in $\mathcal{L}_P$---holds immediately by definition, and was proved in~\cite{CGLQ20}.}! 
    Take an arbitrary unit vector $|\psi\rangle \in \mathcal L_{P}$. Then by definition, it has an expression $|\psi\rangle \sim \sum_{|D|\le P}|\phi_{D}\rangle_O |z_{D}\rangle_A$. We prepare the state
    \[
        |\Psi_P\rangle \sim \sum_{|D|\le P} |\phi_D\rangle_O |0\rangle_A |D\rangle_R,
    \]
    which can be done by applying $P$ quantum queries on each support of $D$, on the state $|\Psi_0\rangle \sim |\Omega_{\mathcal F}\rangle_O \otimes |0\rangle_A \otimes \sum_{|D|\le P} |D\rangle_R $ where $R$ is a helper register.  
    Setting the operator acting on register $RA$,
    \[
        F:=\sum_{|D|\le P}|0\rangle \langle D|_R \otimes |z_{D}\rangle \langle 0|_{A},
    \]
    we have $(I_O \otimes F)|\Psi_{P}\rangle\sim |\psi\rangle_{OA}|0\rangle_R$. Taking a constant $c\in \mathbb{R}_{> 0}$ sufficiently small, we obtain $K:=cF$ with $||K||_{op}\le 1$. From the singular value decomposition of $K$, we observe that $K$ can be implemented by postselection. 
    
    Taking a supremum over $|\psi\rangle \in \mathcal L_{P}$ and online algorithms, we get
    \[
        \mathsf{Adv}_{\mathcal G}^{\mathsf {BF}}(P,T)=\sup_{\mathcal A}||\Pi_{\le P} M_{\mathcal E}\Pi_{\le P}||_{op}.
    \]
    Therefore,
    \[
        |\E_{H}[\Tr(\rho_HE^H)]| \le \sqrt{2}\mathsf{Adv}_{\mathcal G}^{\mathsf {BF}}(Sr,T).
    \]

\paragraph{Salting and the database contribution.}
    From the discussion above, it suffices to bound $||\Pi_{\le P} M_{\mathcal E}\Pi_{\le P}||_{op}$. 
    For many decision games like PRGs and Yao's box, it is already studied in~\cite{CGLQ20}; thus we can immediately get an improved bound for these decision games using the new operator-norm-based reduction. 
    We also improve a term in the PRG security using the polynomial method \cite{FOCS:BBCMW98}. We omit this detail in this overview.
    Next, we will mostly focus on salting.
    
    As $M_{\mathcal E}$ is Hermitian, it suffices to bound $|\langle \psi|M_{\mathcal E}|\psi\rangle|$ for any $|\psi\rangle \in \mathcal L_{P}$. 
    Let $\Pi_{=j}$ be the orthogonal projection to the space of joint states having database size exactly $j$.
    We define the database-size observable $\Lambda :=\sum_{j} j\Pi_{= j}$ to measure the expected size of the database with which $|\psi\rangle$ is entangled. Its value is $\langle \psi | \Lambda |\psi\rangle$ and it is less than or equal to $P$ for $|\psi\rangle \in \mathcal L_{P}$.  
    We bound $|\langle \psi|M_{\mathcal E}|\psi\rangle|$ with $f(\langle \psi | \Lambda |\psi\rangle)$ for some real-valued function $f$. 
    
    We define the salted database hierarchy. Let $[K]$ be the salt space. Every salted oracle is of the form $(H_{1},\cdots,H_{K})$. Let $O_{s}$ be the oracle register for the salt $s\in [K]$. Then the joint oracle register $O$ is $O_{1}\otimes \cdots \otimes O_{K}$. 
    Each salt $s$ has its own database $D_{s}$. The salted database is $D=(D_{1},\cdots,D_{K})$ and its size is $\sum_{s\in [K]} |D_{s}|$.
    Let $\Lambda_s$ be the database-size observable on $O_s$, acting as the identity on all other registers. Define $\mathcal L_P^{(K)}$ as the subspace spanned by eigenvectors of $\sum_{s\in[K]}\Lambda_s$ with eigenvalues at most $P$. Thus every unit vector $|\psi\rangle\in\mathcal L_P^{(K)}$ satisfies
    \[
        \sum_{s\in[K]}\langle\psi|\Lambda_s|\psi\rangle\le P.
    \]
    We define the operator $M_{\mathcal E_{s}}$ in a similar way for a fixed salt $s$, acting on register $O_{s}AZ$ where $Z$ includes the oracle registers $O_{t}$ for $t\ne s$. For a uniformly random salt, the salted game operator is
    \[
        M_{\mathcal E^{(K)}}:=\frac{1}{K}\sum_{s\in [K]}M_{\mathcal E_{s}}.
    \]
    Since each $M_{\mathcal E_{s}}$ raises the database size by at most $r$, the same is true for $M_{\mathcal E^{(K)}}$.

    We prove that the adversary's advantage in the salted game is bounded by the uniform security of the base game plus $Sr/K$ for a search game and $\sqrt{Sr/K}$ for a decision game. 
    For a fixed salt $s$, we may treat all other oracle registers other than $O_{s}$ as auxiliary workspace and absorb their queries into $H_{s}$-independent local unitaries. It follows that if a joint state has $|D_{s}|=0$, it can be considered a valid base-game uniform security experiment on the oracle $H_{s}$. Let $\Pi_{s,=0}$ be the projection to the space of joint states with $|D_{s}|=0$. 

    For any $|\psi\rangle\in \mathcal L_{P}^{(K)}$, we split it into $|\psi_{0}\rangle:=\Pi_{s,=0}|\psi\rangle$ and $|\psi_{1}\rangle:=(I-\Pi_{s,=0})|\psi\rangle$. Then $|\langle \psi_{0}|M_{\mathcal E_{\mathsf s}}|\psi_{0}\rangle |$ can be bounded by its uniform security, and $|\langle \psi_{1}|M_{\mathcal E_{\mathsf s}}|\psi_{1}\rangle |$ can be bounded by $\langle \psi|\Lambda_s|\psi\rangle$ as $I-\Pi_{s,=0}\preceq \Lambda_{s}$. 
    For a search game, we exploit the fact that $M_{\mathcal E_{\mathsf s}}^{1/2}$ exists and use the triangle inequality. Then, we get
    \[
        \langle \psi|M_{\mathcal{E}_s}|\psi\rangle
        \le \left(\sqrt{\langle \psi_0|M_{\mathcal{E}_s}|\psi_0\rangle } + \sqrt{\langle \psi_1|M_{\mathcal{E}_s}|\psi_1\rangle}\right)^2
        \le \left(\sqrt{\nu_{\mathcal{G}}(T)} + \sqrt{\langle \psi|\Lambda_{s}|\psi\rangle}\right)^2.
    \]
    Averaging it over $s\in [K]$ gives $\langle \psi|M_{\mathcal{E}^{(K)}}|\psi\rangle \le  \left(\sqrt{\nu_{\mathcal{G}}(T)} + \sqrt{P/K}\right)^2$.

    For a decision game, as $-I/2\preceq M_{\mathcal E_{\mathsf s}}\preceq I/2$, it need not admit a positive semi-definite square root, and thus a triangle inequality does not apply. We expand $\langle \psi|M_{\mathcal E_{\mathsf s}}|\psi\rangle$ in terms of $\langle \psi_i|M_{\mathcal E_{\mathsf s}}|\psi_j\rangle$ where $i,j\in \{0,1\}$ and bound each term.
    This results in the weaker exponent $1/2$. 
    
    These generic bounds can be loose for specific games, as already observed in the classical setting~\cite{EC:CDGS18}. For example, when a database of expected size $L$ contributes only $O(L/N)$, averaging over salts gives an $O(P/(KN))$ contribution, yielding an $O(Sr/(KN))$ term after our reduction. The following diffuseness condition captures such additional structure.
    We say a game is a $(\gamma,\Delta)-$diffuse game if there exist $\tau_{T}\ge 0$ and $\kappa \ge 0$ such that, for any $|\psi\rangle$,
    \[
        |\langle \psi|M_{\mathcal E}|\psi\rangle | \le \tau_{T}+\kappa\left(\frac{\langle \psi|\Lambda |\psi\rangle}{\Delta}\right)^{\gamma}.
    \]
    Then, in the salting case, the same thing applies: for any $|\psi\rangle$,
    \[ 
        |\langle \psi|M_{\mathcal E_s}|\psi\rangle |  \le \tau_{T}+\kappa\left(\frac{\langle \psi|\Lambda_s |\psi\rangle}{\Delta}\right)^{\gamma}.
    \]
    Averaging the above inequality over $s\in [K]$ proves that if a base game is a $(\gamma,\Delta)-$diffuse game, then its salt term is $\kappa(Sr/K\Delta)^{\gamma}$. We prove that many games satisfy this property, including OWF, PRG, and Yao's box; this naturally gives tighter or even tight non-uniform bounds with quantum advice, as in~\Cref{tab:unsalted-bounds,tab:salted-bounds}.

\ifnum\fullpage=0
\paragraph{AI use disclosure.} ChatGPT-5.6 Sol Ultra discovered a slightly worse PRG bound in a one-shot query. 
Subsequent conversations improved and generalized the results to the current bounds. 
The original proof was written in the language of operator inequalities.
The authors interpreted the intermediate terms as the operational characterization of the game advantages and observed their relations. The current manuscript is entirely written by the authors based on this understanding.
\fi
\section{Preliminaries}\label{sec: preliminaries}
We assume that the readers are familiar with the basics of quantum computation and information, and refer to~\cite{NC10} for standard background.
A positive operator-valued measure (POVM) on a register $A$ is a family
$\{E_z\}_z$ of positive semidefinite operators satisfying $\sum_z E_z=I_A$.
For a state $\rho$ on $A$, outcome $z$ occurs with probability
$\operatorname{Tr}(\rho E_z)$. We call $E_z$ the POVM element corresponding
to outcome $z$; in particular, $0\preceq E_z\preceq I_A$.

\subsection{Compressed Oracle and Function Database Hierarchy}
\label{sec:Random Functions: The Fourier-Degree Filtration}
In this section, we briefly review several key concepts from~\cite{Zha19}. We refer the reader to the original paper for further details.

In the quantum random oracle model (QROM), a query to a fixed oracle $H\in\mathcal F:=\{H:[N]\to[M]\}$ is coherently processed by the following unitary
\[
|x\rangle|u\rangle\longmapsto|x\rangle|u+H(x)\rangle.
\]
We work with the \emph{purified} oracles by default.
In this setting, the oracle register $O$ is initialized to
\[
|\Omega_{\mathcal F}\rangle_O:=\frac{1}{\sqrt{|\mathcal F|}}\sum_{H\in\mathcal F}|H\rangle_O
\]
and the queries to the oracle are implemented by the unitary
\[
|H\rangle|x\rangle|u\rangle\longmapsto|H\rangle|x\rangle|u+H(x)\rangle.
\]


Another equivalent and useful formulation is the phase oracle:
\[
    |H\rangle|x\rangle|u\rangle \mapsto \omega_M^{uH(x)}|H\rangle|x\rangle|u\rangle
\]
where $\omega_M:=e^{2\pi i/M}$ is the primitive $M$-th root of unity.
It is well known that the output distribution of any algorithm having access to the phase oracles is identical to the output distribution with the quantum random oracles by changing the response register to the Fourier basis.
Therefore, results for one oracle imply the same results for the other.
From now on, $\mathcal O$ denotes the phase-oracle query to the underlying oracle $H$.

\begin{lemma}[{\cite[Lemma~3]{Zha19}}]
    Let $\mathcal A$ be an (unbounded) quantum algorithm making queries to
    a standard oracle. Let $\mathcal B$ be the algorithm that is identical to $\mathcal A$, except it performs $\mathsf{QFT}$ and $\mathsf{QFT}^{\dagger}$
    before and after each query. Then the output distributions of $\mathcal A$ (given access to a standard oracle) and $\mathcal B$
    (given access to a phase oracle) are identical. Therefore, a quantum random oracle can be perfectly
    simulated as a phase oracle.
    \label{lem:std oracle and phase oracle}
\end{lemma}

The compressed oracle analysis shows that, after
\(P\) queries, the oracle state lies in the span of Fourier basis
states supported on at most \(P\) inputs. We formalize database hierarchies below, building on these approaches.
We follow the notation of \cite[Section~2.3]{CGLQ20}. 
We identify the set $\mathcal{F}:=\{H:[N]\rightarrow [M]\}$ with $\mathbb{Z}_{M}^{N}$
by the map $D \mapsto (D(1),...,D(N))$. 
For $D\in\mathbb{Z}_{M}^{N}$, define 
\[
    \langle D,H\rangle
    :=\sum_{x\in[N]}D(x)H(x)\pmod M,
    \qquad
    |\phi_D\rangle
    :=\frac{1}{M^{N/2}}\sum_{H\in\mathcal F}\omega_M^{\langle D,H\rangle}|H\rangle
\]
where the state $|\phi_D\rangle$ is called the Fourier (basis) state.
The set of Fourier states $\ket {\phi_D}$ forms an orthonormal basis for all $D \in \mathbb{Z}^N_M$.

Let $|D|$ denote the number of nonzero coordinates of $D$, which is often referred to as ``database size'' in the original~\cite{Zha19}.
For an oracle register $O$ and an integer $P\ge0$, define the space spanned by database states of size at most $P$ by
\[
    \mathcal{L}_P
    :=\operatorname{span}\{|\phi_D\rangle_O:|D|\le P\}
\]
Let $\Pi_{\le j}$ be the orthogonal projection onto $\mathcal L_{j}$. 
The function database hierarchy refers to the nested family
\[
\mathcal{L}_{0}\subseteq \mathcal{L}_{1}\subseteq \mathcal{L}_{2} \subseteq \cdots \subseteq \mathcal{L}_N
\]
where $\mathcal{L}_N$ is the full ambient space of the database.
A state has database size at most $P$ if it lies in $\mathcal L_P$.

The following fact shows that the database size of the states is consistent with the number of queries made by the algorithm. 
We include the proof in \cref{sec:Appendix A} for completeness.
\begin{fact}[Database Size Increase for Random Functions~\cite{Zha19}]
\label{fact:function-query-level}
Let $R$ contain the query registers and any auxiliary registers. For
every integer $P\ge0$, a phase-oracle query $\mathcal O$ on the register $OR$ satisfies
\[
    \mathcal O\bigl(\mathcal L_P\otimes\mathcal H_R\bigr)
    \subseteq\mathcal L_{P+1}\otimes\mathcal H_R.
\]
Thus one query increases the database size by at most one, even allowing
arbitrary entanglement between the oracle and the registers in $R$.
\end{fact}

We will use the subspaces of exact database size, which are defined by
$\mathcal D_0:=\mathcal{L}_0$ and, for $j\ge1$, $\mathcal D_j:=\mathcal{L}_j\cap\mathcal{L}_{j-1}^{\bot}$. Orthogonality of the Fourier states gives
\begin{align}
    \label{eqn: function exact size database}
    \mathcal D_j=\operatorname{span}\{|\phi_D\rangle_O:|D|=j\}.
\end{align}
We define the orthogonal projection $\Pi_{= j}$ onto $\mathcal{D}_j$ for $j\ge 0$. For $j\ge 1$, $\Pi_{= j}=\Pi_{\le j}-\Pi_{\le j-1}$ and for $j=0$, $\Pi_{=0}=\Pi_{\le 0}$. 
We define the associated database-size observable by
\[
    \Lambda=\sum_{j=1}^{N}j\Pi_{=j}.
\]
By the orthogonality of the Fourier states $|\phi_{D}\rangle$, $\Lambda$ can be rewritten as
\begin{align}
    \label{eqn: observable decomposition function}
    \Lambda=\sum_{D\in\Z_M^N}|D|\,|\phi_D\rangle\langle\phi_D|.
\end{align}
It shows that 
for any unit vector $|\psi\rangle\in \mathcal{L}_P$,  $\langle\psi|\Lambda|\psi\rangle\le P$.

\subsection{Permutation Analogue and Permutation Database Hierarchy}
\label{sec:Random Permutations: The Partial-Assignment Filtration}

In the quantum random permutation model (QRPM), the oracle is a uniformly
random permutation $H\in\mathcal F^{\perm}:=S_N$, where $S_N$ is the set of
permutations of $[N]$. 
We only consider the forward evaluation access: for a fixed $H$,
a query acts as
\[
    |x\rangle|u\rangle
    \longmapsto
    |x\rangle|u+H(x)\rangle
\]
where the addition in the response register
is modulo $N$. Thus, the query interface is the same as for random
functions, with the oracle distribution restricted to permutations.

We again use the purified oracle representation. Initialize the oracle
register $O$ to the uniform superposition of permutations
\[
    |\Omega_{\mathcal F^{\perm}}\rangle_O
    :=\frac{1}{\sqrt{N!}}\sum_{H\in S_N}|H\rangle_O,
\]
and implement each query by the controlled unitary
\[
    |H\rangle|x\rangle|u\rangle
    \longmapsto
    |H\rangle|x\rangle|u+H(x)\rangle.
\]

As in \cref{sec:Random Functions: The Fourier-Degree Filtration}, we use
the equivalent phase-oracle representation, obtained by changing the
response register to the Fourier basis:
\[
    \mathcal O:
    |H\rangle|x\rangle|u\rangle
    \longmapsto
    \omega_N^{uH(x)}|H\rangle|x\rangle|u\rangle.
\]

To track the oracle states generated by these queries, we use the database
subspaces of~\cite[Section~2.3]{ABC+26} which stem from \cite{Ros22}.
While their formalization requires the representation theory of symmetric groups, 
we only need a hierarchy structure of the permutation database. We introduce the relevant definitions and results in this section.

A permutation database in the QRPM is represented by 
a \emph{partial permutation assignment},
which is an injective function $\alpha:D\to [N]$ for a subset $D\subseteq [N]$.
The size of a partial permutation assignment (henceforth, a database) is $|\alpha|:=|D|$. 
For $H\in\mathcal F^{\perm}$, write $\alpha\subseteq H$ if $H(x)=\alpha(x)$ for every $x\in D$. 
Let $\operatorname{im}(\alpha)$ denote the image of $\alpha$. For any $x\notin D$ and $y\notin\operatorname{im}(\alpha)$, let $\alpha\cup \{x\mapsto y\}$ denote 
a larger database with domain $D\cup\{x\}$ that agrees with $\alpha$ on $D$ and maps $x$ to $y$.

Define the partial assignment state
\[
    |v_{\alpha}\rangle := \frac{1}{\sqrt{(N-|\alpha|)!}}\sum_{\alpha\subseteq H\in \mathcal{F}^{\perm}} |H\rangle.
\]
This is the uniform superposition over all permutations consistent with
the database $\alpha$. There are $(N-|\alpha|)!$ such permutations, which
explains the normalization. In particular, the empty database gives
$|v_{\varnothing}\rangle=|\Omega_{\mathcal F^{\perm}}\rangle$.

For an oracle register $O$ and an integer $P\ge0$, we define the space spanned by (permutation) database states of size at most $P$ by
\[
    \mathcal{L}_{P}^{\perm}:=\operatorname{span}\{|v_{\alpha}\rangle_O: |\alpha|\le P\}.
\]
Let $\Pi_{\le j}^{\perm}$ be the orthogonal projection onto $\mathcal{L}_j^{\perm}$. 
This gives the permutation database hierarchy
\[
\mathcal{L}_{0}^{\perm}\subseteq \mathcal{L}_{1}^{\perm}\subseteq \mathcal{L}_{2}^{\perm} \subseteq \cdots \subseteq \mathcal L_{N-1}^{\perm}
\]
where the last term is $\mathcal L_{N-1}^{\perm}$ because it is already the full ambient space.
Analogously, we say a state has database size at most $P$ if it lies in $\mathcal L_{P}^{\perm}$.

\begin{remark}
    In \cite{ABC+26}, they define the space spanned by (permutation) database states by $\mathcal{L}_{P}^{\perm}=\operatorname{span}\{|v_{\alpha}\rangle_O : |\alpha|= P\}$, i.e., the span of the database states of size exactly $P$. The two definitions are identical for $0\le P\le N$ because of the identity
    \[
        |v_{\alpha}\rangle=\frac{1}{\sqrt{N-j}}\sum_{y\not\in \operatorname{im}(\alpha)}|v_{\alpha \cup \{x\rightarrow y\}}\rangle,
    \]
    for any $|\alpha|=j<N$ and $x\not\in \operatorname{dom}(\alpha)$. If $P>N$, $\mathcal L_{P}^{\perm}$ is the full ambient space. 
\end{remark}

Similarly to \cref{fact:function-query-level}, the database size of the states is consistent with the number of queries made by the algorithm. 
See \cref{sec:Appendix A} for the proof.

\begin{fact}[Database Size Increase for Random Permutations~\cite{ABC+26}]
\label{fact:permutation-query-level}
Let $R$ contain the query registers and any auxiliary registers. For
every integer $P\ge0$, a phase-oracle query $\mathcal O$ on $OR$ satisfies
\[
    \mathcal O\bigl(\mathcal L_P^{\perm}\otimes\mathcal H_R\bigr)
    \subseteq\mathcal L_{P+1}^{\perm}\otimes\mathcal H_R.
\]
Thus one query increases the database size by at most one, allowing
arbitrary entanglement between the oracle and the registers in $R$.
\end{fact}



The subspaces of exact database size for permutations are defined analogously to the function case.
Define $\mathcal D_0^{\perm}:=\mathcal{L}_0^{\perm}$ and, for $j\ge1$, define $\mathcal D_j^{\perm}:=\mathcal{L}_j^{\perm}\cap(\mathcal{L}_{j-1}^{\perm})^\bot$. 
We define the orthogonal projection $\Pi_{= j}^{\perm}$ onto $\mathcal{D}_j^{\perm}$ for $j\ge 0$. For $j\ge 1$, $\Pi_{= j}^{\perm}:=\Pi_{\le j}^{\perm}-\Pi_{\le j-1}^{\perm}$ and for $j=0$, $\Pi_{=0}^{\perm}:=\Pi_{\le 0}^{\perm}$.
Note that
\begin{align}
    \label{eqn: permutation orthogonal decomposition}
    \sum_{j=0}^{N-1} \Pi_{=j}^{\perm} = \Pi_{\le N-1}^{\perm} = I, \qquad \Pi_{=i}^{\perm}\Pi_{=j}^{\perm} = \delta_{i,j} \Pi_{=i}^{\perm}
\end{align}
where $\delta_{i,j}=1$ if and only if $i=j$. This gives an orthogonal decomposition of the space of permutations.

The associated database-size observable is defined by
\[
    \Lambda^{\perm}:=\sum_{j=1}^{N-1}j \Pi_{=j}^{\perm}.
\]
We remark that the vectors $|v_\alpha\rangle$ are not mutually orthogonal in the permutation case.
Consequently,
the subspaces of exact database size and the database-size observable do not admit the simple basis description as in the function case (see \cref{eqn: function exact size database,eqn: observable decomposition function}).

For permutations, database size is measured by the orthogonal hierarchy sectors using \cref{eqn: permutation orthogonal decomposition} instead; $|v_\alpha\rangle$ may have components at several sizes, all at most $|\alpha|$.
For any unit vector $|\psi\rangle\in\mathcal{L}_P^{\perm}$ with $P\le N-1$, decompose $|\psi\rangle$ orthogonally as $|\psi\rangle=\sum_{j=0}^P|\psi_j\rangle$, where $|\psi_j\rangle:=\Pi_{= j}^{\perm}|\psi\rangle\in\mathcal D_j^{\perm}$. Then
\[
    \langle\psi|\Lambda^{\perm}|\psi\rangle
    =\sum_{j=1}^P j\cdot \lVert\psi_j\rVert_2^2\le P.
\]

\subsection{Uniform One-Point Reprogramming}
We introduce one-point reprogramming. Let $H\in\mathcal F$ and $h\in[M]$. For each $x\in[N]$, define the
one-point reprogrammed function
\[
    H_{x\to h}(w)
    :=
    \begin{cases}
        h, & w=x,\\
        H(w), & w\ne x.
    \end{cases}
    \label{eq:one-point-reprogrammed-function}
\]

The following lemma bounds the change in acceptance probability when
the reprogrammed point is uniformly random and independent of the
circuit. 
We present a direct proof
of the stated version using the polynomial method \cite{FOCS:BBCMW98} in \Cref{sec:random_reprogramming}. 

\begin{lemma}[Uniform one-point reprogramming]
    \label{lem:uniform-one-point-reprogramming}
    Fix an arbitrary oracle $H:[N]\to\mathbb{Z}_M$,
    $h\in\mathbb{Z}_M$, and a quantum circuit making at most $T$
    oracle queries. Suppose that its initial state, inter-query
    operations, and final measurement are independent of
    $x\sample[N]$.
    Write $a(x)$ for the circuit's acceptance
    probability with oracle $H_{x\to h}$ and $a(\varnothing)$ for
    its acceptance probability with oracle $H$. Then
    \[
        \left|
            \E_{x\sample[N]}[a(x)]-a(\varnothing)
        \right|
        \le
        \frac{4T^2}{N}.
    \]
\end{lemma}

\section{Connecting Non-uniform Security with Bit-Fixing Security}
\label{sec:The degree-Restricted Moment Method for Oracle-Dependent Advice}

This section gives a new reduction from the non-uniform security to the bit-fixing security. 
Our reduction gives a better bound for the decision games, improving upon \cite{CGLQ20,Liu23}.
To this end, we provide an operational view of the security games. 

\subsection{Game, QROM and QRPM}

A game $\mathcal G=(\mathcal C)$, where $\mathcal C=(\mathsf{Sample},\mathsf{Query},\mathsf{Verify})$, is an interactive protocol between a challenger and an adversary $\mathcal A$. 
Before beginning the game, we sample a uniformly random function $H:[N]\to[M]$ in the QROM or a uniformly random permutation $H\in S_N$ in the QRPM. Throughout this section, $\mathcal F$ denotes the corresponding oracle space, and $\{\mathcal L_P\}_{P\ge0}$ denotes its database hierarchy from \cref{sec: preliminaries}.

The game $\mathcal G$ proceeds as follows.
The challenger runs $\mathsf{Sample}^H$ and gives the resulting challenge $y$ to $\mathcal A$. The adversary accesses $H$ through $\mathsf{Query}^H$, which is $H$ itself in most games except for Yao's box problem.
Each query increases the database size by at most one.
The challenger then runs $\mathsf{Verify}^H$ and stores the resulting bit in the verification register $\mathsf V$.

For any joint register $R$ of the algorithm and the challenger (which is disjoint from the oracle register $O$), we extend the definition of the bounded-size database spaces, associated projections, and the database-size observable by
\[
\mathcal L_P(R):=\mathcal L_P\otimes\mathcal H_R,\qquad
\Pi_{\le P}(R):=\Pi_{\le P}\otimes I_R,\qquad
\Lambda(R):=\Lambda\otimes I_R.
\]
We use the same convention for exact size spaces and their projectors and implicitly extend operators by appending the identity on untouched registers. We omit $R$ when it is clear from the context.

\subsection{Bit-Fixing, Quantum Auxiliary-Input, and Advantage}

\paragraph{Bit-Fixing Model.}
We use the quantum bit-fixing framework of~\cite{TCC:GLLZ21}.
\begin{definition}[Bit-Fixing Model]
    A $(P,T)$-bit-fixing adversary $\mathcal A$ for a game $\mathcal G=(\mathcal C)$, where $\mathcal C=(\mathsf{Sample},\mathsf{Query},\mathsf{Verify})$, consists of the following two stages.
    \begin{itemize}[label=--, leftmargin=*]
        \item \textbf{Offline phase.}
    \begin{enumerate}[leftmargin=*, topsep=0pt, itemsep=0pt]
        \item A uniformly random oracle $H\in\mathcal F$ is sampled;
        
        \item The offline algorithm makes at most $P$ quantum queries to $H$,
        and outputs a bit $b$ such that $\Pr[b=0]>0$;
        
        \item If $b \neq 0$, the process restarts from Step~1.
    \end{enumerate}
    
        \item \textbf{Online phase.}
    \begin{enumerate}[leftmargin=*, topsep=0pt, itemsep=0pt]
        \item The challenger runs \(\mathsf{Sample}^H\), gives the resulting
        challenge \(y\) to the online algorithm, and retains
        the state required for the verification;
        
        \item On input \(y\), the online algorithm makes at most
        \(T\) additional quantum queries through \(\mathsf{Query}^H\) and outputs
        \(\mathsf{ans}\);

        \item The challenger runs \(\mathsf{Verify}^H\).
    \end{enumerate}
    \end{itemize}
\end{definition}

The purified oracle version of the bit-fixing model is naturally defined by starting with the initial state
\[
|\phi_0\rangle:=|\Omega_{\mathcal F}\rangle_O|0\rangle_{AZB},
\]
where $O$ denotes the oracle register,
$A,Z$ denote the adversary's workspace and auxiliary register, and $B$ denotes the flag qubit register, respectively.
The final step of the offline algorithm can be understood as a postselection. 
The offline algorithm outputs the postselected state $|\psi_{0}\rangle_{OAZB}$, and it becomes the input of the online algorithm. The online algorithm can access register $A$ but not registers $ZB$. Let $W$ denote the joint register of all other online workspace registers. 
We refer the reader to \cite[Section 2.2]{ABC+26} for a more detailed discussion.

%

\paragraph{Quantum Auxiliary-Input Model.}

\begin{definition}[Quantum Auxiliary-Input Model]
    An $(S,T)$-quantum-auxiliary-input adversary $\mathcal A$ for a game $\mathcal G=(\mathcal C)$, where $\mathcal C=(\mathsf{Sample},\mathsf{Query},\mathsf{Verify})$, consists of the following single stage.
    \begin{enumerate}[leftmargin=*, topsep=0pt, itemsep=0pt]
        \item  A uniformly random oracle $H\in\mathcal F$ is sampled;
        
        \item $\mathcal A$ receives an $S$-qubit auxiliary-input state $\rho_H$ in register $A$, which may depend on $H$;

        \item The challenger runs \(\mathsf{Sample}^H\), gives the resulting
          challenge \(y\) to \(\mathcal A\), and retains the state
          required for the verification;

        \item $\mathcal A$ makes at most $T$
        quantum queries to $\Query^H$, and outputs an answer $\mathsf{ans}$;

        \item The challenger runs \(\mathsf{Verify}^H\).
    \end{enumerate}
\end{definition}

With respect to the purified oracles, we can view the input to the online algorithm as the mixed state, up to normalization,
\[
\sum_{H\in \mathcal F} \ketbra{H} \otimes \rho_H.
\]


\paragraph{Adversary's Advantage.}

We define the adversary's advantage according to whether the game is a search or decision game and whether the adversary is in the bit-fixing or quantum auxiliary-input model.
For a search game $\mathcal G$, the adversary's advantage is defined to be the supremum of the winning probability
\[
    \delta^{\mathsf{QAI}}_{\mathcal G}(S,T):=\sup_{\mathcal A}\Pr[\mathcal A\text{ wins}], \qquad \delta^{\mathsf{BF}}_{\mathcal G}(P,T):=\sup_{\mathcal A}\Pr[\mathcal A\text{ wins}]
\]
where the first supremum is over the $(S,T)$-auxiliary-input adversaries, and the second is over the $(P,T)$-bit-fixing adversaries.

For a decision game $\mathcal G$, the adversary's advantage is the supremum of the absolute bias of its winning probability from $1/2$:
\[
    \epsilon_{\mathcal G}^{\mathsf{QAI}}(S,T):=\sup_{\mathcal A}\left|\Pr[\mathcal A\text{ wins}]-\frac{1}{2}\right|, \qquad \epsilon^{\mathsf{BF}}_{\mathcal G}(P,T):=\sup_{\mathcal A}\left|\Pr[\mathcal A\text{ wins}]-\frac{1}{2}\right|
\]
where the ranges of the supremum are the same as above.

When a statement applies to either type of game, we simply write
\[
    \mathsf{Adv}^{\mathsf{QAI}}_{\mathcal G}(S,T),\qquad \mathsf{Adv}_{\mathcal G}^{\mathsf{BF}}(P,T)
\]

\subsection{Game Operators and Their Purified Representation}


Recall that $A$ denotes the workspace of the algorithm. In the bit-fixing model, this register may contain some state after the postselection. In the quantum auxiliary input model, this register includes an $S$-qubit advice state depending on the oracle.

\begin{theorem}
Fix a game $\mathcal G=(\mathcal C)$ and the algorithm $\mathcal A$ used after the challenge is generated. There exists a game-operator family $\mathcal E:=\{E^H\}_{H\in\mathcal F}$ on register $A$, induced by $(\mathcal G,\mathcal A)$, such that every $E^H$ is Hermitian and $\|E^H\|_{\mathrm{op}}\le1$; moreover, $0\preceq E^H\preceq I$ if $\mathcal G$ is a search game. Define the corresponding purified game operator on $OA$ by
\[
M_{\mathcal E}:=\sum_{H\in\mathcal F}|H\rangle\langle H|_O\otimes E^H.
\]
For a quantum auxiliary-input adversary using $\mathcal A$ with advice family $\{\rho_H\}_{H\in\mathcal F}$, the advantage of this adversary is
\[
|\mathbb E_H[\operatorname{Tr}(\rho_H E^H)]|.
\]
For a bit-fixing adversary whose online algorithm is $\mathcal A$, let $|\psi_0\rangle_{OAZB}$ be the joint state produced by its offline phase, including the algorithm's register as well as oracle registers. The advantage of this adversary is
\[
|\langle\psi_0|(M_{\mathcal E}\otimes I_{ZB})|\psi_0\rangle|.
\]
\end{theorem}

\begin{proof}
Let $W$ collect all registers used in the game $\mathcal G$ by the sampler, the adversary $\mathcal A$, and the verifier other than $A$; these registers are initialized to $|0\rangle_W$.
By purifying the computation and deferring measurements, we may assume that, for each fixed $H\in\mathcal F$, the sampler, $\mathcal A$, and the verifier are jointly implemented by a unitary $V^H$.

\paragraph{Search Game.}
Let $\Pi:=|1\rangle \langle 1|_{\mathsf V}\otimes I$ be the projection onto the accepting subspace acting on the verify register $\mathsf {V}$. 
Conditioned on $H$, the winning probability of $\mathcal A$ is
\[
    \Pr[\mathsf{win}\mid H]
    =\Tr\!\left(\Pi V^H(\rho_H\otimes|0\rangle\langle0|_W)(V^H)^\dagger\right).
\]
Let $E^H$ be the operator defined by
\[
    E^H
    :=(I_A\otimes\langle0|_W)(V^H)^\dagger\Pi V^H(I_A\otimes|0\rangle_W).
\]
This is the POVM element on the initial advice register $A$ corresponding
to the winning outcome, so that
$\Pr[\mathsf{win}\mid H]=\operatorname{Tr}(\rho_H E^H)$.

\paragraph{Decision Game.}

Let $\Pi_{\mathsf{T}}:=|1\rangle \langle 1|_{\mathsf V}\otimes I$ and $\Pi_{\mathsf{F}}:=|0\rangle \langle 0|_{\mathsf V}\otimes I$ be the final projectors
corresponding to a true and a false decision, respectively.
Define the POVM elements on $A$ corresponding to these two outcomes:
\[
    E_{\mathsf{T}}^H
    :=
    (I_A\otimes\langle0|_W)
    (V^H)^\dagger
    \Pi_{\mathsf{T}}
    V^H
    (I_A\otimes|0\rangle_W),
\]
\[
    E_{\mathsf{F}}^H
    :=
    (I_A\otimes\langle0|_W)
    (V^H)^\dagger
    \Pi_{\mathsf{F}}
    V^H
    (I_A\otimes|0\rangle_W),
\]
and define the decision-game operator as half their difference
\[
    E^H
    :=
    \frac{1}{2}(E_{\mathsf{T}}^H-E_{\mathsf{F}}^H).
\]

We use the same notation $E^H$ in both cases, although the construction differs. In the search case, $0\preceq E^H\preceq I$; in either case, $E^H$ is Hermitian and $\|E^H\|_{\mathrm{op}}\le1$.
The fixed-oracle identities above hold for any input state on $A$.
In the quantum auxiliary-input model, this state is $\rho_H$,
and averaging over uniformly random $H$ gives the claimed formula.
In the bit-fixing model, the online phase starts from the normalized
postselected state $|\psi_0\rangle_{OAZB}$ produced by the offline phase.
The online algorithm receives register $A$ from the offline phase,
while $ZB$ remain untouched.
The same operators $E^H$ apply, since their construction depends only
on the online experiment, not on how its input state was prepared.
Since the joint experiment is controlled by $O$ and acts trivially on $ZB$, its pulled-back operator on $OAZB$ is $M_{\mathcal E}\otimes I_{ZB}$, which gives the bit-fixing formula.
\end{proof}

\subsection{QAI Security from Database-Size-Restricted Operator Bounds}

The following theorem bounds quantum auxiliary-input advantage by
bit-fixing advantage. We first characterize bit-fixing security by database-size-restricted norms of the purified game operators, and then prove the reduction using a moment argument.

\begin{theorem} \label{thm:QAI bounded by BF}
    Let $\mathcal G=(\mathcal C)$ be a game with $\mathcal C=(\Samp,\Query,\Ver)$. Let $T_{\Samp},T,T_{\Ver}$ be the numbers of queries that the sampler, adversary, and verifier make, respectively. Then,
    \[
        \mathsf{Adv}^{\mathsf{QAI}}_{\mathcal G}(S,T) \le \sqrt{2}\mathsf{Adv}^{\mathsf{BF}}_{\mathcal G}(2S(T_{\Samp}+T+T_{\Ver}),T).
    \]
\end{theorem}
\Cref{thm:QAI bounded by BF} strengthens the results of~\cite{Liu23} for decision games; it establishes a previously unknown bound, which eventually enables improvement for many decision games, including PRG and Yao's box.
We defer the proof to the end of the section, after establishing the following theorem.

\begin{theorem}[{Size-Restricted Operator Norm} and Bit-Fixing Model]
\label{thm:degree-restricted-moment}
\label{thm:bf-operator-characterization}
Fix a game $\mathcal G=(\mathcal C)$ and integers $P,T\ge0$.
For each online algorithm $\mathcal A$ making at most $T$ queries through
$\Query^H$, let $\mathcal E_{\mathcal A}$ be its induced game-operator
family on its input register $A$. Then
\begin{equation}
    \mathsf{Adv}^{\mathsf{BF}}_{\mathcal G}(P,T)
    =\sup_{\mathcal A}
    \bigl\|\Pi_{\le P}M_{\mathcal E_{\mathcal A}}\Pi_{\le P}\bigr\|_{\mathrm{op}},
    \label{eq:bf-operator-characterization}
\end{equation}
where the supremum ranges over all such online algorithms and their
input-register choices. For each fixed online algorithm, its database-size-restricted
operator norm equals the optimal bit-fixing advantage over offline
phases making at most $P$ queries.
\end{theorem}

The key idea is that every joint state of the oracle and algorithm registers produced by a $(P,T)$-bit-fixing offline phase is supported on $\operatorname{im}(\Pi_{\le P}\otimes I)$, and conversely every state with this support can be prepared using at most $P$ queries and postselection.
\ifnum\shorter=1
The proof of this theorem is deferred to \cref{subsec: degree-restricted-proof}.
\else

\begin{proof}[Proof of \cref{thm:degree-restricted-moment}]
Fix an online algorithm $\mathcal A$, and write
$\mathcal E:=\mathcal E_{\mathcal A}$. Any offline phase making at most
$P$ queries produces, after postselection, a normalized state
$|\psi_0\rangle_{OAZB}\in\mathcal L_P(AZB)$: the initial oracle state
has size zero; each query increases the size by at most one
by~\cref{fact:function-query-level,fact:permutation-query-level},
and postselection on the algorithm's registers preserves membership in
this subspace. Its reduced state
$\sigma_{OA}:=\operatorname{Tr}_{ZB}|\psi_0\rangle\langle\psi_0|$
therefore satisfies $\sigma_{OA}=\Pi_{\le P}\sigma_{OA}\Pi_{\le P}$.
The bit-fixing advantage is consequently bounded by
\[
    \bigl|\operatorname{Tr}(\sigma_{OA}M_{\mathcal E})\bigr|
    =\bigl|\operatorname{Tr}(\sigma_{OA}\Pi_{\le P}M_{\mathcal E}\Pi_{\le P})\bigr|
    \le\bigl\|\Pi_{\le P}M_{\mathcal E}\Pi_{\le P}\bigr\|_{\mathrm{op}}.
\]

Conversely, we show that every unit vector
$|\psi\rangle_{OA}\in\mathcal L_P(A)$ can be prepared using at most
$P$ queries and postselection with positive probability. We give the
construction for both database hierarchies.
\paragraph{Case 1: Function.}
    Let $C_{\mathrm F}:=\sum_{j=0}^{P}\binom Nj(M-1)^j$ be the number of all possible databases of size at most $P$.
    Adjoining the register $R$ to the purified oracle state $|\Omega_{\mathcal F}\rangle = |\phi_{0}\rangle$, prepare
    \[
    |\Psi^{\mathrm F}_0\rangle_{OR}:=\frac{1}{\sqrt{C_{\mathrm F}}}\sum_{\substack{D\in\mathbb Z_M^N\\|D|\le P}}|\phi_0\rangle_O|D\rangle_R.
    \]
    For each $D$, write $\operatorname{supp}(D)=\{x_1<\cdots<x_k\}$, where $k=|D|$. For $1\le i\le k$, coherently compute $(x_i,D(x_i))$ from $D$, apply the phase oracle, and uncompute the query registers. If $k<P$, each of the remaining queries is made at a fixed point with phase label $0$. Since
    \[
    \mathcal O\bigl(|\phi_{D'}\rangle_O|x,u\rangle\bigr)=|\phi_{D'+u e_x}\rangle_O|x,u\rangle,
    \]
    the resulting state is
    \[
    |\Psi^{\mathrm F}_P\rangle_{OR}=\frac{1}{\sqrt{C_{\mathrm F}}}\sum_{\substack{D\in\mathbb Z_M^N\\|D|\le P}}|\phi_D\rangle_O|D\rangle_R.
    \]

\paragraph{Case 2: Permutation.}
    First suppose that $P\le N-1$. Prepare
    \[
    |\Psi^{\mathrm{perm}}_0\rangle_{OR}:=\frac{1}{\sqrt{\binom NP\,N!}}\sum_{X\in\binom{[N]}P}\sum_{H\in\mathcal F^{\mathrm{perm}}}|H\rangle_O|X,0^P\rangle_R.
    \]
    For $X=\{x_1<\cdots<x_P\}$, query $x_1,\ldots,x_P$ in the standard oracle form, implemented using one phase-oracle query per point by \cref{lem:std oracle and phase oracle}, and store their images in $R$. This gives
    \[
    |\Psi^{\mathrm{perm}}_P\rangle_{OR}=\frac{1}{\sqrt{\binom NP\,N!}}\sum_{X\in\binom{[N]}P}\sum_{H\in\mathcal F^{\mathrm{perm}}}|H\rangle_O|H|_X\rangle_R.
    \]
    Since $\sum_{H\supseteq\alpha}|H\rangle_O=\sqrt{(N-P)!}\,|v_\alpha\rangle_O$, we have
    \[
    |\Psi^{\mathrm{perm}}_P\rangle_{OR}=\sqrt{\frac{(N-P)!}{\binom NP\,N!}}\sum_{\substack{\alpha\text{ database}\\|\alpha|=P}}|v_\alpha\rangle_O|\alpha\rangle_R.
    \]
    By \cref{sec:Random Permutations: The Partial-Assignment Filtration}, the vectors $|v_\alpha\rangle$ with $|\alpha|=P$ span $\mathcal L_P^{\mathrm{perm}}$. 
    For $P>N-1$, use the construction with $P=N-1$, since $\mathcal L_P^{\mathrm{perm}}=\mathcal L_{N-1}^{\mathrm{perm}}$.
    
\paragraph{Constructing the unitary.} 
In either case, the resulting state has the form
    \[
    |\Psi_P\rangle_{OR}=\gamma\sum_a|w_a\rangle_O|a\rangle_R,
    \]
    where $\gamma>0$ is the corresponding coefficient above, and $(a,|w_a\rangle)$ denotes $(D,|\phi_D\rangle)$ in the function case and $(\alpha,|v_\alpha\rangle)$ in the permutation case. Since the vectors $|w_a\rangle$ span $\mathcal L_P$, choose vectors $|z_a\rangle_A$ such that
    \[
    |\psi\rangle_{OA}=\sum_a|w_a\rangle_O|z_a\rangle_A.
    \]
    Define an operator $F$ on $RA$ by
    \[
    F:=\sum_a|0\rangle\langle a|_R\otimes|z_a\rangle\langle0|_A.
    \]
    Then
    \[
    (I_O\otimes F)(|\Psi_P\rangle_{OR}|0\rangle_A)=\gamma|\psi\rangle_{OA}|0\rangle_R.
    \]

    Take $c:=1/(1+||F||_{op})>0$ and set $K:=cF$. Then $||K||_{op}\le 1$. Let $K=U\Sigma V^{\dagger}$ be its singular value decomposition, and choose an orthonormal basis $\{|j\rangle\}$ of $RA$ such that $\Sigma|j\rangle=s_j|j\rangle$, where $0\le s_j\le1$. We implement $K$ by postselection.
    
    Let $B$ be a flag qubit. For an arbitrary state $|r\rangle_{RA}$, start with $|r\rangle_{RA}|0\rangle_B$. Apply $V^{\dagger}$:
    \[
        |r\rangle_{RA}|0\rangle_B\mapsto \sum_{j}r_j |j\rangle|0\rangle_{B},\qquad r_{j}:=\langle j|V^{\dagger}|r\rangle.
    \]
    Apply the rotation:
    \[
        |j\rangle|0\rangle_{B} \mapsto |j\rangle (s_{j}|0\rangle_B+\sqrt{1-s_{j}^2}|1\rangle_B).
    \]
    After applying $U$, the resulting state is
    \[
    \begin{aligned}
    &\sum_j r_jU|j\rangle
      \left(s_j|0\rangle_B+\sqrt{1-s_j^2}|1\rangle_B\right)\\
    &\quad=
    U\Sigma V^\dagger|r\rangle|0\rangle_B
    +
    U\sqrt{I-\Sigma^2}V^\dagger|r\rangle|1\rangle_B\\
    &\quad=
    K|r\rangle|0\rangle_B+L|r\rangle|1\rangle_B
    \end{aligned}
    \]
    where $L:=U\sqrt{I-\Sigma^2}V^\dagger$. Let this unitary be $\mathcal U$.
    It follows that
    \[
    \begin{aligned}
    &(I_O\otimes\mathcal U)
    \bigl(|\Psi_P\rangle_{OR}|0\rangle_A|0\rangle_B\bigr)\\
    &\quad=
    c\gamma
    |\psi\rangle_{OA}|0\rangle_R|0\rangle_B+
    (I_O\otimes L)
    \bigl(|\Psi_P\rangle_{OR}|0\rangle_A\bigr)|1\rangle_B.
    \end{aligned}
    \]
    Thus, postselecting on $B=0$, which occurs with probability $c^2\gamma^2>0$, leaves $R$ in the fixed state $|0\rangle_R$ and the remaining registers in the state $|\psi\rangle_{OA}$.

For each unit vector $|\psi\rangle\in\mathcal L_P(A)$, the construction
gives a valid bit-fixing adversary with online algorithm $\mathcal A$
and advantage $|\langle\psi|M_{\mathcal E}|\psi\rangle|$.
The maximum of this expression is the database-size-restricted operator norm, attained
by an eigenvector whose eigenvalue has the largest absolute value.
Together with the upper bound, this proves the characterization for each
fixed online algorithm. Taking the supremum over online algorithms
proves~\eqref{eq:bf-operator-characterization}.
\end{proof}
\fi

\begin{proof}[Proof of \cref{thm:QAI bounded by BF}]
For $S=0$, the conclusion is obvious as there is no advice and no offline part in both games.

We assume $S\ge1$.
Fix an arbitrary $(S,T)$-quantum auxiliary-input adversary with online
algorithm $\mathcal A$ and advice family $\{\rho_H\}_{H\in\mathcal F}$.
Let $\mathcal E$ be its induced game-operator family, put
$d:=\dim\mathcal H_A\le2^S$, and set
$r:=2(T_{\Samp}+T+T_{\Ver})$.
For the fixed online query budget $T$, set
\[
    \beta(P):=\mathsf{Adv}^{\mathsf{BF}}_{\mathcal G}(P,T)
    \qquad\text{for every integer }P\ge0.
\]
This function is nondecreasing because increasing the offline query
budget enlarges the set of allowed adversaries.
By~\cref{thm:bf-operator-characterization}, the purified game operator of the fixed online algorithm $\mathcal A$ satisfies
\begin{equation}
    \|\Pi_{\le P}M_{\mathcal E}\Pi_{\le P}\|_{\mathrm{op}}
    \le\beta(P)\qquad\text{for every integer }P\ge0.
    \label{eq:restricted-degree-bound}
\end{equation}

We consider the iterative application of $M_{\mathcal E}$ on the state $|v_0\rangle:=|\Omega_{\mathcal F}\rangle_O|\Phi_d\rangle_{AR}$, where $|\Phi_{d}\rangle $ is the maximally entangled state obtained by adjoining a $d$-dimensional reference register $R$:
    \[
        |\Phi_{d}\rangle := \frac{1}{\sqrt{d}}\sum_{a=1}^{d}|a\rangle_{A} |a\rangle_R.
    \]
    For $j\ge1$, set $|v_j\rangle:=M_{\mathcal E}^j|v_0\rangle$.
    Since $|v_0\rangle\in\mathcal{L}_0$ and $\lVert |v_0\rangle\rVert_2=1$, \(|v_j\rangle\in\mathcal L_{jr}\) for every \(j\ge0\). Moreover, we can write
    \[
        |v_k\rangle=\frac{1}{\sqrt{|\mathcal F|\,d}}\sum_{H\in\mathcal F}\sum_{a=1}^d|H\rangle(E^H)^k|a\rangle_A|a\rangle_R.
    \]
    Since every $E^H$ is Hermitian, orthogonality of the oracle and reference-register bases gives
    \begin{equation}
        \lVert |v_k\rangle\rVert_2^2
        =\E_H\left[\frac1d\sum_{a=1}^d\langle a|(E^H)^{2k}|a\rangle\right]
        =\E_H\left[\frac1d\Tr|E^H|^{2k}\right].
        \label{eq:norm-moment}
    \end{equation}
    The right-hand side is the $2k$-th moment of a random eigenvalue, obtained by sampling $H$ uniformly and then sampling one of the $d$ eigenvalues of $E^H$ uniformly, counting multiplicities.

    On the other hand,
    since \(|v_j\rangle\in\mathcal L_{jr}\subseteq\mathcal L_{(j+1)r}\) and \(|v_{j+1}\rangle\in\mathcal L_{(j+1)r}\), for \(0\le j<k\),
    \[
        |v_{j+1}\rangle = \Pi_{\le (j+1)r} |v_{j+1}\rangle 
        =\Pi_{\le (j+1)r}M_{\mathcal E}|v_j\rangle
        =\Pi_{\le (j+1)r}M_{\mathcal E}\Pi_{\le (j+1)r}|v_j\rangle.
    \]
    Hence, by \eqref{eq:restricted-degree-bound},
    \[
        \lVert |v_{j+1}\rangle\rVert_2
        \le
        \lVert \Pi_{\le (j+1)r}M_{\mathcal E}\Pi_{\le (j+1)r}\rVert_{\mathrm{op}}\lVert |v_j\rangle\rVert_2
        \le
        \beta((j+1)r)\lVert |v_j\rangle\rVert_2.
    \]
    Iterating gives, because $\beta$ is nondecreasing,
    \begin{equation}
        \lVert |v_k\rangle\rVert_2
        \le\prod_{j=0}^{k-1}\beta((j+1)r)
        \le\beta(kr)^k.
        \label{eq:norm-beta}
    \end{equation}
    Combining \eqref{eq:norm-moment} and \eqref{eq:norm-beta} proves the moment bound
    \begin{equation}
        \E_H\left[\frac1d\Tr|E^H|^{2k}\right]^{1/(2k)} \le \beta(kr)
        \label{eq:Schatten moment bound}
    \end{equation}
    for any $k\ge 1$.
    The bound $|\Tr(\rho_H E^H)|\le\lVert E^H\rVert_{\mathrm{op}}$, monotonicity of $L_p$ norms, and \eqref{eq:Schatten moment bound} give, for any $k\ge 1$,
    \begin{align*}
        \left|\E_H\left[\Tr(\rho_H E^H)\right]\right|
        &\le\E_H\left[\lVert E^H\rVert_{\mathrm{op}}\right] \notag\\
        &\le\left(\E_H\left[\lVert E^H\rVert_{\mathrm{op}}^{2k}\right]\right)^{1/(2k)} \notag\\
        &\le d^{1/(2k)}
        \left(\E_H\left[\frac1d\Tr|E^H|^{2k}\right]\right)^{1/(2k)} \notag\\
        &\le d^{1/(2k)}\beta(kr).
    \end{align*}

    Taking $k=S$ gives $d^{1/(2S)}\le (2^S)^{1/(2S)}=\sqrt2$ and hence
    \[
        \left|\E_H\left[\Tr(\rho_H E^H)\right]\right|
        \le\sqrt2\,\beta(Sr).
    \]

Substituting $Sr=2S(T_{\Samp}+T+T_{\Ver})$ in the bound above and taking
the supremum over all $(S,T)$-adversaries proves the claimed inequality.
\end{proof}

Combining \cref{thm:QAI bounded by BF,thm:bf-operator-characterization}
gives the following corollary, which converts database-size-restricted operator
bounds into QAI security bounds.
\begin{corollary}[QAI Security from Database-Size-Restricted Operator Bounds]
\label{cor:qai-from-level-bounds}
Let $\mathcal G$ be a game in the QROM or QRPM.
For integers $S,T\ge0$, let
$    P:=2S(T_{\Samp}+T+T_{\Ver}).$
Then
\[
    \mathsf{Adv}^{\mathsf{QAI}}_{\mathcal G}(S,T)
    \le\sqrt2\sup_{\mathcal A}
    \bigl\|\Pi_{\le P}M_{\mathcal E_{\mathcal A}}\Pi_{\le P}\bigr\|_{\mathrm{op}},
\]
where the supremum ranges over all online algorithms making at most
$T$ queries through $\Query^H$ and their input-register choices, and
$\Pi_{\le P}$ is the projector onto the corresponding database
hierarchy from Section~2.
\end{corollary}
\begin{proof}
Apply \cref{thm:QAI bounded by BF} and substitute the
bit-fixing bound from \cref{thm:bf-operator-characterization}.
\end{proof}

\section{Improving Non-uniform Security for Salting}
\label{sec:salting}

This section proves a theorem for non-uniform salting security whose decision-game bound improves on \cite[Theorem 7.5]{Liu23}. 
We use the salted game formalization of \cite[Definition 7.1]{CGLQ20}. Throughout this section, $\mathcal F$ is either a random-function family or a random-permutation family, and each call to $\mathsf{Query}^H$ makes at most one query to the underlying oracle $H$.

\begin{definition}[Salted Security Game]
\label{def:public-salt}
Let $\mathcal G=(\mathcal C)$ be a security game whose base-game oracle
is sampled as $H\sample\mathcal F$, where $\mathcal C=(\Samp,\Query,\Ver).$
Let
\ifnum\shorter=1
$\boldsymbol H=(H_1,\ldots,H_K)\sample\mathcal F^K$
\else
\[
    \boldsymbol H=(H_1,\ldots,H_K)\sample\mathcal F^K
\]
\fi
be $K$ independent salt-indexed oracles. The salted game
$\mathcal G_{S}=(\mathcal C_{\mathsf S})$ is specified by
\begin{enumerate}
    \item
        \(\mathsf{Sample}_{\mathsf S}^{\mathbf H}\) samples \(s\sample [K]\), runs \(\mathsf{Sample}^{H_s}\) with fresh randomness, outputs \((s,y)\), and retains \((s,t)\), where \(y\) and \(t\) are the base sampler’s challenge and retained state, respectively;

    \item
        $\Query_{\mathsf S}^{\boldsymbol H}
        ((s,t),(s',x))
        =
        \begin{cases}
            \Query^{H_s}(t,x), & s'=s,\\
            H_{s'}(x), & s'\neq s;
        \end{cases}$

    \item
        $\Ver_{\mathsf S}^{\boldsymbol H}((s,t),a)
        =
        \Ver^{H_s}(t,a).$
\end{enumerate}
\end{definition}

Fix an adversary for a salted game. We can similarly define the corresponding game operator and purified operator. For each challenge salt $s\in[K]$,
let $E_s^{\boldsymbol H}$ be its fixed-salt game operator on the adversary's input register $A$.
Let 
$\mathcal E_s:=\{E_s^{\boldsymbol H}\}_{\boldsymbol H\in\mathcal F^K}$
and
$\mathcal E^{(K)}
:=
\{E_{(K)}^{\boldsymbol H}\}_{\boldsymbol H\in\mathcal F^K}$. We define the averaged game operator
\ifnum\shorter=1
\[
    E_{(K)}^{\boldsymbol H}
    :=
    \frac{1}{K}\sum_{s\in[K]}E_s^{\boldsymbol H},\qquad
    M_{\mathcal E^{(K)}}
    =
    \frac{1}{K}\sum_{s\in[K]}M_{\mathcal E_s}.
\]
\else
\[
    E_{(K)}^{\boldsymbol H}
    :=
    \frac{1}{K}\sum_{s\in[K]}E_s^{\boldsymbol H}.
\]
By linearity,
\[
    M_{\mathcal E^{(K)}}
    =
    \frac{1}{K}\sum_{s\in[K]}M_{\mathcal E_s}.
\]
\fi

\subsection{Salted Database Hierarchy}
\label{subsec:Salted Filtration}

Purify \(\mathbf H=(H_1,\ldots,H_K)\leftarrow\mathcal F^K\) in \(\mathcal O:=\mathcal O_1\otimes\cdots\otimes\mathcal O_K\), where \(\mathcal O_s\) holds \(H_s\).
For the base game without salting, recall that $\mathcal D_j$ is the space of exact database size on the purified oracle register, and $\Pi_{=j}$ is the orthogonal projection onto $\mathcal D_j$. The database hierarchy and the database-size observable are
\[
    \mathcal L_P=\bigoplus_{j=0}^{P}\mathcal D_j,
    \qquad
    \Lambda=\sum_{j\ge0}j\Pi_{=j}.
\]

We define the salted database hierarchy and its orthogonal projections by
\[
    \mathcal{L}_{P}^{(K)}:=\bigoplus_{j_{1}+\cdots+j_{K}\le P} \mathcal{D}_{j_{1}}\otimes \cdots \otimes \mathcal{D}_{j_{K}}, \qquad \Pi_{\le P}^{(K)}:=\operatorname{Proj}_{\mathcal{L}_{P}^{(K)}}.
\]
In other words, it keeps track of the database sizes for each individual salt-indexed oracle, and the sum of them is the total database size.

\Cref{cor:qai-from-level-bounds} also applies to the salted
database hierarchy. The initial oracle state satisfies
$|\Omega_{\mathcal F}\rangle^{\otimes K}\in\mathcal L_0^{(K)}$,
and each query increases the total database size by at most one
as shown in \cref{sec:Appendix A}. Thus the moment argument in the proof of
\cref{thm:QAI bounded by BF}, using the database-size-restricted operator norms
directly, gives the corollary with projectors $\Pi_{\le P}^{(K)}$.

Analogously to the base-game database-size observable, for each salt $s\in[K]$
we define the local database-size observable
\[
    \Lambda_s
    :=
    I_{O_1}\otimes\cdots\otimes I_{O_{s-1}}
    \otimes\Lambda
    \otimes I_{O_{s+1}}\otimes\cdots\otimes I_{O_K}
\]
where $\Lambda$ is the database-size observable in $O_{s}$.
Note that for any unit vector $|\psi\rangle \in \mathcal{L}^{(K)}_{P}$, 
\begin{equation}
    \langle\psi|\sum_{s\in [K]}\Lambda_{s} |\psi\rangle\le P.
    \label{eq:averaged lambda inequality}
\end{equation}
This is because on each summand $\mathcal D_{j_1}\otimes\cdots\otimes\mathcal D_{j_K}$ of $\mathcal L_P^{(K)}$, the operator $\sum_{s\in[K]}\Lambda_s$ has eigenvalue $j_1+\cdots+j_K\le P$.

For \(s\in[K]\) and \(j\ge0\), let \(\Pi_{s,= j}\) act as \(\Pi_{=j}\) on \(\mathcal O_s\) and as the identity on every other oracle register:
\[
    \Pi_{s,= j}:=
I_{O_1}\otimes\cdots\otimes I_{O_{s-1}}
\otimes \Pi_{=j}
\otimes I_{O_{s+1}}\otimes\cdots\otimes I_{O_K}.
\]
The following inequality will be useful.
\begin{equation}
    I-\Pi_{s,=0}=\sum_{j\ge 1} \Pi_{s,=j} \preceq \sum_{j\ge 1} j\Pi_{s,=j}=\Lambda_{s}
    \label{eq:lambda_s and Pi_s=0}.
\end{equation}

\subsection{Salting Defeats Quantum Advice}

In this section, we improve the bound of~\cite[Theorem~7.5]{Liu23} for decision games with quantum advice.

\begin{theorem}[Generic salting]
    \label{thm:generic-public-salting}
    Put $r:=2(T_{\Samp}+T+T_{\Ver})$.
    For $S\ge0$, the following statements hold.
    \begin{enumerate}[label=(\roman*)]
        \item
        If the base search game has uniform success probability at most
        $\nu_{\mathcal G}(T)$, then
        \[
            \delta_{\mathcal G_{\mathsf S}}^{\mathsf{QAI}}(S,T)
            =
            O\left(
                \nu_{\mathcal G}(T)+\frac{Sr}{K}
            \right).
        \]
        \item
        If the base decision game has uniform advantage at most
        $\epsilon_{\mathcal G}(T)$, then
        \[
            \epsilon_{\mathcal G_{\mathsf S}}^{\mathsf{QAI}}(S,T)
            =
            O\left(
                \epsilon_{\mathcal G}(T)
                +
                \sqrt{\frac{Sr}{K}}
            \right).
        \]
    \end{enumerate}
\end{theorem}

We use the following lemma to relate the experiment with a fixed challenge salt to the base game, which appears in the proof of~\cite[Lemma~7.2]{CGLQ20}; we restate it here.
At a high level, any bound by a function of the expected database size that holds for all base-game states also holds in the experiment with fixed challenge salt $s$, with the database size measured only in $O_s$. 

\begin{lemma}
    \label{lem:fixed-salt-reduction}
    Fix a challenge salt $s\in[K]$.
    \begin{enumerate}[label=(\roman*)]
        \item
        If $|\psi\rangle\in\operatorname{im}\Pi_{s,=0}$, then the
        fixed-salt experiment is a valid base-game experiment for the
        uniform oracle $H_s$, with all remaining registers included in the internal register
        and with at most $T$ queries to
        $\Query^{H_s}$.
        \label{item:image in level 0}

        \item
        Suppose there exists a real-valued function $f$ such that, for every $T$-query base-game adversary with induced game-operator family $\mathcal E$ and every unit vector $|\phi\rangle$ on the purified-oracle and auxiliary registers,
         \begin{equation}
         |\langle\phi|M_{\mathcal E}|\phi\rangle|
         \le f(\langle\phi|\Lambda|\phi\rangle).
         \label{eq:base-game phi}
        \end{equation}
         
         Then every unit vector \(|\psi\rangle\) on the salt-indexed oracle and auxiliary registers satisfies
         \[
         |\langle\psi|M_{\mathcal E_s}|\psi\rangle|
         \le f(\langle\psi|\Lambda_s|\psi\rangle).
        \]
        \label{item:restricted to salt}
    \end{enumerate}
\end{lemma}
To see this, treat all other oracle registers as auxiliary workspace and absorb their queries into $H_s$-independent local unitaries; each query to the salted interface then uses at most one controlled query to $\Query^{H_s}$. At database size zero, $O_s$ is in the uniform oracle state and is independent of the auxiliary workspace, which gives part (i). For arbitrary states, the same circuit reduction identifies the base-game database-size observable $\Lambda$ with $\Lambda_s$, giving part (ii).
\ifnum\shorter=1
We prove \cref{thm:generic-public-salting} in \cref{subsec:generic-salting} using this lemma.
\else

\begin{proof}[Proof of \cref{thm:generic-public-salting}]
    Fix a challenge salt $s$ and a unit vector $|\psi\rangle\in \mathcal{L}_{P}^{(K)}$. Define $|\psi_{0}\rangle:=\Pi_{s,=0}|\psi\rangle$ and $|\psi_{1}\rangle:=(I-\Pi_{s,=0})|\psi\rangle$. Let $M_{\mathcal E_s}$ be the fixed-salt purified game operator of the salted game.

    For a search game, $0\preceq M_{\mathcal E_s}\preceq I$, and hence
    $M_{\mathcal E_s}^{1/2}$ is well defined.
    By \hyperref[item:image in level 0]
{\cref*{lem:fixed-salt-reduction}\ref*{item:image in level 0}} and the uniform-security assumption,
    \begin{equation}
        \lVert M_{\mathcal{E}_s}^{1/2}|\psi_{0}\rangle\rVert_{2}^{2}=\langle \psi_{0}|M_{\mathcal{E}_s}|\psi_{0}\rangle \le \lVert |\psi_{0}\rangle\rVert_{2}^{2}\nu_{\mathcal{G}}(T) \le \nu_{\mathcal{G}}(T)
        \label{eq:bound psi_0}
    \end{equation}
    From $M_{\mathcal{E}_s}\preceq I$ and \cref{eq:lambda_s and Pi_s=0}, we have
    \begin{align}
        \lVert M_{\mathcal{E}_s}^{1/2}|\psi_{1}\rangle\rVert_{2}^{2}&=\langle \psi_{1}|M_{\mathcal{E}_s}|\psi_{1}\rangle \notag \\
        &\le \langle \psi_{1}|I|\psi_{1}\rangle = \langle \psi|(I-\Pi_{s,=0})|\psi\rangle \notag \\
        &\le \langle \psi|\Lambda_{s}|\psi\rangle.
        \label{eq:bound psi_1}
    \end{align}
    Combining \eqref{eq:bound psi_0} and \eqref{eq:bound psi_1}, we get
    \begin{align}
        \langle \psi|M_{\mathcal{E}_s}|\psi\rangle &= \lVert M_{\mathcal{E}_s}^{1/2}|\psi\rangle\rVert_{2}^{2} \notag\\
        &\le (\lVert M_{\mathcal{E}_s}^{1/2}|\psi_{0}\rangle\rVert_{2}+\lVert M_{\mathcal{E}_s}^{1/2}|\psi_{1}\rangle\rVert_{2})^{2} \notag\\
        &\le  \left(\sqrt{\nu_{\mathcal{G}}(T)} + \sqrt{\langle \psi|\Lambda_{s}|\psi\rangle}\right)^2
        \label{eq:nu before averaging}
    \end{align}
    Averaging \eqref{eq:nu before averaging} over $s \in [K]$ and using \eqref{eq:averaged lambda inequality} and Jensen's inequality gives
    \[
        \langle \psi|M_{\mathcal E^{(K)}}|\psi\rangle \le \left(\sqrt{\nu_{\mathcal{G}}(T)} + \sqrt{P/K}\right)^2.
    \]
    Taking the supremum over unit vectors \(|\psi\rangle\in\mathcal L_P^{(K)}\) gives, for any $P\ge 0$,
    \[
        \bigl\|\Pi_{\le P}^{(K)}M_{\mathcal E^{(K)}}\Pi_{\le P}^{(K)}\bigr\|_{\mathrm{op}}
        \le \left(\sqrt{\nu_{\mathcal{G}}(T)} + \sqrt{P/K}\right)^2.
    \]
    
    This bound holds for every salted online algorithm, so
    applying \cref{cor:qai-from-level-bounds} with $P=Sr$ gives
    \[
        \delta_{\mathcal G_{\mathsf S}}^{QAI}(S,T)
        =
        O\left(
            \left(
                \sqrt{\nu_{\mathcal G}(T)}
                +
                \sqrt{\frac{Sr}{K}}
            \right)^2
        \right)
        =
        O\left(
            \nu_{\mathcal G}(T)+\frac{Sr}{K}
        \right).
    \]
    This proves $(i)$.

   For a decision game, \(M_{\mathcal E_s}\) need not be positive, so the preceding positive-square-root argument does not apply. \hyperref[item:image in level 0]
{\cref*{lem:fixed-salt-reduction}\ref*{item:image in level 0}}, the hypothesis of (ii), and \eqref{eq:lambda_s and Pi_s=0} give
    \begin{equation}
        |\langle \psi_{0}|M_{\mathcal{E}_s}|\psi_{0}\rangle| \le \epsilon_{\mathcal{G}}(T),\qquad \lVert |\psi_{1}\rangle\rVert_{2}^{2} \le \langle \psi|\Lambda_{s}|\psi\rangle.
        \label{eq:decision bound}
    \end{equation}
    By the Cauchy–Schwarz inequality and $||M_{\mathcal{E}_s}||_{op}\le 1$, 
    \begin{equation}
        |\langle \psi_{0}|M_{\mathcal{E}_s}|\psi_{1}\rangle| \le || M_{\mathcal{E}_s}||_{op}\lVert |\psi_{0}\rangle\rVert_{2} \lVert |\psi_{1}\rangle\rVert_{2} \le \lVert |\psi_{1}\rangle\rVert_{2}.
        \label{eq:cauchy}
    \end{equation}
    Combining \eqref{eq:cauchy} and \eqref{eq:decision bound} yields
    \begin{equation}
        |\langle \psi|M_{\mathcal{E}_s}|\psi\rangle| \le \epsilon_{\mathcal{G}}(T) + 2\sqrt{\langle \psi|\Lambda_{s}|\psi\rangle}+\langle \psi|\Lambda_{s}|\psi\rangle.
        \label{eq:epsilon before averaging}
    \end{equation}
    Averaging \eqref{eq:epsilon before averaging} over $s\in[K]$ and using the triangle inequality, \eqref{eq:averaged lambda inequality}, and Jensen’s inequality gives
    \[
        |\langle \psi|M_{\mathcal{E}^{(K)}}|\psi\rangle| \le \epsilon_{\mathcal{G}}(T) + 2\sqrt{P/K}+P/K. 
    \]
    As $|\psi\rangle$ and $P\ge 0$ are arbitrary, for any $P \ge 0$, 
    \[
        \bigl\|\Pi_{\le P}^{(K)}M_{\mathcal E^{(K)}}\Pi_{\le P}^{(K)}\bigr\|_{\mathrm{op}}
        \le \epsilon_{\mathcal{G}}(T) + 2\sqrt{P/K}+P/K.
    \]
    This bound holds for every salted online algorithm, so
    applying \cref{cor:qai-from-level-bounds} with $P=Sr$ gives
    \[
        \epsilon_{\mathcal G_{\mathsf S}}^{QAI}(S,T)
        =
        O\left(
            \epsilon_{\mathcal G}(T)
            +
            \sqrt{\frac{Sr}{K}}
            +
            \frac{Sr}{K}
        \right).
    \]
    If $Sr/K\le1$, then $Sr/K\le\sqrt{Sr/K}$. If $Sr/K>1$, the claimed
    bound follows from the trivial bound on the decision advantage.
    Therefore,
    \[
        \epsilon_{\mathcal G_{\mathsf S}}^{QAI}(S,T)
        =
        O\left(
            \epsilon_{\mathcal G}(T)
            +
            \sqrt{\frac{Sr}{K}}
        \right).
    \]
    This proves $(ii)$.
\end{proof}
\fi

\subsection{\texorpdfstring{$(\gamma,\Delta)$-Diffuse Games}{(gamma, Delta)-Diffuse Games}}

The generic bound in \cref{thm:generic-public-salting} can be
loose for games with an additional random challenge: for salted
permutation inversion on $[N]$, it gives $O(T^2/N+ST/K)$ for $T\ge1$.
In the classical setting, Dong, Liu, and Wu~\cite[Corollary~4.20]{C:DonLiuWu24}
obtain the bound $O(T/N+ST/(KN))$ by exploiting the randomness of both
the salt and the inversion challenge. We introduce the following notion
of a $(\gamma,\Delta)$-diffuse game to capture such improvements through
an additive bound on the contributions of the existing database and
subsequent queries. This yields $O(T^2/N+ST/(KN))$ for salted permutation
inversion with quantum advice and quantum queries
(\cref{thm:sPERM-main}), and also gives improved salting bounds for
PRGs, Yao's box, and function inversion.

\begin{definition}[$(\gamma,\Delta)$-diffuse game]
    \label{def:D-diffuse}
    Fix a query budget $T\ge0$, an exponent $0<\gamma\le1$,
    and a scale $\Delta>0$. A game $\mathcal G$ is
    $(\gamma,\Delta)$-diffuse at query budget $T$, with parameters
    $\tau_T\ge0$ and $\kappa>0$, if every $T$-query adversary with
    induced game-operator family $\mathcal E$ and every unit vector
    $|\psi\rangle$ on the purified-oracle and auxiliary registers satisfy
    \[
        \left|\langle\psi|M_{\mathcal E}|\psi\rangle\right|
        \le \tau_T+
        \kappa\left(\frac{\langle\psi|\Lambda|\psi\rangle}{\Delta}\right)^\gamma.
    \]
\end{definition}

Write $L:=\langle\psi|\Lambda|\psi\rangle$ for the expected size of the existing database. The bound says that the advantages can come from a query-dependent
baseline $\tau_T$ and a database contribution $\kappa(L/\Delta)^\gamma$, in an additive form.
The name ``diffuse'' reflects the weak contribution of a database with
$L\ll\Delta$: for fixed $\gamma$ and $\kappa$, obtaining a constant
advantage above $\tau_T$ requires an expected database size of at least a constant multiple of $\Delta$. This is an additive upper bound; it does
not require the database and subsequent queries to act on orthogonal
subspaces.

The exponent $\gamma$ controls how the database contribution
grows. Our applications use
$\gamma=1/2$ for PRG and Yao's box, and $\gamma=1$ for OWF and
permutation inversion. For a salted database of total database size at most
$P$, averaging over the challenge salt gives expected database size at most
$P/K$. The salting argument therefore replaces the database contribution
by $\kappa(P/(K\Delta))^\gamma$, while leaving the baseline $\tau_T$
unchanged.

\begin{theorem}[Salting for $(\gamma,\Delta)$-diffuse games]
    \label{thm:D-diffuse-salt}
    Fix $T\ge0$ and suppose that the base game $\mathcal G$ is
    $(\gamma,\Delta)$-diffuse at query budget $T$, with parameters
    $\tau_T$ and $\kappa$. Put $r:=2(T_{\Samp}+T+T_{\Ver})$.
    Then, for every integer $S\ge0$ and salt-space size $K\ge1$,
    \[
        \mathsf{Adv}_{\mathcal G_{\mathsf S}}^{\mathsf{QAI}}(S,T)
        \le \sqrt2\left[
            \tau_T+\kappa\left(\frac{Sr}{K\Delta}\right)^\gamma
        \right].
    \]
    This bound holds for both search and decision games.
\end{theorem}

\begin{proof}
    Fix a $T$-query salted online algorithm and a unit vector
    $|\psi\rangle\in\mathcal L_P^{(K)}$.
    By \hyperref[item:restricted to salt]
    {\cref*{lem:fixed-salt-reduction}\ref*{item:restricted to salt}},
    \begin{align*}
        \left|\langle\psi|M_{\mathcal E^{(K)}}|\psi\rangle\right|
        &\le \frac1K\sum_{s\in[K]}
            \left|\langle\psi|M_{\mathcal E_s}|\psi\rangle\right|\\
        &\le \tau_T+\frac\kappa K\sum_{s\in[K]}
            \left(\frac{\langle\psi|\Lambda_s|\psi\rangle}{\Delta}\right)^\gamma\\
        &\le \tau_T+\kappa\left(\frac{P}{K\Delta}\right)^\gamma.
    \end{align*}
    The last inequality follows from
    \eqref{eq:averaged lambda inequality} and Jensen's inequality,
    since $x\mapsto x^\gamma$ is concave and nondecreasing for
    $0<\gamma\le1$.

    Since $M_{\mathcal E^{(K)}}$ is Hermitian, this gives
    \[
        \left\|\Pi_{\le P}^{(K)}M_{\mathcal E^{(K)}}
            \Pi_{\le P}^{(K)}\right\|_{\mathrm{op}}
        \le \tau_T+\kappa\left(\frac{P}{K\Delta}\right)^\gamma.
    \]
    This bound holds for every $T$-query salted online algorithm,
    so applying \cref{cor:qai-from-level-bounds} with $P=Sr$
    proves the claim.
\end{proof}

By \cref{thm:bf-operator-characterization}, bit-fixing
bounds control the size-restricted operator norm at each database size.
For salting, total database size at most $P$ bounds the sum of expected local database sizes by $P$, but each salt's local support can still reach
database size $P$. The following lemma converts a linear size cutoff bound
into a bound in terms of the expected database size.
We use it for OWF and permutation inversion; the decision-game
analogue appears in \cref{subsec:pf of lem:back to psi decision}. 
See \cref{subsec:pf of lem:back to psi search} for a full proof.

\begin{lemma}
    \label{lem:back to psi search}
    Fix an integer $T\ge0$. For a search game $\mathcal G=(\mathcal C)$ with $C=(\Samp,\Query,\Ver)$, put $r:=\max\{1,2(T_{\Samp}+T+T_{\Ver})\}$. Suppose there exist $\tau_T\ge0$, a universal constant $\kappa>0$, and a constant $\Delta>0$ determined by $\mathcal G$ such that, for every $T$-query online algorithm and every integer $P\ge0$,
    \begin{equation}
        ||\Pi_{\le P}M_{\mathcal E}\Pi_{\le P}||_{op}\le \tau_T + \kappa\frac{P}{\Delta}.
        \label{eq:BF bound search}
    \end{equation}
    Then $\mathcal G$ is $(1,\Delta)$-diffuse at query budget $T$ with parameters
    \[
        \tau_T':=2\left(\tau_T+\frac{\kappa(r-1)}{\Delta}\right),\qquad \kappa':=2\kappa.
    \]
\end{lemma}

\ifnum\shorter=1
\section{Applications to Standard Security Games}
\label{sec:Applications to Standard Security Games}
In this section, we briefly summarize the applications of our results. The full details of the game formalizations and proofs are presented in \cref{app:Applications to Standard Security Games}, and the supporting lemmas are proved in \revised{\cref{sec:Missing Proofs}}.



\begin{definition}[PRG Game, informal]
  Let \(H\sample\{H:[N]\to[M]\}\). The PRG game $\mathcal G_{\mathsf{PRG}}$ asks to distinguish $H(x)$ for random $x\in [N]$ from random $y\in [M]$.
\end{definition}

\begin{lemma}
  \label{lem:prg-degree bound}
  For every unit vector $\lvert\psi\rangle$ on the purified-oracle register and any auxiliary registers, 
  \begin{equation}
    \abs{\langle\psi\rvert
      M_{\mathcal E_{\mathsf{PRG}}}
      \lvert\psi\rangle}
    \leq
    \min\!\left\{
      1,
      \sqrt{\frac{\langle\psi\rvert\Lambda\lvert\psi\rangle}{N}}+\frac{2T^2}{N}
    \right\}.
  \end{equation}
  Consequently, for every integer $P\geq0$,
  \begin{equation}
    ||\Pi_{\le P} M_{\mathcal E_{\mathsf{PRG}}}\Pi_{\le P}||_{op} \le 
    \min\left\{ 1, \sqrt{\frac{P}{N}}+\frac{2T^2}{N} \right\}.
  \end{equation}
\end{lemma}

\begin{theorem}
    For all integers $S,T\ge0$, $    \epsilon_{\mathsf{PRG}}^{\mathsf{QAI}}(S,T)
    =O\left(
      \sqrt{\frac{S(T+1)}{N}}+\frac{T^2}{N}
    \right).$
\end{theorem}

\begin{proof}
The size-restricted bound in \cref{lem:prg-degree bound} holds uniformly
over all $T$-query online algorithms. Applying
\cref{cor:qai-from-level-bounds} with $P=Sr_{\mathsf{PRG}} = 2S(T+1)$ gives
\ifnum\shorter=1
the desired bound.
\else
\[
    \epsilon_{\mathsf{PRG}}^{\mathsf{QAI}}(S,T)
    \le
    2\sqrt{
        \frac{S(T+1)}{N}
    }
    +
    2\sqrt{2}\frac{T^2}{N}.
\]
\fi
\end{proof}

\begin{theorem}
  Let $\mathsf{PRG}_{\mathsf{S}}$ be the \revised{salted} PRG game with salt space $[K]$. 
  For all $K,N,M\ge 1$ and $S,T\ge 0$, $\epsilon_{\mathsf{PRG}_{\mathsf{S}}}^{\mathsf{QAI}}(S,T)
    =
    O\left(\sqrt{\frac{S(T+1)}{KN}}+\frac{T^2}{N}\right).$
\end{theorem}

\begin{proof}
    \cref{lem:prg-degree bound} shows that the PRG game is \revised{$(1/2,N)$-diffuse at query budget $T$} with
    \ifnum\shorter=1
    $\tau_T={2T^2}/{N},\kappa={1}.$ \cref{thm:D-diffuse-salt} proves the theorem.
    \else
    \[
            \tau_T=\frac{2T^2}{N},
            \qquad
            \kappa=\revised{1}.
    \]
   \cref{thm:D-diffuse-salt} proves the theorem.
   \fi
\end{proof}


\begin{definition}[Yao's Box]
  Let $H\sample\{H:[N]\rightarrow \{0,1\}\}$. The Yao's box game $\mathcal G_{\mathsf{YB}}$ asks to determine $H(x)$ for $x\sample[N]$ given access to $H'$ that is identical to $H$ except for $x$; $H'(x):=\bot$.
\end{definition}

\begin{lemma}
  \label{lem:yb degree bound}
  For every unit vector $\lvert\psi\rangle$ on the purified-oracle register and any auxiliary registers, 
  \[
    |{\langle\psi\rvert
      M_{\mathcal{E}_{\mathsf{YB}}}
      \lvert\psi\rangle}|
    \leq
    \min\!\left\{\frac12,
      \revised{(1+\sqrt{2})}\sqrt{\frac{\langle\psi\rvert\Lambda\lvert\psi\rangle}{N}}
    \right\}.
  \]
  Consequently, for every integer $P\geq0$,
  \[
    ||\Pi_{\le P} M_{\mathcal{E}_{\mathsf{YB}}}\Pi_{\le P}||_{op}
    \leq
    \min\!\left\{\frac12,\revised{(1+\sqrt{2})}\sqrt{\frac PN}\right\}.
  \]
\end{lemma}

\begin{theorem}
  For all integers $S,T\ge 0$, $ \epsilon_{\mathsf{YB}}^{\mathrm{QAI}}(S,T)
    =
      O\left(
      \sqrt{\frac{S(T+1)}{N}}\right).$
\end{theorem}

\begin{proof}
The size-restricted bound in \cref{lem:yb degree bound} holds uniformly
over all $T$-query online algorithms. Applying
\cref{cor:qai-from-level-bounds} with $P=Sr_{\mathsf{YB}}=2S(T+1)$ gives
\ifnum\shorter=1
the desired upper bound.
\else
\[
    \epsilon_{\mathsf{YB}}^{\mathsf{QAI}}(S,T)
    \le \revised{(2\sqrt{2}+2)}\sqrt{\frac{S(T+1)}{N}}.
\]
\fi
\end{proof}

\begin{theorem}
  Let $\mathsf{YB}_{\mathsf{S}}$ be the \revised{salted} Yao's box with salt space $[K]$. For all $K,N\ge1$ and $S,T\ge0$,
  $\epsilon_{\mathsf{YB}_\mathsf{S}}^{\mathsf{QAI}}(S,T)
        =
      O\left(
        \sqrt{\frac{S(T+1)}{KN}}
      \right).$
\end{theorem}

\begin{proof}
    \cref{lem:yb degree bound} shows that the Yao's-box game is \revised{$(1/2,N)$-diffuse at query budget $T$} with
    \ifnum\shorter=1
    $\tau_{T}=0,\kappa=\revised{1+\sqrt{2}}.$ Applying \cref{thm:D-diffuse-salt} proves the theorem.
    \else
    \[
        \tau_{T}=0,\qquad \kappa=\revised{1+\sqrt{2}}.
    \]
 \cref{thm:D-diffuse-salt} proves the theorem.
 \fi
\end{proof}


\begin{definition}[OWF Game]
  Let $H\sample\mathcal{F}:=\{H:[N]\rightarrow[M]\}$.  The OWF game $\mathcal G_{\mathsf{OWF}}$ asks to find $x'$ given $H(x)$ for $x\sample [N]$ such that $H(x')=H(x)$.
\end{definition}

\begin{lemma}
  \label{lem:owf-degree bound}
  There exists a universal constant $C_{\mathsf {OWF}}>0$ such that for every unit vector $\lvert\psi\rangle$ on the purified-oracle register and any auxiliary registers,
  \[
    \langle\psi\rvert
      M_{\mathcal{E}_{\mathsf{OWF}}}
      \lvert\psi\rangle
    \leq
    \min\!\left\{1,
      C_{\mathsf {OWF}}
      \frac{ \langle\psi\rvert\Lambda\lvert\psi\rangle+(T+1)^2}{\min\{N,M\}}
    \right\}.
  \]
  Consequently, for every integer $P\ge 0$,
  \[
      ||\Pi_{\le P}M_{\mathcal{E}_{\mathsf{OWF}}}\Pi_{\le P}||_{op}\le \min\left\{1,C_{\mathsf {OWF}}\frac{P+(T+1)^{2}}{\min \{N,M\}}\right\}.
  \]
\end{lemma}

\begin{theorem}
  For all integers $S,T\ge 0$,
  $\delta_{\mathsf{OWF}}^{\mathrm{QAI}}(S,T)
    =
      O\left(\frac{S(T+1) + (T+1)^{2}}{\min \{N,M\}}\right).$
\end{theorem}

\begin{proof}
The size-restricted bound in \cref{lem:owf-degree bound} holds uniformly
over all $T$-query online algorithms. Applying
\cref{cor:qai-from-level-bounds} with $P=Sr_{\mathsf{OWF}}=2S(T+2)$ gives
\ifnum\shorter=1
the above bound.
\else
\[
    \delta_{\mathsf{OWF}}^{\mathsf{QAI}}(S,T)
    \le C_{\mathsf{OWF}}\sqrt{2}
    \frac{2S(T+2)+(T+1)^2}{\min\{N,M\}}.
\]
\fi
\end{proof}

\begin{theorem}
  Let $\mathsf{OWF}_{\mathsf{S}}$ be the \revised{salted} OWF game with salt space $[K]$. For all $K,N,M\ge1$ and $S,T\ge0$,
  $\delta_{\mathsf{OWF}_\mathsf{S}}^{\mathsf{QAI}}(S,T)
        =
      O\left(\frac{S(T+1)}{K\min\{N,M\}}+\frac{(T+1)^{2}}{\min\{N,M\}}\right).$
\end{theorem}

\begin{proof}
    \cref{lem:owf-degree bound} shows that the OWF game is $(1,\min\{N,M\})$-diffuse at query budget $T$ with
    \ifnum\shorter=1
    $\tau_{T}=\frac{C_{\mathsf {OWF}}(T+1)^2}{\min\{N,M\}}, \kappa=C_{\mathsf {OWF}}.$
    \cref{thm:D-diffuse-salt} proves the theorem.
    \else
    \[
        \tau_{T}=C_{\mathsf {OWF}}\frac{(T+1)^2}{\min\{N,M\}}, \qquad \kappa=C_{\mathsf {OWF}}.
    \]
    \cref{thm:D-diffuse-salt} proves the theorem.
    \fi
\end{proof}


\begin{definition}[Permutation Inversion Game]
  Let $H\sample\mathcal{F}^{\perm}:=S_{N}$.  The permutation inversion game $\mathcal G_{\PermInv}$ asks to find $x$ given $H(x)$ for $x\sample [N]$.
\end{definition}

\begin{lemma}
  \label{lem:perm-degree bound}
  There exists a universal constant $C_{\perm}>0$ such that for every unit vector $\lvert\psi\rangle$ on the purified-oracle and auxiliary registers,
  \[
    \langle\psi\rvert
      M_{\mathcal{{E}_{\PermInv}}}
      \lvert\psi\rangle
    \leq
    \min\!\left\{1,
      C_{\perm}
        \frac{\langle\psi\rvert\Lambda^{\perm}\lvert\psi\rangle+(T+1)^2}{N}
    \right\}
  \]
  Consequently, for every integer $P\ge 0$,
  \[
      ||\Pi_{\le P}^{\perm}M_{\mathcal{E}_{\PermInv}}\Pi_{\le P}^{\perm}||_{op}\le \min\left\{1,C_{\perm}
        \frac{P+(T+1)^2}{N}
      \right\}.
  \]
\end{lemma}

\begin{theorem}
  For all integers $S,T\ge 0$,
  $
    \delta_{\PermInv}^{\mathrm{QAI}}(S,T)
    =
      O\left(
      \frac{S(T+1) + (T+1)^{2}}{N}\right).$
\end{theorem}

\begin{proof}
The size-restricted bound in \cref{lem:perm-degree bound} holds uniformly
over all $T$-query online algorithms. Applying
\cref{cor:qai-from-level-bounds} with $P=Sr_{\PermInv}=2S(T+1)$ 
\ifnum\shorter=1
proves the theorem.
\else
gives
\[
    \delta_{\PermInv}^{\mathsf{QAI}}(S,T)
    \le C_{\perm}\sqrt{2}\frac{2S(T+1)+(T+1)^2}{N}.
\]
\fi
\end{proof}

\begin{theorem}
  Let $\PermInv_{\mathsf{S}}$ be the \revised{salted} permutation inversion game with salt space $[K]$. For $S,T\ge 0$ and $K,N\ge 1$, $\delta_{\PermInv_\mathsf{S}}^{\mathsf{QAI}}(S,T)
        =
      O\left(\frac{S(T+1)}{KN}+\frac{(T+1)^{2}}{N}\right).$
\end{theorem}

\begin{proof}
    \cref{lem:perm-degree bound} shows that the permutation-inversion game is \revised{$(1,N)$-diffuse at query budget $T$} with
    \ifnum\shorter=1
    $\tau_{T}=\frac{C_{\perm}(T+1)^2}{N},\kappa=C_{\perm}.$
    Applying
    \cref{thm:D-diffuse-salt} proves the theorem.
    \else
    \[
        \tau_{T}=C_{\perm}\frac{(T+1)^2}{N}, \qquad \kappa=C_{\perm}.
    \]
    \cref{thm:D-diffuse-salt} proves the theorem.
    \fi
\end{proof}

\else
\section{Details of Applications}
\label{app:Applications to Standard Security Games}

\ifnum\shorter=1
In this section, we explain the full details of \cref{sec:Applications to Standard Security Games}.
For convenience, we reiterate all the final results and the supporting lemmas. The lemmas are proved in \cref{sec:Missing Proofs}.

We apply our new reductions to the PRG, Yao's box, OWF, and permutation-inversion games, and to their salted variants. For PRG and Yao's box, we prove the required diffuse bounds directly; for OWF and permutation inversion, we derive them from known bit-fixing bounds using \cref{lem:back to psi search}. We use the formulations of the first three games from~\cite{CGLQ20} and that of permutation inversion from~\cite{ABC+26}. 
\else
We apply the reductions established in \cref{sec:The degree-Restricted Moment Method for Oracle-Dependent Advice} and \cref{sec:salting} to the PRG, Yao’s box, OWF, and permutation-inversion games, and to their salted variants. For PRG and Yao's box, we prove the required diffuse bounds directly; for OWF and permutation inversion, we derive them from known bit-fixing bounds using \cref{lem:back to psi search}. We use the formulations of the first three games from~\cite{CGLQ20} and that of permutation inversion from~\cite{ABC+26}. The supporting lemmas are proved in \cref{sec:Missing Proofs}.
\fi

\subsection{PRG Game}

\begin{definition}[PRG Game]
  \label{def:prg-game}
  Let \(H\sample \mathcal F:=\{H:[N]\to[M]\}\). The decision game $\mathcal G_{\mathsf{PRG}}=(\mathcal C_{\mathsf{PRG}})$ is specified by
  \begin{enumerate}[label=\arabic*.]
    \item $\Samp^H(b,x,y)$ outputs $H(x)$ if $b=0$ and outputs $y$ if $b=1$, where $b\sample\{0,1\}$, $x\sample[N]$, and $y\sample[M]$;
    \item $\Query^H((b,x,y),x')=H(x')$;
    \item $\Ver^H((b,x,y),b')=1$ if and only if $b=b'$.
  \end{enumerate}
\end{definition}

Throughout this subsection, the adversary outputs $0$ when it decides
that the challenge is a random image.
Fix a \(T\)-query adversary \(\mathcal A\).
For every oracle $H$ and challenge $y\in[M]$, let $E_y^H$ be the
POVM element on the advice register $A$ corresponding to output $0$.
Define
\[
    B_{H,\unif}^{\mathsf{PRG}}:= \frac{1}{M}\sum_{y\in [M]}E_{y}^{H}, \qquad B_{H,\img}^{\mathsf{PRG}}:=\frac{1}{N}\sum_{x \in [N]}E_{H(x)}^{H}.
\]
Let $p_{\unif}$ and $p_{\img}$ denote the probabilities that the adversary outputs $0$ in the uniform and random-image branches,
respectively. Then
\[
    \Pr[\mathsf{win}]
    =
    \frac{1}{2}(1-p_{\unif})
    +
    \frac{1}{2}p_{\img}
    =
    \frac{1}{2}
    +
    \frac{1}{2}(p_{\img}-p_{\unif}).
\]

Define the PRG game operator by
\[
    E_{\mathsf{PRG}}^H
    :=
    \frac{1}{2}(
    B_{H,\img}^{\mathsf{PRG}}
    -
    B_{H,\unif}^{\mathsf{PRG}}).
\]
We have $0\preceq B_{H,\img}^{\mathsf{PRG}}\preceq I$ and $0\preceq B_{H,\unif}^{\mathsf{PRG}}\preceq I$.
Hence, $E_{\mathsf{PRG}}^H$ is Hermitian and
$\lVert E_{\mathsf{PRG}}^H\rVert_{\mathsf{op}}\le1$.

For an advice family $\{\rho_H\}_H$,
\[
    \frac{1}{2}(p_{\img}-p_{\unif})
    =
    \E_H\!\left[
        \Tr(\rho_HE_{\mathsf{PRG}}^H)
    \right].
\]
Consequently,
\[
    \epsilon_{\mathsf{PRG}}^{\mathsf{QAI}}(S,T)
    =
    \sup_{\mathcal A}
    \left|
        \E_H\!\left[
            \Tr(\rho_HE_{\mathsf{PRG}}^H)
        \right]
    \right|.
    \label{eq:prg-advantage-operator}
\]

Let
$\mathcal E_{\mathsf{PRG}}
:=
\{E_{\mathsf{PRG}}^H\}_{H\in\mathcal F}$
and let $M_{\mathcal E_{\mathsf{PRG}}}$ be its purified game operator.
$M_{\mathcal{{E}_{\mathsf{PRG}}}}$ and $M_{{\mathcal E}_{\mathsf{PRG}}^{(K)}}$ increase the database size by at most $r_{\mathsf {PRG}}=2(T+1)$.

\begin{lemma}
  \label{applem:prg-degree bound}
    \ifnum\shorter=0
  \label{lem:prg-degree bound}
    \fi
  For every unit vector $\lvert\psi\rangle$ on the purified-oracle register and any auxiliary registers, 
  \begin{equation}
    \abs{\langle\psi\rvert
      M_{\mathcal E_{\mathsf{PRG}}}
      \lvert\psi\rangle}
    \leq
    \min\!\left\{
      1,
      \sqrt{\frac{\langle\psi\rvert\Lambda\lvert\psi\rangle}{N}}+\frac{2T^2}{N}
    \right\}.
    \label{eq:prg-complete-degree-sensitive}
  \end{equation}
  Consequently, for every integer $P\geq0$,
  \begin{equation}
    ||\Pi_{\le P} M_{\mathcal E_{\mathsf{PRG}}}\Pi_{\le P}||_{op} \le 
    \min\left\{ 1, \sqrt{\frac{P}{N}}+\frac{2T^2}{N} \right\}.
    \label{eq:prg-degree-restricted}
  \end{equation}
\end{lemma}
\begin{proof}
    See \cref{subsec:pf of lem:prg-degree bound}.
\end{proof}

\begin{theorem}[QAI bound for the PRG game]
    \label{thm:prg-main}
    For all integers $S,T\ge0$,
  \[
    \epsilon_{\mathsf{PRG}}^{\mathsf{QAI}}(S,T)
    =O\left(
      \sqrt{\frac{S(T+1)}{N}}+\frac{T^2}{N}
    \right).
  \]
\end{theorem}

\begin{proof}
The size-restricted bound in \cref{applem:prg-degree bound} holds uniformly
over all $T$-query online algorithms. Applying
\cref{cor:qai-from-level-bounds} with $P=Sr_{\mathsf{PRG}} = 2S(T+1)$ gives
\ifnum\shorter=1
the desired bound.
\else
\[
    \epsilon_{\mathsf{PRG}}^{\mathsf{QAI}}(S,T)
    \le
    2\sqrt{
        \frac{S(T+1)}{N}
    }
    +
    2\sqrt{2}\frac{T^2}{N}.
\]
\fi
\end{proof}

\begin{theorem}[QAI bound for salted PRG game]
  \label{thm:sprg-main}
  Let $\mathsf{PRG}_{\mathsf{S}}$ be the salted PRG game with salt space $[K]$. For all $K,N,M\ge 1$ and $S,T\ge 0$,
  \[
    \epsilon_{\mathsf{PRG}_{\mathsf{S}}}^{\mathsf{QAI}}(S,T)
    =
    O\left(\sqrt{\frac{S(T+1)}{KN}}+\frac{T^2}{N}\right).
    \]
\end{theorem}

\begin{proof}
    \cref{applem:prg-degree bound} shows that the PRG game is $(1/2,N)$-diffuse at query budget $T$ with
    \ifnum\shorter=1
    $\tau_T={2T^2}/{N},\kappa={1}.$ \cref{thm:D-diffuse-salt} proves the theorem.
    \else
    \[
            \tau_T=\frac{2T^2}{N},
            \qquad
            \kappa=1.
    \]
   \cref{thm:D-diffuse-salt} proves the theorem.
   \fi
\end{proof}

\subsection{Yao's Box}

\begin{definition}[Yao's Box]
  \label{def:yb-game}
  Let $H\sample\mathcal{F}:=\{H:[N]\rightarrow \{0,1\}\}$. The decision game $\mathcal G_{\mathsf{YB}}=(\mathcal C_{\mathsf{YB}})$ is specified by
  \begin{enumerate}[label=\arabic*.]
    \item $\Samp^H(x)=x$ for $x\sample[N]$;
    \item $\Query^H(x,x')=H(x')$ when $x'\neq x$, and $\Query^H(x,x)=1$;
    \item $\Ver^H(x,b)=1$ if and only if $b=H(x)$.
  \end{enumerate}
\end{definition}

The query interface is independent of the hidden bit $H(x)$. Fix a \(T\)-query adversary \(\mathcal A\).
For every oracle $H$ and challenge $x\in[N]$, let $E_{x,b}^H$ be the
POVM element on the advice register $A$ corresponding to output $b\in\{0,1\}$.
Then, $E_{x,0}^{H}+E_{x,1}^{H}=I$. Defining
\[
    A_{x}^H:=E_{x,0}^H-E_{x,1}^H,
\]
we have 
\[
    E_{x,H(x)}^H=\frac{1}{2}(I+(-1)^{H(x)}A_{x}^H).
\]
Letting $\chi_x(H):=(-1)^{H(x)}$, define
\[
    E^H_{\mathsf{YB}}:=\frac{1}{2N}\sum_{x\in [N]}\chi_{x}(H)A_{x}^H.
\]
For an advice family $\{\rho_H\}_{H}$, the winning probability is
\begin{equation}
    \Pr[\mathsf{win}]=\E_{H,x}[\Tr(\rho_HE_{x,H(x)}^H)]=\frac{1}{2}+\E_{H}[\Tr(\rho_HE_{\mathsf{YB}}^H)].
    \label{eq:PRG winning pb}
\end{equation}
Therefore,
\[
    \epsilon_{\mathsf{YB}}^{\mathsf{QAI}}(S,T)=\sup_{\mathcal{A}}|\E_{H}[\Tr(\rho_HE_{\mathsf{YB}}^H)]|.
\]
Let $\mathcal{E}_{\mathsf{YB}}:=\{E^{H}_{\mathsf{YB}}\}_{H\in\mathcal{F}}$ and let $M_{\mathcal{E}_{\mathsf{YB}}}$ be its purified game operator.
$M_{\mathcal{{E}_{\mathsf{YB}}}}$ and $M_{{\mathcal E}_{\mathsf{YB}}^{(K)}}$ increase the database size by at most $r_{\mathsf {YB}}=2(T+1)$.

\begin{lemma}
  \label{applem:yb degree bound}
    \ifnum\shorter=0
  \label{lem:yb degree bound}
    \fi
  For every unit vector $\lvert\psi\rangle$ on the purified-oracle register and any auxiliary registers, 
  \[
    |{\langle\psi\rvert
      M_{\mathcal{E}_{\mathsf{YB}}}
      \lvert\psi\rangle}|
    \leq
    \min\!\left\{\frac12,
      (1+\sqrt{2})\sqrt{\frac{\langle\psi\rvert\Lambda\lvert\psi\rangle}{N}}
    \right\}.
  \]
  Consequently, for every integer $P\geq0$,
  \[
    ||\Pi_{\le P} M_{\mathcal{E}_{\mathsf{YB}}}\Pi_{\le P}||_{op}
    \leq
    \min\!\left\{\frac12,(1+\sqrt{2})\sqrt{\frac PN}\right\}.
  \]
\end{lemma}
\begin{proof}
    See \cref{subsec:pf of lem:yb degree bound}.
\end{proof}

\begin{theorem}[QAI bound for Yao's box]
  \label{thm:yb-main}
  For all integers $S,T\ge 0$,
  \[
    \epsilon_{\mathsf{YB}}^{\mathrm{QAI}}(S,T)
    =
      O\left(
      \sqrt{\frac{S(T+1)}{N}}\right).
  \]
\end{theorem}

\begin{proof}
The size-restricted bound in \cref{applem:yb degree bound} holds uniformly
over all $T$-query online algorithms. Applying
\cref{cor:qai-from-level-bounds} with $P=Sr_{\mathsf{YB}}=2S(T+1)$ gives
\ifnum\shorter=1
the desired upper bound.
\else
\[
    \epsilon_{\mathsf{YB}}^{\mathsf{QAI}}(S,T)
    \le (2\sqrt{2}+2)\sqrt{\frac{S(T+1)}{N}}.
\]
\fi
\end{proof}

\begin{theorem}[QAI bound for salted Yao's box]
  \label{thm:syb-main}
  Let $\mathsf{YB}_{\mathsf{S}}$ be the salted Yao's box with salt space $[K]$. For all $K,N\ge1$ and $S,T\ge0$,
  \[
      \epsilon_{\mathsf{YB}_\mathsf{S}}^{\mathsf{QAI}}(S,T)
        =
      O\left(
        \sqrt{\frac{S(T+1)}{KN}}
      \right).
  \]
\end{theorem}

\begin{proof}
    \cref{applem:yb degree bound} shows that the Yao's-box game is $(1/2,N)$-diffuse at query budget $T$ with
    \ifnum\shorter=1
    $\tau_{T}=0,\kappa=1+\sqrt{2}.$ Applying \cref{thm:D-diffuse-salt} proves the theorem.
    \else
    \[
        \tau_{T}=0,\qquad \kappa=1+\sqrt{2}.
    \]
 \cref{thm:D-diffuse-salt} proves the theorem.
 \fi
\end{proof}

\subsection{OWF Game}

\begin{definition}[OWF Game]
  \label{def:owf-game}
  Let $H\sample\mathcal{F}:=\{H:[N]\rightarrow[M]\}$.  The security game $\mathcal G_{\mathsf{OWF}}=(\mathcal C_{\mathsf{OWF}})$ is specified by
  \begin{enumerate}[label=\arabic*.]
    \item $\Samp^H(x)$ takes $x\sample[N]$, outputs $y=H(x)$, and retains $(x,y)$;
    \item $\Query^H((x,y),x')=H(x')$;
    \item $\Ver^H((x,y),x')=1$ if and only if $x'\in [N]$ and $y=H(x')$.
  \end{enumerate}
\end{definition}

Fix a \(T\)-query adversary \(\mathcal A\). For every oracle $H$ and challenge $y\in [M]$, let $E_{y,z}^{H}$ be the POVM element on the advice register $A$ corresponding to output $z\in [N]\cup \{\bot\}$. Define
\[
    W_{y}^H:=\sum_{z:H(z)=y}E_{y,z}^{H}, \qquad E^H_{\mathsf{OWF}}:=\frac{1}{N}\sum_{x\in [N]}W^H_{H(x)}.
\]
Then, $E^H_{\mathsf{OWF}}$ is the game operator corresponding to the winning outcome and satisfies $0\preceq E^H_{\mathsf{OWF}}\preceq I$.
For an advice family $\{\rho_H\}_{H}$,
\[
    \Pr[\mathsf{win}]=\E_H[\Tr(\rho_H E^H_{\mathsf{OWF}})].
\]
Consequently,
\[
    \delta^{\mathsf{QAI}}_{\mathsf{OWF}}(S,T)=\sup_{\mathcal{A}}\E_H[\Tr(\rho_H E^H_{\mathsf{OWF}})].
\]
Let $\mathcal{E}_{\mathsf{OWF}}:=\{E^H_{\mathsf{OWF}}\}_{H\in\mathcal{F}}$ and let $M_{\mathcal{{E}_{\mathsf{OWF}}}}$ be its purified game operator.
$M_{\mathcal{{E}_{\mathsf{OWF}}}}$ and $M_{{\mathcal E}_{\mathsf{OWF}}^{(K)}}$ increase the database size by at most $r_{\mathsf {OWF}}=2(T+2)$.

\begin{lemma}
    \ifnum\shorter=0
  \label{lem:owf-degree bound}
    \fi
  \label{applem:owf-degree bound}
  There exists a universal constant $C_{\mathsf {OWF}}>0$ such that for every unit vector $\lvert\psi\rangle$ on the purified-oracle register and any auxiliary registers,
  \[
    \langle\psi\rvert
      M_{\mathcal{E}_{\mathsf{OWF}}}
      \lvert\psi\rangle
    \leq
    \min\!\left\{1,
      C_{\mathsf {OWF}}
      \frac{ \langle\psi\rvert\Lambda\lvert\psi\rangle+(T+1)^2}{\min\{N,M\}}
    \right\}.
  \]
  Consequently, for every integer $P\ge 0$,
  \[
      ||\Pi_{\le P}M_{\mathcal{E}_{\mathsf{OWF}}}\Pi_{\le P}||_{op}\le \min\left\{1,C_{\mathsf {OWF}}\frac{P+(T+1)^{2}}{\min \{N,M\}}\right\}.
  \]
\end{lemma}
\begin{proof}
    See \cref{subsec:pf of lem:owf-degree bound}.
\end{proof}

\begin{theorem}[QAI bound for the OWF game]
  \label{thm:owf-main}
  For all integers $S,T\ge 0$,
  \[
    \delta_{\mathsf{OWF}}^{\mathrm{QAI}}(S,T)
    =
      O\left(\frac{S(T+1) + (T+1)^{2}}{\min \{N,M\}}\right).
  \]
\end{theorem}

\begin{proof}
The size-restricted bound in \cref{applem:owf-degree bound} holds uniformly
over all $T$-query online algorithms. Applying
\cref{cor:qai-from-level-bounds} with $P=Sr_{\mathsf{OWF}}=2S(T+2)$ gives
\ifnum\shorter=1
the above bound.
\else
\[
    \delta_{\mathsf{OWF}}^{\mathsf{QAI}}(S,T)
    \le C_{\mathsf{OWF}}\sqrt{2}
    \frac{2S(T+2)+(T+1)^2}{\min\{N,M\}}.
\]
\fi
\end{proof}

\begin{theorem}[QAI bound for salted OWF game]
  \label{thm:sowf-main}
  Let $\mathsf{OWF}_{\mathsf{S}}$ be the salted OWF game with salt space $[K]$. For all $K,N,M\ge1$ and $S,T\ge0$,
  \[
      \delta_{\mathsf{OWF}_\mathsf{S}}^{\mathsf{QAI}}(S,T)
        =
      O\left(\frac{S(T+1)}{K\min\{N,M\}}+\frac{(T+1)^{2}}{\min\{N,M\}}\right).
  \]
\end{theorem}

\begin{proof}
    \cref{applem:owf-degree bound} shows that the OWF game is $(1,\min\{N,M\})$-diffuse at query budget $T$ with
    \ifnum\shorter=1
    $\tau_{T}=\frac{C_{\mathsf {OWF}}(T+1)^2}{\min\{N,M\}}, \kappa=C_{\mathsf {OWF}}.$
    \cref{thm:D-diffuse-salt} proves the theorem.
    \else
    \[
        \tau_{T}=C_{\mathsf {OWF}}\frac{(T+1)^2}{\min\{N,M\}}, \qquad \kappa=C_{\mathsf {OWF}}.
    \]
    \cref{thm:D-diffuse-salt} proves the theorem.
    \fi
\end{proof}

\subsection{Permutation inversion game}

\begin{definition}[Permutation Inversion Game]
  \label{def:permutation-game}
  Let $H\sample\mathcal{F}^{\perm}:=S_{N}$.  The security game $\mathcal G_{\PermInv}=(\mathcal C_{\PermInv})$ is specified by
  \begin{enumerate}[label=\arabic*.]
    \item $\Samp^H(y)=y$ for $y\sample[N]$;
    \item $\Query^H(y,x)=H(x)$;
    \item $\Ver^H(y,x)=1$ if and only if $y=H(x)$.
  \end{enumerate}
\end{definition}

Fix a \(T\)-query adversary \(\mathcal A\). For every oracle $H$ and challenge $y\in[N]$, let $E_{y,z}^{H}$ be the POVM element on the advice register $A$ corresponding to output $z\in [N]$. Define
\[
    E_{\PermInv}^{H}:=\frac{1}{N}\sum_{y\in [N]}E_{y,H^{-1}(y)}^H.
\]
Then, $E^{H}_{\PermInv}$ is the game operator corresponding to the winning outcome and satisfies $0\preceq E^{H}_{\PermInv}\preceq I$. For an advice family $\{\rho_H\}_{H}$,
\[
    \Pr[\mathsf{win}]=\E_{H}[\Tr(\rho_{H}E_{\PermInv}^{H})].
\]
Consequently,
\[
    \delta^{\mathsf{QAI}}_{\PermInv}(S,T)=\sup_{\mathcal{A}}\E_H[\Tr(\rho_H E^H_{\PermInv})].
\]
Let $\mathcal{E}_{\PermInv}:=\{E^H_{\PermInv}\}_{H\in\mathcal{F}^{\perm}}$ and let $M_{\mathcal{{E}_{\PermInv}}}$ be its purified game operator.
$M_{\mathcal{{E}_{\mathsf{PermInv}}}}$ and $M_{{\mathcal E}_{\mathsf{PermInv}}^{(K)}}$ increase the database size by at most $r_{\mathsf {PermInv}}=2(T+1)$.

\begin{lemma}
\ifnum\shorter=0
  \label{lem:perm-degree bound}
  \fi
  \label{apple:perm-degree bound}
  There exists a universal constant $C_{\perm}>0$ such that for every unit vector $\lvert\psi\rangle$ on the purified-oracle and auxiliary registers,
  \[
    \langle\psi\rvert
      M_{\mathcal{{E}_{\PermInv}}}
      \lvert\psi\rangle
    \leq
    \min\!\left\{1,
      C_{\perm}
        \frac{\langle\psi\rvert\Lambda^{\perm}\lvert\psi\rangle+(T+1)^2}{N}
    \right\}.
  \]
  Consequently, for every integer $P\ge 0$,
  \[
      ||\Pi_{\le P}^{\perm}M_{\mathcal{E}_{\PermInv}}\Pi_{\le P}^{\perm}||_{op}\le \min\left\{1,C_{\perm}
        \frac{P+(T+1)^2}{N}
      \right\}.
  \]
\end{lemma}
\begin{proof}
    See \cref{subsec:pf of lem:perm-degree bound}.
\end{proof}

\begin{theorem}[QAI bound for Permutation Inversion Game]
  \label{thm:perm-main}
  For all integers $S,T\ge 0$,
  \[
    \delta_{\PermInv}^{\mathrm{QAI}}(S,T)
    =
      O\left(
      \frac{S(T+1) + (T+1)^{2}}{N}\right).
  \]
\end{theorem}

\begin{proof}
The size-restricted bound in \cref{apple:perm-degree bound} holds uniformly
over all $T$-query online algorithms. Applying
\cref{cor:qai-from-level-bounds} with $P=Sr_{\PermInv}=2S(T+1)$ 
\ifnum\shorter=1
proves the theorem.
\else
gives
\[
    \delta_{\PermInv}^{\mathsf{QAI}}(S,T)
    \le C_{\perm}\sqrt{2}\frac{2S(T+1)+(T+1)^2}{N}.
\]
\fi
\end{proof}

\begin{theorem}[QAI bound for salted Permutation Inversion game]
  \label{thm:sPERM-main}
  Let $\PermInv_{\mathsf{S}}$ be the salted permutation inversion game with salt space $[K]$. For all integers $S,T\ge 0$ and $K,N\ge 1$,
  \[
      \delta_{\PermInv_\mathsf{S}}^{\mathsf{QAI}}(S,T)
        =
      O\left(\frac{S(T+1)}{KN}+\frac{(T+1)^{2}}{N}\right).
  \]
\end{theorem}

\begin{proof}
    \cref{apple:perm-degree bound} shows that the permutation-inversion game is $(1,N)$-diffuse at query budget $T$ with
    \ifnum\shorter=1
    $\tau_{T}=\frac{C_{\perm}(T+1)^2}{N},\kappa=C_{\perm}.$
    Applying
    \cref{thm:D-diffuse-salt} proves the theorem.
    \else
    \[
        \tau_{T}=C_{\perm}\frac{(T+1)^2}{N}, \qquad \kappa=C_{\perm}.
    \]
    \cref{thm:D-diffuse-salt} proves the theorem.
    \fi
\end{proof}

\fi

\bibliographystyle{alpha}
\bibliography{cryptobib/abbrev3,cryptobib/crypto,ref}
\appendix
\crefalias{section}{appendix}

\ifnum\shorter=1
\section{Proofs of the Main Theorems}
\subsection{Proof of \texorpdfstring{\cref{thm:degree-restricted-moment}}{thm~degree-restricted-moment}}
\label{subsec: degree-restricted-proof}

\begin{proof}[Proof of \cref{thm:degree-restricted-moment}]
Fix an online algorithm $\mathcal A$, and write
$\mathcal E:=\mathcal E_{\mathcal A}$. Any offline phase making at most
$P$ queries produces, after postselection, a normalized state
$|\psi_0\rangle_{OAZB}\in\mathcal L_P(AZB)$: the initial oracle state
has size zero; each query increases the size by at most one
by~\cref{fact:function-query-level,fact:permutation-query-level},
and postselection on the algorithm's registers preserves membership in
this subspace. Its reduced state
$\sigma_{OA}:=\operatorname{Tr}_{ZB}|\psi_0\rangle\langle\psi_0|$
therefore satisfies $\sigma_{OA}=\Pi_{\le P}\sigma_{OA}\Pi_{\le P}$.
The bit-fixing advantage is consequently bounded by
\[
    \bigl|\operatorname{Tr}(\sigma_{OA}M_{\mathcal E})\bigr|
    =\bigl|\operatorname{Tr}(\sigma_{OA}\Pi_{\le P}M_{\mathcal E}\Pi_{\le P})\bigr|
    \le\bigl\|\Pi_{\le P}M_{\mathcal E}\Pi_{\le P}\bigr\|_{\mathrm{op}}.
\]

Conversely, we show that every unit vector
$|\psi\rangle_{OA}\in\mathcal L_P(A)$ can be prepared using at most
$P$ queries and postselection with positive probability. We give the
construction for both database hierarchies.
\paragraph{Case 1: Function.}
    Let $C_{\mathrm F}:=\sum_{j=0}^{P}\binom Nj(M-1)^j$ be the number of all possible databases of size at most $P$.
    Adjoining the register $R$ to the purified oracle state $|\Omega_{\mathcal F}\rangle = |\phi_{0}\rangle$, prepare
    \[
    |\Psi^{\mathrm F}_0\rangle_{OR}:=\frac{1}{\sqrt{C_{\mathrm F}}}\sum_{\substack{D\in\mathbb Z_M^N\\|D|\le P}}|\phi_0\rangle_O|D\rangle_R.
    \]
    For each $D$, write $\operatorname{supp}(D)=\{x_1<\cdots<x_k\}$, where $k=|D|$. For $1\le i\le k$, coherently compute $(x_i,D(x_i))$ from $D$, apply the phase oracle, and uncompute the query registers. If $k<P$, each of the remaining queries is made at a fixed point with phase label $0$. Since
    \[
    \mathcal O\bigl(|\phi_{D'}\rangle_O|x,u\rangle\bigr)=|\phi_{D'+u e_x}\rangle_O|x,u\rangle,
    \]
    the resulting state is
    \[
    |\Psi^{\mathrm F}_P\rangle_{OR}=\frac{1}{\sqrt{C_{\mathrm F}}}\sum_{\substack{D\in\mathbb Z_M^N\\|D|\le P}}|\phi_D\rangle_O|D\rangle_R.
    \]

\paragraph{Case 2: Permutation.}
    First suppose that $P\le N-1$. Prepare
    \[
    |\Psi^{\mathrm{perm}}_0\rangle_{OR}:=\frac{1}{\sqrt{\binom NP\,N!}}\sum_{X\in\binom{[N]}P}\sum_{H\in\mathcal F^{\mathrm{perm}}}|H\rangle_O|X,0^P\rangle_R.
    \]
    For $X=\{x_1<\cdots<x_P\}$, query $x_1,\ldots,x_P$ in the standard oracle form, implemented using one phase-oracle query per point by \cref{lem:std oracle and phase oracle}, and store their images in $R$. This gives
    \[
    |\Psi^{\mathrm{perm}}_P\rangle_{OR}=\frac{1}{\sqrt{\binom NP\,N!}}\sum_{X\in\binom{[N]}P}\sum_{H\in\mathcal F^{\mathrm{perm}}}|H\rangle_O|H|_X\rangle_R.
    \]
    Since $\sum_{H\supseteq\alpha}|H\rangle_O=\sqrt{(N-P)!}\,|v_\alpha\rangle_O$, we have
    \[
    |\Psi^{\mathrm{perm}}_P\rangle_{OR}=\sqrt{\frac{(N-P)!}{\binom NP\,N!}}\sum_{\substack{\alpha\text{ database}\\|\alpha|=P}}|v_\alpha\rangle_O|\alpha\rangle_R.
    \]
    By \cref{sec:Random Permutations: The Partial-Assignment Filtration}, the vectors $|v_\alpha\rangle$ with $|\alpha|=P$ span $\mathcal L_P^{\mathrm{perm}}$. 
    For $P>N-1$, use the construction with $P=N-1$, since $\mathcal L_P^{\mathrm{perm}}=\mathcal L_{N-1}^{\mathrm{perm}}$.
    
\paragraph{Constructing the unitary.} 
In either case, the resulting state has the form
    \[
    |\Psi_P\rangle_{OR}=\gamma\sum_a|w_a\rangle_O|a\rangle_R,
    \]
    where $\gamma>0$ is the corresponding coefficient above, and $(a,|w_a\rangle)$ denotes $(D,|\phi_D\rangle)$ in the function case and $(\alpha,|v_\alpha\rangle)$ in the permutation case. Since the vectors $|w_a\rangle$ span $\mathcal L_P$, choose vectors $|z_a\rangle_A$ such that
    \[
    |\psi\rangle_{OA}=\sum_a|w_a\rangle_O|z_a\rangle_A.
    \]
    Define an operator $F$ on $RA$ by
    \[
    F:=\sum_a|0\rangle\langle a|_R\otimes|z_a\rangle\langle0|_A.
    \]
    Then
    \[
    (I_O\otimes F)(|\Psi_P\rangle_{OR}|0\rangle_A)=\gamma|\psi\rangle_{OA}|0\rangle_R.
    \]

    Take $c:=1/(1+||F||_{op})>0$ and set $K:=cF$. Then $||K||_{op}\le 1$. Let $K=U\Sigma V^{\dagger}$ be its singular value decomposition, and choose an orthonormal basis $\{|j\rangle\}$ of $RA$ such that $\Sigma|j\rangle=s_j|j\rangle$, where $0\le s_j\le1$. We implement $K$ by postselection.
    
    Let $B$ be a flag qubit. For an arbitrary state $|r\rangle_{RA}$, start with $|r\rangle_{RA}|0\rangle_B$. Apply $V^{\dagger}$:
    \[
        |r\rangle_{RA}|0\rangle_B\mapsto \sum_{j}r_j |j\rangle|0\rangle_{B},\qquad r_{j}:=\langle j|V^{\dagger}|r\rangle.
    \]
    Apply the rotation:
    \[
        |j\rangle|0\rangle_{B} \mapsto |j\rangle (s_{j}|0\rangle_B+\sqrt{1-s_{j}^2}|1\rangle_B).
    \]
    After applying $U$, the resulting state is
    \[
    \begin{aligned}
    &\sum_j r_jU|j\rangle
      \left(s_j|0\rangle_B+\sqrt{1-s_j^2}|1\rangle_B\right)\\
    &\quad=
    U\Sigma V^\dagger|r\rangle|0\rangle_B
    +
    U\sqrt{I-\Sigma^2}V^\dagger|r\rangle|1\rangle_B\\
    &\quad=
    K|r\rangle|0\rangle_B+L|r\rangle|1\rangle_B
    \end{aligned}
    \]
    where $L:=U\sqrt{I-\Sigma^2}V^\dagger$. Let this unitary be $\mathcal U$.
    It follows that
    \[
    \begin{aligned}
    &(I_O\otimes\mathcal U)
    \bigl(|\Psi_P\rangle_{OR}|0\rangle_A|0\rangle_B\bigr)\\
    &\quad=
    c\gamma
    |\psi\rangle_{OA}|0\rangle_R|0\rangle_B+
    (I_O\otimes L)
    \bigl(|\Psi_P\rangle_{OR}|0\rangle_A\bigr)|1\rangle_B.
    \end{aligned}
    \]
    Thus, postselecting on $B=0$, which occurs with probability $c^2\gamma^2>0$, leaves $R$ in the fixed state $|0\rangle_R$ and the remaining registers in the state $|\psi\rangle_{OA}$.

For each unit vector $|\psi\rangle\in\mathcal L_P(A)$, the construction
gives a valid bit-fixing adversary with online algorithm $\mathcal A$
and advantage $|\langle\psi|M_{\mathcal E}|\psi\rangle|$.
The maximum of this expression is the database-size-restricted operator norm, attained
by an eigenvector whose eigenvalue has the largest absolute value.
Together with the upper bound, this proves the characterization for each
fixed online algorithm. Taking the supremum over online algorithms
proves~\eqref{eq:bf-operator-characterization}.
\end{proof}

\subsection{Proof of \texorpdfstring{\cref{thm:generic-public-salting}}{thm~generic-public-salting}}
\label{subsec:generic-salting}

\begin{proof}[Proof of \cref{thm:generic-public-salting}]
    Fix a challenge salt $s$ and a unit vector $|\psi\rangle\in \mathcal{L}_{P}^{(K)}$. Define $|\psi_{0}\rangle:=\Pi_{s,=0}|\psi\rangle$ and $|\psi_{1}\rangle:=(I-\Pi_{s,=0})|\psi\rangle$. Let $M_{\mathcal E_s}$ be the fixed-salt purified game operator of the salted game.

    For a search game, $0\preceq M_{\mathcal E_s}\preceq I$, and hence
    $M_{\mathcal E_s}^{1/2}$ is well defined.
    By \hyperref[item:image in level 0]
{\cref*{lem:fixed-salt-reduction}\ref*{item:image in level 0}} and the uniform-security assumption,
    \begin{equation}
        \lVert M_{\mathcal{E}_s}^{1/2}|\psi_{0}\rangle\rVert_{2}^{2}=\langle \psi_{0}|M_{\mathcal{E}_s}|\psi_{0}\rangle \le \lVert |\psi_{0}\rangle\rVert_{2}^{2}\nu_{\mathcal{G}}(T) \le \nu_{\mathcal{G}}(T)
        \label{eq:bound psi_0}
    \end{equation}
    From $M_{\mathcal{E}_s}\preceq I$ and \cref{eq:lambda_s and Pi_s=0}, we have
    \begin{align}
        \lVert M_{\mathcal{E}_s}^{1/2}|\psi_{1}\rangle\rVert_{2}^{2}&=\langle \psi_{1}|M_{\mathcal{E}_s}|\psi_{1}\rangle \notag \\
        &\le \langle \psi_{1}|I|\psi_{1}\rangle = \langle \psi|(I-\Pi_{s,=0})|\psi\rangle \notag \\
        &\le \langle \psi|\Lambda_{s}|\psi\rangle.
        \label{eq:bound psi_1}
    \end{align}
    Combining \eqref{eq:bound psi_0} and \eqref{eq:bound psi_1}, we get
    \begin{align}
        \langle \psi|M_{\mathcal{E}_s}|\psi\rangle &= \lVert M_{\mathcal{E}_s}^{1/2}|\psi\rangle\rVert_{2}^{2} \notag\\
        &\le (\lVert M_{\mathcal{E}_s}^{1/2}|\psi_{0}\rangle\rVert_{2}+\lVert M_{\mathcal{E}_s}^{1/2}|\psi_{1}\rangle\rVert_{2})^{2} \notag\\
        &\le  \left(\sqrt{\nu_{\mathcal{G}}(T)} + \sqrt{\langle \psi|\Lambda_{s}|\psi\rangle}\right)^2
        \label{eq:nu before averaging}
    \end{align}
    Averaging \eqref{eq:nu before averaging} over $s \in [K]$ and using \eqref{eq:averaged lambda inequality} and Jensen's inequality gives
    \[
        \langle \psi|M_{\mathcal E^{(K)}}|\psi\rangle \le \left(\sqrt{\nu_{\mathcal{G}}(T)} + \sqrt{P/K}\right)^2.
    \]
    Taking the supremum over unit vectors \(|\psi\rangle\in\mathcal L_P^{(K)}\) gives, for any $P\ge 0$,
    \[
        \bigl\|\Pi_{\le P}^{(K)}M_{\mathcal E^{(K)}}\Pi_{\le P}^{(K)}\bigr\|_{\mathrm{op}}
        \le \left(\sqrt{\nu_{\mathcal{G}}(T)} + \sqrt{P/K}\right)^2.
    \]
    
    This bound holds for every salted online algorithm, so
    applying \cref{cor:qai-from-level-bounds} with $P=Sr$ gives
    \[
        \delta_{\mathcal G_{\mathsf S}}^{QAI}(S,T)
        =
        O\left(
            \left(
                \sqrt{\nu_{\mathcal G}(T)}
                +
                \sqrt{\frac{Sr}{K}}
            \right)^2
        \right)
        =
        O\left(
            \nu_{\mathcal G}(T)+\frac{Sr}{K}
        \right).
    \]
    This proves $(i)$.

   For a decision game, \(M_{\mathcal E_s}\) need not be positive, so the preceding positive-square-root argument does not apply. \hyperref[item:image in level 0]
{\cref*{lem:fixed-salt-reduction}\ref*{item:image in level 0}}, the hypothesis of (ii), and \eqref{eq:lambda_s and Pi_s=0} give
    \begin{equation}
        |\langle \psi_{0}|M_{\mathcal{E}_s}|\psi_{0}\rangle| \le \epsilon_{\mathcal{G}}(T),\qquad \lVert |\psi_{1}\rangle\rVert_{2}^{2} \le \langle \psi|\Lambda_{s}|\psi\rangle.
        \label{eq:decision bound}
    \end{equation}
    By the Cauchy–Schwarz inequality and $||M_{\mathcal{E}_s}||_{op}\le 1$, 
    \begin{equation}
        |\langle \psi_{0}|M_{\mathcal{E}_s}|\psi_{1}\rangle| \le || M_{\mathcal{E}_s}||_{op}\lVert |\psi_{0}\rangle\rVert_{2} \lVert |\psi_{1}\rangle\rVert_{2} \le \lVert |\psi_{1}\rangle\rVert_{2}.
        \label{eq:cauchy}
    \end{equation}
    Combining \eqref{eq:cauchy} and \eqref{eq:decision bound} yields
    \begin{equation}
        |\langle \psi|M_{\mathcal{E}_s}|\psi\rangle| \le \epsilon_{\mathcal{G}}(T) + 2\sqrt{\langle \psi|\Lambda_{s}|\psi\rangle}+\langle \psi|\Lambda_{s}|\psi\rangle.
        \label{eq:epsilon before averaging}
    \end{equation}
    Averaging \eqref{eq:epsilon before averaging} over $s\in[K]$ and using the triangle inequality, \eqref{eq:averaged lambda inequality}, and Jensen’s inequality gives
    \[
        |\langle \psi|M_{\mathcal{E}^{(K)}}|\psi\rangle| \le \epsilon_{\mathcal{G}}(T) + 2\sqrt{P/K}+P/K. 
    \]
    As $|\psi\rangle$ and $P\ge 0$ are arbitrary, for any $P \ge 0$, 
    \[
        \bigl\|\Pi_{\le P}^{(K)}M_{\mathcal E^{(K)}}\Pi_{\le P}^{(K)}\bigr\|_{\mathrm{op}}
        \le \epsilon_{\mathcal{G}}(T) + 2\sqrt{P/K}+P/K.
    \]
    This bound holds for every salted online algorithm, so
    applying \cref{cor:qai-from-level-bounds} with $P=Sr$ gives
    \[
        \epsilon_{\mathcal G_{\mathsf S}}^{QAI}(S,T)
        =
        O\left(
            \epsilon_{\mathcal G}(T)
            +
            \sqrt{\frac{Sr}{K}}
            +
            \frac{Sr}{K}
        \right).
    \]
    If $Sr/K\le1$, then $Sr/K\le\sqrt{Sr/K}$. If $Sr/K>1$, the claimed
    bound follows from the trivial bound on the decision advantage.
    Therefore,
    \[
        \epsilon_{\mathcal G_{\mathsf S}}^{QAI}(S,T)
        =
        O\left(
            \epsilon_{\mathcal G}(T)
            +
            \sqrt{\frac{Sr}{K}}
        \right).
    \]
    This proves $(ii)$.
\end{proof}

\fi

\section{Missing Proofs}
\label{sec:Missing Proofs}

In this section,  we present proofs omitted from earlier sections. 
Throughout this section, we fix an orthonormal basis $\{|z\rangle\}_Z$ of $\mathcal H_Z$ for the internal register $Z$.

\subsection{Proof for the Uniform one-point reprogramming lemma}
\label{sec:random_reprogramming}

We use the following classical polynomial inequality.

\begin{lemma}[Markov's polynomial inequality~{\cite{Markov1890}}]
    \label{lem:markov-polynomial}
    Every real polynomial $p$ of degree at most $d$ satisfies
    \[
        \max_{s\in[-1,1]} |p'(s)|
        \le
        d^2 \max_{s\in[-1,1]} |p(s)|.
    \]
\end{lemma}
\begin{proof}[Proof of \cref{lem:uniform-one-point-reprogramming}]
    Purify the circuit's initial state and defer all intermediate
    measurements, so that the computation is unitary until a final
    projective measurement. These transformations can be chosen
    independently of $x$.

    For $v\in\mathbb{Z}_M$, let $S_v$ be the shift operator defined by
    $S_v|u\rangle=|u+v\rangle$, where addition is modulo $M$.
    For each $z\in\{0,1\}^N$, define the oracle query
    \[
        U_z
        :=
        \sum_{w\in[N]}
        |w\rangle\langle w|
        \otimes
        \bigl(
            (1-z_w)S_{H(w)}+z_wS_h
        \bigr).
    \]
    Thus, $U_0$ is the query to $H$, while $U_{e_x}$ is the query to
    $H_{x\to h}$, where $e_x$ denotes the Boolean vector whose only
    nonzero coordinate is $x$. We run the same circuit with $U_z$
    in place of each oracle query, keeping the initial state,
    inter-query operations, and final measurement fixed for every $z$.

    Every matrix entry of $U_z$ is affine in $z$. Consequently, if
    we write the final state as
    \[
        |\Psi(z)\rangle
        =
        \sum_j q_j(z)|j\rangle
    \]
    in a fixed orthonormal basis, each amplitude $q_j(z)$ is a
    polynomial of degree at most $T$. Let $\Pi$ be the final
    accepting projector and define
    \[
        A(z):=\langle\Psi(z)|\Pi|\Psi(z)\rangle.
    \]
    The acceptance probability $A(z)$ is therefore a polynomial
    of degree at most $2T$. Multilinearizing using
    $z_w^2=z_w$ on $\{0,1\}^N$, we obtain
    \begin{equation}
        A(z)
        =
        \sum_{\substack{S\subseteq[N]\\|S|\le 2T}}
        c_S\prod_{w\in S}z_w,
        \label{eq:uniform-acceptance-multilinear}
    \end{equation}
    where the coefficients $c_S$ are real. Moreover,
    $0\le A(z)\le 1$ for every $z\in\{0,1\}^N$, and
    \[
        A(0)=a(\varnothing),
        \qquad
        A(e_x)=a(x).
    \]

    For each $t\in[0,N]$, let $Z_1(t),\ldots,Z_N(t)$ be independent
    $\operatorname{Bernoulli}(t/N)$ random variables, and write
    $Z(t):=(Z_1(t),\ldots,Z_N(t))$.
    Taking expectations in
    \eqref{eq:uniform-acceptance-multilinear} gives
    \begin{equation}
        q(t)
        :=
        \E[A(Z(t))]
        =
        \sum_{\substack{S\subseteq[N]\\|S|\le 2T}}
        c_S\left(\frac{t}{N}\right)^{|S|}.
        \label{eq:uniform-univariate-polynomial}
    \end{equation}
    Hence $q$ is a real polynomial of degree at most $2T$.
    Since $q(t)$ is an average of acceptance probabilities,
    \[
        0\le q(t)\le 1
        \qquad\text{for every }t\in[0,N].
    \]
    
    Differentiating
    \eqref{eq:uniform-univariate-polynomial} at zero leaves only
    the singleton terms:
    \[
        q'(0)
        =
        \frac{1}{N}\sum_{x\in[N]}c_{\{x\}}.
    \]
    On the other hand,
    \eqref{eq:uniform-acceptance-multilinear} gives
    $A(0)=c_\emptyset$ and
    $A(e_x)=c_\emptyset+c_{\{x\}}$. Therefore,
    \[
    q'(0)
        =
        \frac{1}{N}\sum_{x\in[N]}
        \bigl(A(e_x)-A(0)\bigr)
        =
        \E_{x\sample[N]}[a(x)]-a(\varnothing).
    \]

    Finally, define
    \[
        r(s)
        :=
        2q\left(\frac{N(s+1)}{2}\right)-1.
    \]
    Then $\deg r\le 2T$ and $|r(s)|\le 1$ on $[-1,1]$.
    By Lemma~\ref{lem:markov-polynomial},
    \[
        |r'(-1)|\le (2T)^2.
    \]
    Since $r'(-1)=Nq'(0)$, it follows that
    \[
        \left|
            \E_{x\sample[N]}[a(x)]-a(\varnothing)
        \right|
        =
        |q'(0)|
        \le
        \frac{4T^2}{N},
    \]
    as claimed.
\end{proof}

\subsection{Diffuse Bounds for Decision Games}
\label{subsec:pf of lem:back to psi decision}

For completeness, we give the decision-game analogue of
\cref{lem:back to psi search}.

\begin{lemma}
    Fix an integer $T\ge0$. For a decision game $\mathcal G=(\mathcal C)$ with $C=(\Samp,\Query,\Ver)$, put $r:=\max\{1,2(T_{\Samp}+T+T_{\Ver})\}$. Suppose there exist $\tau_T\ge0$, a universal constant $\kappa>0$, and a constant $\Delta>0$ determined by $\mathcal G$ such that, for every $T$-query online algorithm and every integer $P\ge0$,
    \begin{equation}
        ||\Pi_{\le P}M_{\mathcal E}\Pi_{\le P}||_{op}\le \tau_T + \kappa\sqrt{\frac{P}{\Delta}}.
        \label{eq:BF bound decision}
    \end{equation}
    Then $\mathcal G$ is $(1/2,\Delta/(1+\lceil\log_2r\rceil))$-diffuse at query budget $T$ with parameters
    \[
        \tau_T':=4\tau_T,\qquad \kappa':=12\kappa.
    \]
    \label{lem:back to psi decision}
\end{lemma}

\begin{proof}

We simply write $r:=\max\{1,2(T_{\Samp}+T+T_{\Ver})\}$. Define the block projector, block vector, and block expectation for $m\ge 0$,
\[
    \Pi_{m}^{\mathsf{B}}:=\sum_{j=mr}^{(m+1)r-1}\Pi_{=j},\qquad |\psi_{m}\rangle:=\Pi_{m}^{\mathsf{B}}|\psi\rangle, \qquad d_{m}:=\langle \psi_m|M_{\mathcal{E}}|\psi_m\rangle.
\]
We adopt the convention \(\Pi_{=j}=0\) for \(j>N\).
As $M_{\mathcal{E}}$ increases the database size by at most $r$,
\[
    \Pi_{m}^{\mathsf{B}}M_{\mathcal{E}}\Pi_{n}^{\mathsf{B}}=0
\]
for $|m-n|\ge 2$. Thus we have
\[
    \langle\psi|M_{\mathcal E}|\psi \rangle = \sum_{m\ge 0}\langle \psi_{m}|M_{\mathcal E}|\psi_{m}\rangle + 2\sum_{m\ge 0}\Re(\langle \psi_{m}|M_{\mathcal E}|\psi_{m+1}\rangle).
\]
Here, we used $M_{\mathcal E}=M_{\mathcal E}^{\dagger}$. As $|\psi_{m}\rangle \in \mathcal L_{(m+1)r-1}$ and $|\psi_{m+1}\rangle \in \mathcal L_{(m+2)r-1}$,
\[
    \begin{aligned}
        &|\langle \psi_{m}|M_{\mathcal E}|\psi_{m}\rangle| \le \left(\tau_{T}+\kappa\sqrt{\frac{(m+1)r-1}{\Delta}}\right)||\psi_{m}||_2^2 \\
        &|\langle \psi_{m}|M_{\mathcal E}|\psi_{m+1}\rangle| \le \left(\tau_{T}+\kappa\sqrt{\frac{(m+2)r-1}{\Delta}}\right)||\psi_{m}||_2||\psi_{m+1}||_2.
    \end{aligned}
\]
We bound $\langle\psi|M_{\mathcal E}|\psi \rangle-\langle\psi_0|M_{\mathcal E}|\psi_0 \rangle$.
We have
\begin{align}
    \sum_{m\ge 1}((m+1)r-1)||\psi_m||_{2}^2 &=\sum_{m\ge 1}\sum_{j=mr}^{(m+1)r-1}((m+1)r-1)||\Pi_{=j}\psi||_{2}^2  \notag\\
    &\le \sum_{m\ge 1}\sum_{j=mr}^{(m+1)r-1}2j||\Pi_{=j}\psi||_{2}^2 \notag\\
    &\le 2\langle \psi | \Lambda | \psi\rangle.
    \label{eq:m r and Lambda}
\end{align}
By Cauchy–Schwarz and \eqref{eq:m r and Lambda},
\begin{align}
    \sum_{m\ge 1}\sqrt{(m+1)r-1}||\psi_{m}||_2^2 &\le \sqrt{\sum_{m\ge 1}((m+1)r-1)||\psi_m||_2^2} \sqrt{\sum_{m\ge 1}||\psi_{m}||_2^2} \notag \\
    &\le \sqrt{2\langle \psi | \Lambda | \psi\rangle}.
    \label{eq:bounded by 2Lambda 1}
\end{align}
Similarly,
\begin{align}
    \sum_{m\ge 0}\sqrt{(m+2)r-1}||\psi_{m}||_2 ||\psi_{m+1}||_2 &\le \sqrt{\sum_{m\ge 0}||\psi_m||_2^2} \sqrt{\sum_{m\ge 0}((m+2)r-1)||\psi_{m+1}||_2^2} \notag \\
    &\le \sqrt{\sum_{m\ge 1}((m+1)r-1)||\psi_{m}||_2^2} \notag  \\
    &\le \sqrt{2\langle \psi | \Lambda | \psi\rangle}.
    \label{eq:bounded by 2Lambda 2}
\end{align}
By Cauchy–Schwarz,
\begin{equation}
    \sum_{m\ge 0}||\psi_{m}||_{2}||\psi_{m+1}||_2 \le \sqrt{\sum_{m\ge 0}||\psi_{m}||_{2}^2} \sqrt{\sum_{m\ge 0}||\psi_{m+1}||_2^2}\le 1.
    \label{eq:psi psi}
\end{equation}
Combining \eqref{eq:bounded by 2Lambda 1,eq:bounded by 2Lambda 2,eq:psi psi}, we have
\[
    |\langle\psi|M_{\mathcal E}|\psi \rangle-\langle\psi_0|M_{\mathcal E}|\psi_0 \rangle| \le (3-||\psi_{0}||_2^2)\tau_T + 3\kappa\sqrt{\frac{2\langle \psi | \Lambda | \psi\rangle}{\Delta}}.
\]

It remains to bound $|d_0|$. Set
\[
q:=\max\left\{1,\left\lceil\frac{\Delta\tau_T^2}{\kappa^2}\right\rceil\right\},\qquad J:=\max\left\{0,\left\lceil\log_2(r/q)\right\rceil\right\}.
\]
Define
\[
|\phi_0\rangle:=\Pi_{\le q-1}|\psi_0\rangle.
\]
For $1\le j\le J$, define
\[
q_j:=\min\{2^jq-1,r-1\},\qquad |\phi_j\rangle:=\sum_{k=2^{j-1}q}^{q_j}\Pi_{=k}|\psi_0\rangle.
\]
These vectors are mutually orthogonal and satisfy
\[
|\psi_0\rangle=\sum_{j=0}^{J}|\phi_j\rangle.
\]
If $q\ge r$, then $J=0$ and $|\phi_0\rangle=|\psi_0\rangle$. Since $q\ge1$, we have $J\le\lceil\log_2r\rceil$.

The choice of $q$ gives
\[
\kappa\sqrt{\frac{q-1}{\Delta}}\le\tau_T.
\]
As $|\phi_0\rangle\in\mathcal L_{q-1}$, it follows that
\begin{align*}
|\langle\phi_0|M_{\mathcal E}|\phi_0\rangle|
&\le\left(\tau_T+\kappa\sqrt{\frac{q-1}{\Delta}}\right)\|\phi_0\|_2^2 \notag\\
&\le2\tau_T\|\phi_0\|_2^2.
\end{align*}
For $1\le j\le J$, we have $q_j\ge q\ge \Delta\tau_T^2/\kappa^2$, and hence
\[
\tau_T+\kappa\sqrt{\frac{q_j}{\Delta}}\le2\kappa\sqrt{\frac{q_j}{\Delta}}.
\]
Both $|\phi_j\rangle$ and $\sum_{i=0}^{j-1}|\phi_i\rangle$ belong to $\mathcal L_{q_j}$. Moreover, orthogonality gives
\[
\|\phi_j\|_2\le\|\psi_0\|_2,\qquad \left\|\sum_{i=0}^{j-1}\phi_i\right\|_2\le\|\psi_0\|_2.
\]
Expanding $d_0$ in this decomposition gives
\[
d_0=\sum_{j=0}^{J}\langle\phi_j|M_{\mathcal E}|\phi_j\rangle+2\sum_{j=1}^{J}\Re\left(\sum_{i=0}^{j-1}\langle\phi_i|M_{\mathcal E}|\phi_j\rangle\right).
\]
For $1\le j\le J$, both $\sum_{i=0}^{j-1}|\phi_i\rangle$ and $|\phi_j\rangle$ belong to $\mathcal L_{q_j}$. Hence,
\begin{align*}
\left|\sum_{i=0}^{j-1}\langle\phi_i|M_{\mathcal E}|\phi_j\rangle\right|
&=\left|\left\langle\sum_{i=0}^{j-1}\phi_i\middle|\Pi_{\le q_j}M_{\mathcal E}\Pi_{\le q_j}\middle|\phi_j\right\rangle\right| \notag\\
&\le\|\Pi_{\le q_j}M_{\mathcal E}\Pi_{\le q_j}\|_{\mathrm{op}}\left\|\sum_{i=0}^{j-1}\phi_i\right\|_2\|\phi_j\|_2 \notag\\
&\le2\kappa\sqrt{\frac{q_j}{\Delta}}\left\|\sum_{i=0}^{j-1}\phi_i\right\|_2\|\phi_j\|_2.
\end{align*}
Similarly,
\[
|\langle\phi_j|M_{\mathcal E}|\phi_j\rangle|\le2\kappa\sqrt{\frac{q_j}{\Delta}}\|\phi_j\|_2^2.
\]
Applying the preceding bounds, we obtain
\begin{align*}
|d_0|
&\le2\tau_T\|\phi_0\|_2^2+2\kappa\sum_{j=1}^{J}\sqrt{\frac{q_j}{\Delta}}\left(\|\phi_j\|_2^2+2\left\|\sum_{i=0}^{j-1}\phi_i\right\|_2\|\phi_j\|_2\right) \notag\\
&\le2\tau_T\|\psi_0\|_2^2+\frac{6\kappa\|\psi_0\|_2}{\sqrt{\Delta}}\sum_{j=1}^{J}\sqrt{q_j}\|\phi_j\|_2.
\end{align*}
For every $k$ in the $j$th interval, $q_j<2(2^{j-1}q)\le2k$. Therefore,
\begin{align*}
\sum_{j=1}^{J}q_j\|\phi_j\|_2^2
&=\sum_{j=1}^{J}\sum_{k=2^{j-1}q}^{q_j}q_j\|\Pi_{=k}\psi_0\|_2^2 \notag\\
&\le2\sum_{j=1}^{J}\sum_{k=2^{j-1}q}^{q_j}k\|\Pi_{=k}\psi_0\|_2^2 \notag\\
&\le2\langle\psi_0|\Lambda|\psi_0\rangle.
\end{align*}
By Cauchy--Schwarz,
\begin{align*}
\sum_{j=1}^{J}\sqrt{q_j}\|\phi_j\|_2
&\le\sqrt J\sqrt{\sum_{j=1}^{J}q_j\|\phi_j\|_2^2} \notag\\
&\le\sqrt{2J\langle\psi_0|\Lambda|\psi_0\rangle}.
\end{align*}
Thus,
\[
|d_0|\le2\tau_T\|\psi_0\|_2^2+6\kappa\|\psi_0\|_2\sqrt{\frac{2J\langle\psi_0|\Lambda|\psi_0\rangle}{\Delta}}.
\]
Combining this with the preceding bound on the terms outside the first block gives
\begin{align*}
|\langle\psi|M_{\mathcal E}|\psi\rangle|
&\le(3+\|\psi_0\|_2^2)\tau_T+3\kappa\sqrt{\frac{2\langle\psi|\Lambda|\psi\rangle}{\Delta}}+6\kappa\|\psi_0\|_2\sqrt{\frac{2J\langle\psi_0|\Lambda|\psi_0\rangle}{\Delta}} \notag\\
&\le4\tau_T+\sqrt2(3+6\sqrt J)\kappa\sqrt{\frac{\langle\psi|\Lambda|\psi\rangle}{\Delta}} \notag\\
&\le4\tau_T+12\kappa\sqrt{\frac{(1+\lceil\log_2r\rceil)\langle\psi|\Lambda|\psi\rangle}{\Delta}}.
\end{align*}
The second inequality uses $\|\psi_0\|_2\le1$ and $\langle\psi_0|\Lambda|\psi_0\rangle\le\langle\psi|\Lambda|\psi\rangle$. The last inequality follows from $J\le\lceil\log_2r\rceil$ and
\[
\sqrt2(3+6\sqrt J)\le3\sqrt{10}\sqrt{1+J}\le12\sqrt{1+J}.\qedhere
\]

\end{proof}

\subsection{Proof of \texorpdfstring{\cref{lem:back to psi search}}{Lemma~back to psi search}}
\label{subsec:pf of lem:back to psi search}

\begin{proof}

We simply write $r:=\max\{1,2(T_{\Samp}+T+T_{\Ver})\}$. Define the block projector for $m\ge 0$,
\[
    \Pi_{m}^{\mathsf{B}}:=\sum_{j=mr}^{(m+1)r-1}\Pi_{=j}.
\]
We adopt the convention \(\Pi_{=j}=0\) for \(j>N\).
As $M_{\mathcal{E}}$ increases the database size by at most $r$,
\begin{equation}
    \Pi_{m}^{\mathsf{B}}M_{\mathcal{E}}\Pi_{n}^{\mathsf{B}}=0
    \label{eq:Pi_m and Pi_n}
\end{equation}
for $|m-n|\ge 2$. Define the odd and even parts
\[
    |\psi_{e}\rangle:=\sum_{2|m}\Pi_{m}^{\mathsf{B}}|\psi\rangle,\qquad |\psi_{o}\rangle:=\sum_{2\nmid m}\Pi_{m}^{\mathsf{B}}|\psi\rangle.
\]
Define
\[
d_{e}:=\langle \psi_{e}|M_{\mathcal{E}}|\psi_{e}\rangle, \quad d_{o}:=\langle \psi_{o}|M_{\mathcal{E}}|\psi_{o}\rangle, \quad
d_{m}:=\langle \psi|\Pi_{m}^{\mathsf{B}}M_{\mathcal{E}}\Pi_{m}^{\mathsf{B}}|\psi\rangle.
\]
Then we have
\begin{equation}
    \langle \psi | M_{\mathcal{E}} | \psi\rangle \le d_{e}+d_{o}+|\langle \psi_{e}|M_{\mathcal{E}}|\psi_{o}\rangle|+|\langle \psi_{o}|M_{\mathcal{E}}|\psi_{e}\rangle|.
    \label{eq:d_e and d_o}
\end{equation}
By \eqref{eq:Pi_m and Pi_n}, $d_e$ and $d_o$ can be written as
\[
    d_e=\sum_{2|m}d_m,\qquad d_o=\sum_{2\nmid m}d_m.
\]
Since $0\preceq M_{\mathcal E}$, the Cauchy--Schwarz inequality gives, for any vectors $|u\rangle$ and $|v\rangle$,
\[
    |\langle u|M_{\mathcal{E}}|v\rangle|\le \sqrt{\langle u|M_{\mathcal{E}}|u\rangle\langle v|M_{\mathcal{E}}|v\rangle}.
\]
This gives
\[
    |\langle\psi_{e}|M_{\mathcal{E}}|\psi_{o}\rangle| \le \sqrt{d_ed_o}.
\]
Hence, \eqref{eq:d_e and d_o} and $2\sqrt{ab}\le a+b$ for nonnegative $a,b$ yield
\begin{equation}
    \langle \psi | M_{\mathcal{E}} | \psi\rangle \le d_e+d_o+2\sqrt{d_ed_o}\le 2(d_e+d_o)\le 2\sum_{m}d_m.
    \label{eq:sum d_m}
\end{equation}
Applying \eqref{eq:BF bound search} to $\Pi_{m}^{\mathsf{B}}|\psi\rangle \in \mathcal L_{(m+1)r-1}$, we get
\[
    d_m \le \tau_T||\Pi_{m}^{\mathsf{B}}|\psi\rangle||_{2}^2 + \kappa\frac{(m+1)r-1}{\Delta}||\Pi_{m}^{\mathsf{B}}|\psi\rangle||_{2}^2.
\]
The equality
\[
    \begin{aligned}
    \Pi_{m}^{\mathsf{B}}\Lambda\Pi_{m}^{\mathsf{B}}+(r-1)\Pi_{m}^{\mathsf{B}}-((m+1)r-1)\Pi_{m}^{\mathsf{B}}  = \sum_{j=mr}^{(m+1)r-1}(j-mr)\Pi_{=j}\succeq 0
    \end{aligned}
\]
gives
\[
    ((m+1)r-1)||\Pi_{m}^{\mathsf{B}}\psi||_{2}^{2} \le \langle \psi|\Pi_{m}^{\mathsf{B}}\Lambda\Pi_{m}^{\mathsf{B}}|\psi\rangle + (r-1)||\Pi_{m}^{\mathsf{B}}\psi||_{2}^{2}.
\]
This implies
\[
    \begin{aligned}
    \sum_{m}d_m &\le \frac{\kappa}{\Delta}\sum_{m}((m+1)r-1)||\Pi_{m}^{\mathsf{B}}\psi||_{2}^{2} +\sum_{m} \tau_T||\Pi_{m}^{\mathsf{B}}\psi||_{2}^{2}  \\
    &\le \frac{\kappa}{\Delta}\sum_{m}\langle \psi|\Pi_{m}^{\mathsf{B}}\Lambda\Pi_{m}^{\mathsf{B}}|\psi\rangle +\sum_{m} \left(\tau_T+\frac{\kappa(r-1)}{\Delta}\right)||\Pi_{m}^{\mathsf{B}}\psi||_{2}^{2}\\
    &= \frac{\kappa \langle \psi|\Lambda|\psi\rangle }{\Delta}+\left(\tau_T+\frac{\kappa(r-1)}{\Delta}\right)
    \end{aligned}
\]
Applying it to \eqref{eq:sum d_m} gives
\[
    \langle \psi | M_{\mathcal{E}} | \psi\rangle \le \frac{2\kappa \langle \psi|\Lambda|\psi\rangle }{\Delta}+2\left(\tau_T+\frac{\kappa(r-1)}{\Delta}\right).\qedhere
\]

\end{proof}

\subsection{Proof of \texorpdfstring{\cref{lem:prg-degree bound}}{Lemma~PRG database-size bound}}
\label{subsec:pf of lem:prg-degree bound}


\begin{proof}
Consider an arbitrary unit vector $|\psi\rangle$ which can be expressed as
\[
  |\psi\rangle
=
\sum_{z,D}
\alpha_{z,D}
|\phi_D\rangle_O|z\rangle_Z,
  \quad
  \sum_{z,D}|{\alpha_{z,D}}|^2=1,
  \quad
  \langle\psi\rvert\Lambda\lvert\psi\rangle=\sum_{z,D}|D||{\alpha_{z,D}}|^2.
\]
In oracle-conditioned components,
\[
  \lvert\psi\rangle
  =\frac1{M^{N/2}}\sum_H\lvert H\rangle\lvert\psi(H)\rangle,
  \qquad
  \lvert\psi(H)\rangle
  =\sum_{z,D}\alpha_{z,D}
    \omega_M^{\langle D,H\rangle}\lvert z\rangle.
\]
Then $\E_H[||{\psi(H)}||_2^2]=1$.

For $x\in[N]$ and $h\in[M]$, recall the one-point reprogrammed
function $H_{x\to h}$.  Introduce an environment register \(\mathsf{env}\) with orthonormal basis \(\{|x,H,h\rangle\}\), a challenge register \(\mathsf{ch}\), and a fresh workspace \(W\), and define
\begin{align}
  \lvert\Gamma_{\mathsf{unif}}\rangle
  &:=\frac1{\sqrt{NM^{N+1}}}
    \sum_{x,H,h}
    \lvert x,H,h\rangle_{\mathsf{env}}\lvert h\rangle_{\mathsf{ch}}
    \lvert\psi(H)\rangle_{Z}\lvert0\rangle_{W},
  \label{eq:prg-gamma-unif}\\
  \lvert\Gamma_{\mathsf{img}}\rangle
  &:=\frac1{\sqrt{NM^{N+1}}}
    \sum_{x,H,h}
    \lvert x,H,h\rangle_{\mathsf{env}}\lvert h\rangle_{\mathsf{ch}}
    \lvert\psi(H_{x\to h})\rangle_{Z}\lvert0\rangle_{W}.
  \label{eq:prg-gamma-img}
\end{align}
We compute the two norms and the Euclidean distance between the states.  We have
$
 ||{\Gamma_{\mathsf{unif}}}||_2^2
 =\E_H[||{\psi(H)}||_2^2]=1.
$
For fixed $x$, the map
\begin{equation}
  (H,h)\longleftrightarrow
  (G=H_{x\to h},a=H(x))
  \label{eq:G-H bijection}
\end{equation}
is a bijection.

Two states $\lvert\Gamma_{\mathsf{unif}}\rangle, \lvert\Gamma_{\mathsf{img}}\rangle$ correspond to $|\phi_{1}\rangle,|\phi_{2}\rangle$ defined in \cite[Section~5.2]{CGLQ20}, respectively. From the proof of \cite[Lemma~5.10]{CGLQ20}, we have 
\begin{align}
  ||{\Gamma_{\unif}-\Gamma_{\img}}||_2^2
   \le \frac{4\langle\psi\rvert\Lambda\lvert\psi\rangle}{N}.
  \label{eq:prg-component-distance}
\end{align}

For an oracle $J$ and challenge $h$, let $V_h^J$ denote the
corresponding adversary unitary.
On the environment register storing $(x,H,h)$, define
\[
\begin{aligned}
    \widehat V_{\unif}
    &:=
    \sum_{x,H,h}
    |x,H,h\rangle\langle x,H,h|
    \otimes V_h^H,
    \\
    \widehat V_{\star}
    &:=
    \sum_{x,H,h}
    |x,H,h\rangle\langle x,H,h|
    \otimes V_h^{H_{x\to h}}.
\end{aligned}
\]
Define
\[
    |\Phi_{\unif}\rangle
    :=
    \widehat V_{\unif}|\Gamma_{\unif}\rangle,
    \qquad
    |\Phi_{\star}\rangle
    :=
    \widehat V_{\star}|\Gamma_{\unif}\rangle,
    \qquad
    |\Phi_{\img}\rangle
    :=
    \widehat V_{\star}|\Gamma_{\img}\rangle.
\]

Let $\Pi_0$ be the projector corresponding to adversary output $0$, and put
\[
    p_{\unif}
    :=
    \lVert\Pi_0\Phi_{\unif}\rVert_2^2,
    \qquad
    p_{\star}
    :=
    \lVert\Pi_0\Phi_{\star}\rVert_2^2,
    \qquad
    p_{\img}
    :=
    \lVert\Pi_0\Phi_{\img}\rVert_2^2.
\]
The last two experiments apply the same controlled unitary
$\widehat V_{\star}$.
Equivalently, they use the same POVM element
\[
    \widetilde E_{\star}
    :=
    \widehat V_{\star}^{\dagger}
    \Pi_0
    \widehat V_{\star},
    \qquad
    0\preceq\widetilde E_{\star}\preceq I.
\]

For unit vectors $|u\rangle,|v\rangle$ and an operator $0\preceq E\preceq I$,
applying Cauchy--Schwarz to $E-\tfrac12I$ gives
\[
\begin{aligned}
    \bigl|\langle u|E|u\rangle-\langle v|E|v\rangle\bigr|
    &\le 2\bigl\|E-\tfrac12I\bigr\|_{\mathrm{op}}\,\lVert u-v\rVert_2\\
    &\le\lVert u-v\rVert_2.
\end{aligned}
\]
Applying this to $\widetilde E_{\star}$ and the normalized states
$|\Gamma_{\img}\rangle,|\Gamma_{\unif}\rangle$, and using
\eqref{eq:prg-component-distance}, gives

\begin{equation}
    |p_{\img}-p_{\star}|
    \le
    \lVert\Gamma_{\img}-\Gamma_{\unif}\rVert_2
    \le
    \sqrt{
        \frac{
            4\langle\psi|\Lambda|\psi\rangle
        }{N}
    }.
    \label{eq:prg-img-star}
\end{equation}

Orthogonality of the environment labels gives
\begin{align}
  p_{\mathsf{unif}}
  &=\frac1{NM^{N+1}}\sum_{x,H,h}
    \langle\psi(H)\rvert E_h^H\lvert\psi(H)\rangle,
  \label{eq:prg-punif-detailed}\\
  p_\star
  &=\frac1{NM^{N+1}}\sum_{x,H,h}
    \langle\psi(H)|E_h^{H_{x\to h}}|\psi(H)\rangle.
  \label{eq:prg-pstar-detailed}
\end{align}
For each $H$, put $w_H:=\norm{\psi(H)}_2^2$ and normalize the nonzero oracle-conditioned component as $\lvert\eta_H\rangle=w_H^{-1/2}\lvert\psi(H)\rangle$.  For $w_H>0$, define
\[
    a_{H,h}(x)
    :=
    \langle\eta_H|
    E_h^{H_{x\to h}}
    |\eta_H\rangle,
    \qquad
    a_{H,h}(\varnothing)
    :=
    \langle\eta_H|
    E_h^H
    |\eta_H\rangle.
\]
When $w_H=0$, set both quantities equal to zero. 
\eqref{eq:prg-punif-detailed,eq:prg-pstar-detailed} give
\begin{align*}
 p_\star-p_{\mathsf{unif}}
 &=\frac1{M^{N+1}}\sum_{H,h}w_H
   \left(\frac1N\sum_xa_{H,h}(x)-a_{H,h}(\varnothing)\right).
\end{align*}
Condition on $H$ and $h$. Then the base oracle $H$, the replacement
value $h$, the normalized initial state $|\eta_H\rangle$, and the entire
$T$-query adversary for challenge $h$ are fixed, while only
$x\sample[N]$ varies. Hence \cref{lem:uniform-one-point-reprogramming} applies.
So, the absolute value of the parenthesis is at most $4T^2/N$.  Hence
\begin{equation}
 |p_\star-p_{\mathsf{unif}}|
 \le\frac{4T^2}{N}\frac1{M^{N+1}}\sum_{H,h}w_H
 =\frac{4T^2}{N}\E_H[w_H]
 =\frac{4T^2}{N}.
 \label{eq:star-unif PRG}
\end{equation}
The same change of variables gives
\begin{equation}
  \frac{1}{2}(p_{\mathsf{img}}-p_{\mathsf{unif}})
  =\langle\psi\rvert M_{\mathcal E_{\mathsf{PRG}}}\lvert\psi\rangle.
  \label{eq:prg-rayleigh-identity}
\end{equation}
Indeed, the uniform experiment simplifies to
\[
 p_{\mathsf{unif}}
 =\frac1{M^N}\sum_H
  \left\langle\psi(H)\left|
    \frac1M\sum_hE_h^H
  \right|\psi(H)\right\rangle.
\]
For the image experiment, apply the bijection
\eqref{eq:G-H bijection} and sum over the original value $a$:
\begin{align*}
 p_{\mathsf{img}}
 &=\frac1{NM^{N+1}}\sum_{x,H,h}
  \langle\psi(H_{x\to h})|E_h^{H_{x\to h}}|\psi(H_{x\to h})\rangle\\
 &=\frac1{NM^N}\sum_{x,G}
  \langle\psi(G)|E_{G(x)}^G|\psi(G)\rangle \\
 &=\frac1{M^N}\sum_{G}
  \left\langle\psi(G)\left|\frac{1}{N}\sum_{x}E_{G(x)}^G
  \right|\psi(G)\right\rangle
\end{align*}
Subtracting these two identities proves
\eqref{eq:prg-rayleigh-identity}.
Combining \eqref{eq:prg-img-star}, \eqref{eq:prg-rayleigh-identity} and
\eqref{eq:star-unif PRG} with the triangle inequality proves
\eqref{eq:prg-complete-degree-sensitive}. Since \(\|E_{\mathsf{PRG}}^H\|_{\mathsf{op}}\le1\) for every \(H\), we have \(\|M_{\mathcal E_{\mathsf{PRG}}}\|_{\mathsf{op}}\le1\), which gives the cap. Restricting to $\mathcal L_P$
and using $\langle \psi|\Lambda|\psi\rangle \le P$ on $\mathcal L_P$ proves
\eqref{eq:prg-degree-restricted}.
\end{proof}

\subsection{Proof of \texorpdfstring{\cref{lem:yb degree bound}}{Lemma~YB database-size bound}}
\label{subsec:pf of lem:yb degree bound}

We prove the lemma by using the method introduced in the proof of \cite[Lemma~6.2]{CGLQ20}.

\begin{proof}

Fix an arbitrary unit vector $|\psi\rangle$.
Define the projector on the oracle register, for $x\in [N]$,
\[
    F_{x}:=\sum_{D:D(x)=0}|\phi_D\rangle \langle \phi_D|_{O}.
\]
Then, $I-F_{x}$ corresponds to the projection $P_{x}$ defined in \cite[Section~6]{CGLQ20}. Define 
\[
    p_{x}:=||(I-F_{x})\psi||_{2}^2
\]
which corresponds to $p_{x}^{(0)}$ defined in \cite[Section~6]{CGLQ20}.
Calculation yields
\[
    \sum_{x\in [N]}p_{x}= \langle \psi|\Lambda |\psi\rangle.
\]
In the proof of \cite[Lemma~6.2]{CGLQ20}, they showed
\[
    \sqrt{\Pr[\mathsf{win}|x]} \le \sqrt{p_{x}}+\sqrt{1/2}.
\]
Squaring both sides and averaging over $x\in [N]$ gives
\[
    \Pr[\mathsf{win}]-\frac{1}{2}\le \frac{\langle \psi|\Lambda |\psi\rangle}{N}+ \sqrt{\frac{2\langle \psi|\Lambda |\psi\rangle}{N}} \le (1+\sqrt{2})\sqrt{\frac{\langle \psi|\Lambda |\psi\rangle}{N}}.
\]
Applying the same argument to the adversary with its output bit flipped bounds $1/2-\Pr[\mathsf{win}]$. Combining the two bounds yields
\[
    \left |\Pr[\mathsf {win}]-\frac{1}{2} \right | \le (1+\sqrt{2})\sqrt{\frac{\langle \psi|\Lambda |\psi\rangle}{N}}.
\]
$|\langle \psi |M_{\mathcal E_{\mathsf YB}}|\psi \rangle |=|\Pr[\mathsf{win}]-\frac{1}{2}|$ proves the lemma.

\end{proof}

\subsection{Proof of \texorpdfstring{\cref{lem:owf-degree bound}}{Lemma OWF database-size bound}}
\label{subsec:pf of lem:owf-degree bound}


We use the following result from \cite{Liu23}.
\begin{lemma}[{\cite[Lemma~4.10]{Liu23}}]
    There exists a universal constant $C_{\mathrm{BF}}>0$ such that, for all integers $P,T\ge0$,
    \[
    \delta_{\mathsf{OWF}}^{\mathsf{BF}}(P,T)\le\min\left\{1,C_{\mathsf{BF}}\frac{P+(T+1)^2}{\min\{N,M\}}\right\}.
    \]
    \label{lem:BF bound owf}
\end{lemma}
For $T\ge1$, this follows from the cited result after absorbing its implicit constant into $C_{\mathrm{BF}}$ and using $T^2\le (T+1)^2$; the case $T=0$ follows by monotonicity in $T$ from the case $T=1$, and the cap by $1$ is trivial.

We now prove \cref{lem:owf-degree bound}.
\begin{proof}

By \cref{thm:bf-operator-characterization} and \cref{lem:BF bound owf}, for any $P\ge 0$,
\[
    ||\Pi_{\le P}M_{\mathcal E_{\mathsf {OWF}}}\Pi_{\le P}||_{op}\le C_{\mathsf {BF}}\frac{P+(T+1)^{2}}{\min \{N,M\}}.
\]
Set
\[
    \tau_T:=C_{\mathsf {BF}}\frac{(T+1)^{2}}{\min \{N,M\}},\qquad \kappa:=C_{\mathsf{BF}},\qquad \Delta:=\min\{N,M\}.
\]
Applying \cref{lem:back to psi search} to the inequality above yields
\[
    \begin{aligned}
    \langle \psi|M_{\mathcal E_{\mathsf{OWF}}}|\psi\rangle &\le \frac{2C_{\mathsf {BF}} \langle \psi|\Lambda|\psi\rangle}{\min \{N,M\}}+\frac{2C_{\mathsf{BF}}(T+1)^{2}}{\min \{N,M\}}+\frac{2C_{\mathsf{BF}}(r_{\mathsf{OWF}}-1)}{\min\{N,M\}} \\
    &\le 8C_{\mathsf {BF}}\frac{\langle \psi|\Lambda|\psi\rangle+(T+1)^{2}}{\min\{N,M\}}.
    \end{aligned}
\]

\end{proof}

\subsection{Proof of \texorpdfstring{\cref{lem:perm-degree bound}}{Lemma~\ref{lem:perm-degree bound}}}
\label{subsec:pf of lem:perm-degree bound}

We recall the following result from \cite{ABC+26}.
\begin{lemma}[{\cite[Lemma~2]{ABC+26}}]
    There exists a universal constant $C_{\mathsf {BF}}>0$ such that, for all integers $P,T\ge0$,
    \[
    \delta_{\mathsf{PermInv}}^{\mathsf{BF}}(P,T)\le\min\left\{1,C_{\mathsf{BF}}\frac{P+(T+1)^2}{N}\right\}.
    \]
    \label{lem:BF bound PermInv}
\end{lemma}
For $T\ge1$, this follows from the cited result after absorbing its implicit constant into $C_{\mathrm{BF}}$ and using $T^2\le (T+1)^2$; the case $T=0$ follows by monotonicity in $T$ from the case $T=1$, and the cap by $1$ is trivial.

We now prove \cref{lem:perm-degree bound}.
\begin{proof}

By \cref{thm:bf-operator-characterization} and \cref{lem:BF bound PermInv}, for any $P\ge 0$,
\[
    ||\Pi_{\le P}M_{\mathcal E_{\mathsf{PermInv}}}\Pi_{\le P}||_{op}\le C_{\mathsf {BF}}\frac{P+(T+1)^{2}}{N}.
\]
Set
\[
    \tau_T:=C_{\mathsf {BF}}\frac{(T+1)^{2}}{N},\qquad \kappa:=C_{\mathsf{BF}},\qquad \Delta:=N.
\]
\cref{lem:back to psi search} with the inequality above yields
\[
    \begin{aligned}
    \langle \psi|M_{\mathcal E_{\mathsf{PermInv}}}|\psi\rangle &\le \frac{2C_{\mathsf {BF}} \langle \psi|\Lambda|\psi\rangle}{N}+\frac{2C_{\mathsf{BF}}(T+1)^{2}}{N}+\frac{2C_{\mathsf{BF}}(r_{\mathsf{PermInv}}-1)}{N} \\
    &\le 8C_{\mathsf {BF}}\frac{\langle \psi|\Lambda|\psi\rangle+(T+1)^{2}}{N}.
    \end{aligned}
\]

\end{proof}

\ifnum\shorter=1

\fi
\section{Database Size Increment per Query}
\label{sec:Appendix A}

We examine how
one family query acts on $\mathcal L_P^{(K)}$. In the following
calculation, $s$, $x$, and $u$ denote the family-query salt index, the
query input, and the Fourier-basis response value, respectively; all
remaining basis labels of $Z$ are suppressed.
For random functions, by linearity, it suffices to consider a basis component of $\mathcal L_P^{(K)}$. Fix $D_1,\ldots,D_K$ with $\sum_{t=1}^K|D_t|\le P$ and consider
\[
\begin{aligned}
&|\psi\rangle=\sum_s c_s\ket{s,x,u}_Z
\otimes
\ket{\phi_{D_1},\ldots,\phi_{D_K}}_O
\\
&=
\frac{1}{M^{NK/2}}
\sum_{s,\mathbf H}
c_s
\omega_M^{\sum_{t=1}^K\langle D_t,H_t\rangle}
\ket{s,x,u}_Z\ket{H_1,\ldots,H_K}_O
\\
&\xrightarrow{\mathcal O}
\frac{1}{M^{NK/2}}
\sum_{s,\mathbf H}
c_s
\omega_M^{\sum_{t=1}^K\langle D_t,H_t\rangle+uH_s(x)}
\ket{s,x,u}_Z\ket{H_1,\ldots,H_K}_O
\\
&=
\frac{1}{M^{NK/2}}
\sum_{s,\mathbf H}
c_s
\omega_M^{
\sum_{t\ne s}\langle D_t,H_t\rangle
+
\langle D_s+ue_x,H_s\rangle}
\ket{s,x,u}_Z\ket{H_1,\ldots,H_K}_O
\\
&=
\sum_s c_s\ket{s,x,u}_Z
\otimes
\ket{\phi_{D_1},\ldots,\phi_{D_s+ue_x},\ldots,\phi_{D_K}}_O.
\end{aligned}
\]
As $|D_s+ue_x|\le |D_s|+1$, one family query increases total database size by at most one.

For random permutations, by linearity, it suffices to fix
$s\in[K]$, $x\in[N]$, $u\in\mathbb Z_N$, and a tensor product of database states in $\mathcal L_P^{(K)}$.
Fix $\alpha_1,\ldots,\alpha_K$ satisfying
$\sum_{t=1}^K|\alpha_t|\le P$, and consider
\[\begin{aligned}
&|\psi\rangle = |s,x,u\rangle_Z
\otimes
|v_{\alpha_1},\ldots,v_{\alpha_K}\rangle_O
\\
&=
\frac{1}{
\sqrt{\prod_{t=1}^K (N-|\alpha_t|)!}
}
\sum_{\substack{
H_1\supseteq\alpha_1,\ldots,H_K\supseteq\alpha_K
}}
|s,x,u\rangle_Z
|H_1,\ldots,H_K\rangle_O
\\
&\xrightarrow{\mathcal O}
\frac{1}{
\sqrt{\prod_{t=1}^K (N-|\alpha_t|)!}
}
\sum_{\substack{
H_1\supseteq\alpha_1,\ldots,H_K\supseteq\alpha_K
}}
\omega_N^{uH_s(x)}
|s,x,u\rangle_Z
|H_1,\ldots,H_K\rangle_O .
\end{aligned}\]
If $x\in\operatorname{dom}(\alpha_s)$, then
\[
    \mathcal O|\psi\rangle
    =
    \omega_N^{u\alpha_s(x)}
    |s,x,u\rangle_Z
    \otimes
    |v_{\alpha_1},\ldots,v_{\alpha_K}\rangle_O.
\]
Hence, the total database size is unchanged.
Suppose $x\not\in \operatorname{dom}(\alpha_{s})$. Then,
\[\begin{aligned}
\sum_{H_s\supseteq\alpha_s}
\omega_N^{uH_s(x)}|H_s\rangle
&=
\sum_{\substack{
H_s\supseteq\alpha_s
}}
\sum_{y\notin\operatorname{im}(\alpha_s)}
\omega_N^{uy}\one[H_s(x)=y]
|H_s\rangle
\\
&=
\sum_{y\notin\operatorname{im}(\alpha_s)}
\omega_N^{uy}
\sum_{H_s\supseteq\alpha_s\cup\{x\mapsto y\}}
|H_s\rangle.
\end{aligned}\]
Note that $|\alpha_s\cup\{x\mapsto y\}|=|\alpha_s|+1$ and
\[
    |v_{\alpha_s\cup\{x\mapsto y\}}\rangle
=
\frac{1}{\sqrt{(N-|\alpha_s|-1)!}}
\sum_{H_s\supseteq\alpha_s\cup\{x\mapsto y\}}
|H_s\rangle.
\]
Hence, 
\[\sum_{H_s\supseteq\alpha_s}
\omega_N^{uH_s(x)}|H_s\rangle
=
\sqrt{(N-|\alpha_s|-1)!}
\sum_{y\notin\operatorname{im}(\alpha_s)}
\omega_N^{uy}
|v_{\alpha_s\cup\{x\mapsto y\}}\rangle.\]
For $t\ne s$,
\[\sum_{H_t\supseteq\alpha_t}|H_t\rangle
=
\sqrt{(N-|\alpha_t|)!}\,|v_{\alpha_t}\rangle.\]
The calculation above gives
\[
\begin{aligned}
    \mathcal O|\psi\rangle
    &=
    \frac{1}{\sqrt{N-|\alpha_s|}}
    \sum_{y\notin\operatorname{im}(\alpha_s)}
    \omega_N^{uy}
    |s,x,u\rangle_Z
    \\
    &\qquad\otimes
    |v_{\alpha_1},\ldots,
    v_{\alpha_s\cup\{x\mapsto y\}},
    \ldots,v_{\alpha_K}\rangle_O.
\end{aligned}
\]
Thus every component in the second case has total database size at most \(P+1\). By linearity, the two cases show that one query to \(\mathbf H\) raises the total database size by at most one.
\begin{lemma}
    \label{lem:generic-degree-increase}
    The salted purified game operator raises the total database size by at most \(r:=2(T_{\mathrm{Sample}}+T+T_{\mathrm{Verify}})\); namely, for every $P\ge0$,
    \[
        M_{\mathcal E^{(K)}}\mathcal L_P^{(K)}
        \subseteq
        \mathcal L_{P+r}^{(K)}.
    \]
\end{lemma}

\begin{proof}
    Let $U$ be the joint unitary implemented by the sampler, adversary, and verifier in the purified-oracle representation. It contains at most
    $T_{\Samp}+T+T_{\Ver}$ underlying oracle queries.
    By the one-query calculations above,
    \[
        U\mathcal L_P^{(K)}
        \subseteq
        \mathcal L_{
            P+T_{\Samp}+T+T_{\Ver}
        }^{(K)}.
    \]
    The same inclusion holds for $U^\dagger$, since each adjoint
    phase query is obtained by negating the response label before and
    after a phase query; these relabelings act only on the query registers
    and preserve the database size.

    For a search game, the purified game operator has the form
    \[
    (I\otimes\langle0|)U^\dagger\Pi U(I\otimes|0\rangle)
    \]
    and therefore increases the total database size by at most $2(T_{\mathsf{Sample}}+T+T_{\mathsf{Verify}})$. For a decision game, the purified game operator is a signed linear combination of two operators of this form. Since the target subspace in the database hierarchy is linear, the same inclusion holds.
\end{proof}
When $K=1$, the salted database hierarchy reduces to that of the base game,
$\mathcal L_P^{(1)}=\mathcal L_P$, and the salted query interface reduces
to the base-game query interface. Hence,
\cref{lem:generic-degree-increase} also gives
\[
    M_{\mathcal E}\mathcal L_P
    \subseteq
    \mathcal L_{
        P+2(T_{\Samp}+T+T_{\Ver})
    }.
\]
\end{document}